%% file: main.tex
\documentclass{article}[11pt]
\usepackage[utf8]{inputenc}
\usepackage{style}
\usepackage{framed}
\usepackage[parfill]{parskip}
\usepackage{setspace}
\usepackage{hyperref}

\allowdisplaybreaks

\title{An Improved Integrality Gap for the C{\u{a}}linescu-Karloff-Rabani Relaxation for Multiway Cut}

\author{
Haris Angelidakis\thanks{Email: \texttt{hangel@ttic.edu}.} \vspace{-0.5em}\\
TTIC
\and
Yury Makarychev\thanks{Email: \texttt{yury@ttic.edu}. Supported by NSF awards CAREER CCF-1150062 and IIS-1302662.} \vspace{-0.5em}\\
TTIC
\and
Pasin Manurangsi\thanks{Email: \texttt{pasin@berkeley.edu}. Supported by NSF Grants No. CCF-1540685 and CCF-1655215. Part of this work was done while the author was visiting TTIC.} \vspace{-0.5em}\\
UC Berkeley
}

\begin{document}

\maketitle

\begin{abstract}
We construct an improved integrality gap instance for the C{\u{a}}linescu-Karloff-Rabani LP relaxation of the Multiway Cut problem. In particular, for $k \geqs 3$ terminals, our instance has an integrality ratio of $6 / (5 + \frac{1}{k - 1}) - \varepsilon$, for every constant $\varepsilon > 0$. For every $k \geqs 4$, our result improves upon a long-standing lower bound of $8 / (7 + \frac{1}{k - 1})$ by Freund and Karloff~\cite{FK00}. Due to Manokaran \etal's result~\cite{MNRS08}, our integrality gap also implies Unique Games hardness of approximating Multiway Cut of the same ratio. 
\end{abstract}

\section{Introduction}

In the Multiway Cut problem, we are given a weighted undirected graph and a set of $k$ terminal vertices, and  we are asked to find a set of edges with minimum total weight whose removal disconnects all the terminals. Equivalently, the objective can also be stated as partitioning the graph into $k$ clusters, each containing exactly one terminal, such that the total weight of edges across clusters is minimized. Such a partition is sometimes called a \emph{$k$-way cut} and the total weight of edges across clusters is called the \emph{cost} of the cut.

Since its introduction in the early 1980s by Dahlhaus, Johnson, Papadimitriou, Seymour, and Yannakakis~\cite{DJPSY83}, the problem has been extensively studied in the approximation algorithms community~\cite{C89,CR91,DJPSY94,BTV99,CT99,CKR00,FK00,KKSTY04,BNS13,SV14} and has, by now, become one of the classic graph partitioning problems taught in many graduate approximation algorithms courses. Despite this, the approximability of the problem when $k \geqs 4$ still remains open.

When $k = 2$, Multiway Cut is simply Minimum $s$-$t$ Cut, which is solvable in polynomial time. For $k \geqs 3$, Dahlhaus \etal~\cite{DJPSY94} showed that the problem becomes APX-hard and gave the first approximation algorithm for the problem, which achieves an approximation ratio of $(2 - 2/k)$. Due to the combinatorial nature of the algorithm, many linear programs were proposed following their work (see e.g.~\cite{C89,CR91,BTV99}). However, it was not until C{\u{a}}linescu, Karloff and Rabani's work~\cite{CKR00} that a significant improvement in the approximation ratio was made. Their linear programming relaxation, commonly referred to as the CKR relaxation, on a graph $G = (V, E, w)$ and terminals $\{t_1, \dots, t_k \} \subseteq V$, can be formulated as
\begin{align*}
\text{minimize } &\sum_{e = (u, v) \in E} w(e) \cdot \frac{1}{2} \|x^u - x^v\|_1 \\ 
\text{subject to } 
&\forall u \in V, \; x^u \in \Delta_k, \\
&\forall i \in [k],\; x^{t_i} = e^i,
\end{align*}
where $[k] = \{1, \dots, k\}$, $\Delta_k = \{(x_1, \dots, x_k) \in [0, 1]^k \mid x_1 + \cdots + x_k = 1\}$ denotes the $k$-simplex and $e^i$ is the $i$-th vertex of the simplex, i.e. $e^i_i = 1$. C{\u{a}}linescu, Karloff and Rabani gave a rounding scheme for this LP that yields a $(3/2 - 1/k)$-approximation algorithm for Multiway Cut~\cite{CKR00}. Exploiting the geometric nature of the relaxation even further, Karger, Klein, Stein, Thorup and Young~\cite{KKSTY04} gave improved rounding schemes for the relaxation; for general $k$, they gave a $1.3438$-approximation algorithm for the problem. They also conducted computational experiments that led to improvements over small $k$'s. For $k = 3$, in particular, they gave a $12/11$-approximation algorithm and proved that this is tight by constructing an integrality gap example of ratio $12/11 - \varepsilon$, for every $\varepsilon > 0$. The same result was also independently discovered by Cunningham and Tang~\cite{CT99}. More recently, Buchbinder, Naor and Schwartz~\cite{BNS13} gave a neat $4/3$-approximation algorithm to the problem for general $k$ and additionally showed how to push the ratio down to $1.3239$. Their algorithm was later improved by Sharma and Vondr\'{a}k~\cite{SV14} to get an approximation ratio of $1.2965$. This remains the state-of-the-art approximation for sufficiently large $k$. We remark that the Sharma-Vondr\'{a}k algorithm is quite complicated and requires a computer-assisted proof. To circumvent this, Buchbinder, Schwartz and Weizman~\cite{BSW17} recently came up with a simplified algorithm and analytically showed that it yielded roughly the same approximation ratio as Sharma and Vondr\'{a}k's.

The CKR relaxation not only leads to improved algorithms for Multiway Cut but also has a remarkable consequence on the approximability of the problem. In a work by Manokaran, Naor, Raghavendra and Schwartz~\cite{MNRS08}, it was shown, assuming the Unique Games Conjecture (UGC), that, if there exists an instance of Multiway Cut with integrality gap $\tau$ for the CKR relaxation, then it is NP-hard to approximate Multiway Cut to within a factor of $\tau - \varepsilon$ of the optimum, for every constant $\varepsilon > 0$. Note here that the integrality gap of an instance of Multiway Cut is the ratio between the integral optimum and the value of the CKR relaxation of the instance.

Roughly speaking, Manokaran \etal's result means that, if one believes in the UGC, the CKR relaxation achieves essentially the best approximation ratio one can hope to get in polynomial time for Multiway Cut. Despite this strong connection, few lower bounds for the CKR relaxation are known. Apart from the aforementioned $12/11 - \varepsilon$ integrality gap for $k = 3$ by Karger \etal~\cite{KKSTY04} and Cunningham and Tang~\cite{CT99}, the only other known lower bound is an $8/(7 + \frac{1}{k - 1})$-integrality gap which was constructed by Freund and Karloff~\cite{FK00} not long after the introduction of the CKR relaxation. Hence, it is a natural and intriguing task to construct improved integrality gap examples for the relaxation in an attempt to bridge the gap between approximation algorithms and hardness of approximation for the Multiway Cut.

\subsection{Our Contributions}

In this paper, we provide a new construction of integrality gap instances for Multiway Cut. Our construction achieves an integrality gap of $6 / (5 + \frac{1}{k - 1}) - \varepsilon$, for every $k \geqs 3$, as stated formally below. For every $k \geqs 4$, our integrality gap improves on the integrality gap of $8 / (7 + \frac{1}{k - 1})$ of Freund and Karloff~\cite{FK00}. When $k = 3$, our gap matches the known results of Karger et al.~\cite{KKSTY04} and Cunningham and Tang~\cite{CT99}. 

\begin{theorem} \label{thm:main}
For every $k \geqs 3$ and every $\varepsilon > 0$, there exists an instance $\cI_{k, \varepsilon}$ of Multiway Cut with $k$ terminals such that the integrality gap of the CKR relaxation for $\cI_{k, \varepsilon}$ is at least $6 / (5 + \frac{1}{k - 1}) - \varepsilon$.
\end{theorem}

As mentioned earlier, Manokaran \etals showed that integrality gaps for the CKR relaxation can be translated directly to a Unique Games hardness of approximation~\cite{MNRS08}. Hence, the following corollary holds as an immediate consequence of our integrality gap.

\begin{corollary}
Assuming the Unique Games Conjecture, for every $k \geqs 3$ and every $\varepsilon > 0$, it is NP-hard to approximate Multiway Cut with $k$ terminals to within a factor of $6 / (5 + \frac{1}{k - 1}) - \varepsilon$ of the optimum.
\end{corollary}

\subsubsection{Techniques}

To see the motivation behind our construction, it is best to first gain additional intuition for the geometry of the CKR relaxation. For brevity, we will not be thoroughly formal in this section. 

\paragraph{Geometric Interpretation of the CKR Relaxation.} A solution of the CKR relaxation embeds the graph into a simplex; each vertex $u$ becomes a point $x^u \in \Delta_k$, while each edge $(u, v)$ becomes a segment $(x^u, x^v)$. As a result, to construct a randomized rounding scheme for the relaxation, it is sufficient to define a randomized $k$-partitioning scheme of the simplex. More specifically, this is a distribution $\cP$ on $k$-way cuts of $\Delta_k$ (i.e. $\forall P \in \cP$, $P: \Delta_k \rightarrow [k]$ such that $P(e^i) = i$ for every $i \in [k]$) and its performance is measured by its \emph{maximum density}, which is defined as $$\tau_k(\cP) := \sup_{x \ne y \in \Delta_k} \frac{\Pr_{P \sim \cP}[P(x) \ne P(y)]}{\frac{1}{2}\|x - y\|_1}.$$ In words, for any two different points $x, y \in \Delta_k$, a random $k$-way cut $P$ sampled from $\cP$ assigns the two points to different clusters with probability at most $\tau_k(\cP) \cdot \frac{1}{2} \|x - y\|_1$. This immediately yields the following randomized rounding scheme for the CKR relaxation: pick a random $P \sim \cP$ and place each $u \in V$ in cluster $P(x^u)$. From the definition of maximum density, the expected contribution of each edge $e = (u, v)$ in the rounded solution is at most $\tau_k(\cP) \cdot w(e) \cdot \frac{1}{2} \|x^u - x^v\|_1$, meaning that this rounding algorithm indeed gives a $\tau_k(\cP)$-approximation for Multiway Cut with $k$ terminals.

The above observation was implicit in~\cite{CKR00} and was first made explicit by Karger \etal~\cite{KKSTY04}, who also proved the opposite side of the inequality: for every $\varepsilon > 0$, there is an instance $\cI$ of Multiway Cut with $k$ terminals whose integrality gap is at least $\tau_k^* - \varepsilon$, where $\tau_k^*$ is the minimum\footnote{In~\cite{KKSTY04}, $\tau_k^*$ is defined as the infimum of $\tau_k(\cP)$ among all $\cP$'s but it was proved in the same work that there exists $\cP$ that achieves the infimum.} of $\tau_k(\cP)$ among all $\cP$'s. In other words, constructing an integrality gap for Multiway Cut is equivalent to proving a lower bound for $\tau_k^*$.

\paragraph{Non-Opposite Cuts.} Let us now consider a special class of partitions of the simplex, which we call \emph{non-opposite cuts}. A non-opposite cut of $\Delta_k$ is a function $P: \Delta_k \rightarrow [k + 1]$ such that $P(e^i) = i$ for every $i \in [k]$ and, for every $x \in \Delta_k$, $P(x) \in \supp(x) \cup \{k + 1\}$, where $\supp(x) := \{i \in [k] \mid x_i \ne 0\}$ is the set of all non-zero coordinates of $x$. Note here that $P$ now partitions $\Delta_k$ into $k + 1$ parts and the additional condition requires that a point that lies on any face of the simplex is either assigned to the $(k + 1)$-th cluster or to a cluster corresponding to a vertex of the simplex that belongs to this face. The representative case to keep in mind is when $k = 3$; in this case, a non-opposite cut partitions the triangle $\Delta_3$ into four parts and any point on the border of $\Delta_3$ is not assigned to the cluster corresponding to the opposite vertex of the triangle. Figure~\ref{fig:non-opposite} demonstrates the difference between non-opposite cuts and normal $3$-way cuts of $\Delta_3$.

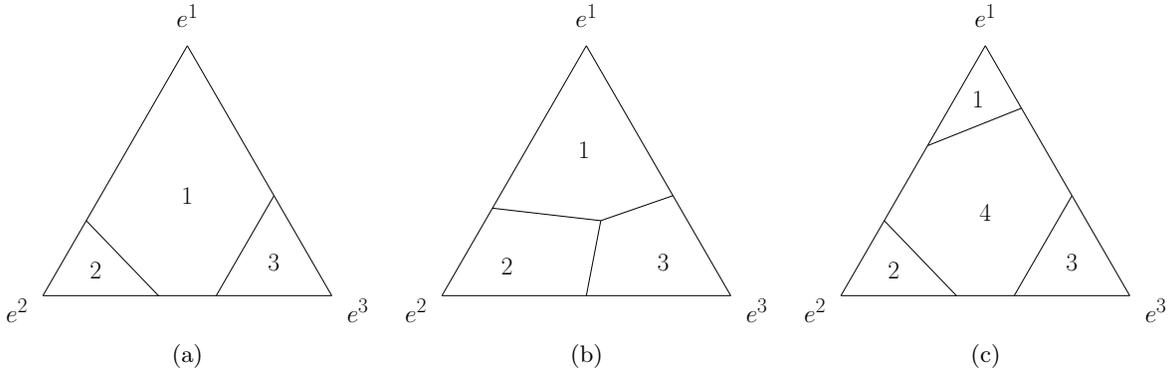
\begin{figure}[h!]
    \centering
    \begin{subfigure}[b]{0.3\textwidth}
        \resizebox{\textwidth}{!}{\input{oppositecut.tikz}}
        \caption{}
        \label{fig:op}
    \end{subfigure}
    ~ 
    \begin{subfigure}[b]{0.3\textwidth}
    	\resizebox{\textwidth}{!}{\input{nonop-ball.tikz}}
    	\caption{}
    	\label{fig:non-op-ball}
    \end{subfigure}
    ~
    \begin{subfigure}[b]{0.3\textwidth}
    	\resizebox{\textwidth}{!}{\input{nonop-corner.tikz}}
    	\caption{}
    	\label{fig:non-op-corner}
    \end{subfigure}
    \caption{Illustrations of cuts of $\Delta_3$. The 3-way cut in Figure~\ref{fig:op} is not non-opposite since some points on the line between $e^2$ and $e^3$ belong to the cluster corresponding to $e^1$. On the other hand, the 3-way cut in Figure~\ref{fig:non-op-ball} is non-opposite. Finally, Figure~\ref{fig:non-op-corner} demonstrates an example of a non-opposite cut that is not a 3-way cut because the fourth cluster is not empty.}
    \label{fig:non-opposite}
\end{figure}

We again define the maximum density $\tau_k$ for a distribution on non-opposite cuts similarly. Additionally, let $\tilde \tau^*_k$ be the infimum of $\tau_k(\cP)$ among all distributions $\cP$ on non-opposite cuts of $\Delta_k$. Observe first that $\tau^*_k \leqs \ttau^*_k$ because each non-opposite cut $P$ of $\Delta_k$ can be turned into a $k$-way cut without separating additional pairs of points by merging the $(k + 1)$-th cluster into the first cluster. This means that it will be easier for us to show a lower bound for $\tilde \tau^*_k$. Unfortunately, it is so far unclear whether such lower bound implies any lower bound for $\tau^*_k$. However, we will see in a moment that we can define an approximate version of $\tilde \tau^*_k$ that will give us a lower bound of $\tau^*_K$ for all $K > k$ as well. 

The approximation of $\tilde \tau^*_k$ we use is the one obtained by the discretization of the problem in which we restrict ourselves to the triangular grid $\Delta_{k, n} := \{x \in \Delta_k \mid \forall i \in [k],\; x_i \text{ is a multiple of } 1/n\}$ instead of $\Delta_k$. For every distribution $\cP$ on non-opposite cuts of $\Delta_k$, we can similarly define $\tau_{k, n}(P)$ as $$\tau_{k, n}(\cP) := \max_{x \ne y \in \Delta_{k, n}} \frac{\Pr_{P \sim \cP}[P(x) \ne P(y)]}{\frac{1}{2}\|x - y\|_1}$$ and $\tilde \tau^*_{k, n}$ as the minimum\footnote{The minimum exists since there are only finite number of cuts when restricted to the discretization $\Delta_{k, n}$ of the simplex.} of $\tau_{k, n}(\cP)$ among all such distributions $\cP$'s on non-opposite cuts. The simple observation that allows us to construct our integrality gap is the following relation between $\tau^*_K$ and $\tilde \tau^*_{k, n}$: for every $K > k$, $\tilde \tau^*_{k, n} - O(kn/(K - k)) \leqs \tau^*_K$. In fact, we only use this inequality when $k = 3$ to construct our gap and, hence, we only sketch the proof of this case here. For completeness, we prove the general case of the inequality in Appendix~\ref{app:tk}.

Before we sketch the proof of this inequality, let us note that Freund and Karloff's result~\cite{FK00} can be seen as a proof of $\tilde \tau_{3, 2}^* \geqs 8/7$; please refer to Appendix~\ref{app:FK} for a more detailed explanation. Note that $\tilde \tau_{3, 2}^* \geqs 8/7$, together with the above fact, immediately implies a lower bound of $\tau^*_K \geqs 8/7 - O(1/K)$. Barring the dependency on $K$, this is exactly the gap proven in~\cite{FK00}. With this interpretation in mind, the rest of our work can be mostly seen as proving that $\tilde \tau_{3,n}^* \geqs 6/5 - O(1/n)$. By selecting $n = \Theta(\sqrt{K})$, this implies a lower bound of $6/5 - O(1/\sqrt{K})$ for $\tau^*_K$; more care can then be taken to get the right dependency on $K$.

We now sketch the proof of $\tilde \tau^*_{3, n} - O(n/K) \leqs \tau^*_K$. Suppose that $\cP$ is the distribution on $K$-way cuts such that $\tau_K(\cP) = \tau^*_K$. We sample a non-opposite cut $\tilde P$ of $\Delta_{3, n}$ as follows. First, sample $P \sim \cP$. Then, randomly select three different indices $i_1, i_2, i_3$ from $[K]$. Let $P_{\{i_1, i_2, i_3\}}: \Delta_3 \rightarrow [4]$ be the cut induced by $P$ on the face induced by the vertices $i_1, i_2, i_3$. Specifically, let $f(i_j) = j$ for every $j \in [3]$. We define $P_{\{i_1, i_2, i_3\}}$ by
\begin{align*}
P_{\{i_1, i_2, i_3\}}(x) =
\begin{cases}
f(P(x_1e^{i_1} + x_2e^{i_2} + x_3e^{i_3})) & \text{ if } P(x_1e^{i_1} + x_2e^{i_2} + x_3e^{i_3}) \in \{i_1, i_2, i_3\}, \\
4 & \text{ otherwise.}
\end{cases}
\end{align*}

If $P_{\{i_1, i_2, i_3\}}$ is a non-opposite cut, then let $\tilde{P} = P_{\{i_1, i_2, i_3\}}$. Otherwise, let $\tilde P(x) = P_{\{i_1, i_2, i_3\}}(x)$ for every $x$ such that $P_{\{i_1, i_2, i_3\}}(x) \in \supp(x) \cup \{4\}$. For the other $x$'s, we simply set $\tilde P(x) = 4$; we will call these $x$'s ``bad''. It is obvious that $\tilde{P}$ is a non-opposite cut. Moreover, if any two points $x, y \in \Delta_{3, n}$ are separated in $\tilde P$, then either $P_{\{i_1, i_2, i_3\}}(x) \ne P_{\{i_1, i_2, i_3\}}(y)$ or exactly one of $x, y$ is bad. From the definition of $P_{\{i_1, i_2, i_3\}}$, the former implies $P(x_1e^{i_1} + x_2e^{i_2} + x_3e^{i_3}) \ne P(y_1e^{i_1} + y_2e^{i_2} + y_3e^{i_3})$, which happens with probability at most $\tau^*_K \cdot \frac{1}{2} \|x - y\|_1$. Hence, we are left to show that the probability that any point $x \in \Delta_{3, n}$ is bad is at most $O(1/K)$; this immediately yields the intended bound since $\|x - y\|_1 \geqs 2/n$ for every $x \ne y \in \Delta_{3, n}$.

Recall that a point $x$ is bad iff $P_{\{i_1, i_2, i_3\}}(x) \in [3] \setminus \supp(x)$, which can only happen when $|\supp(x)| = 2$. Assume without loss of generality that $\supp(x) = \{1, 2\}$, i.e. $x_3 = 0$. We want to bound the probability that $P_{\{i_1, i_2, i_3\}}(x) = 3$. Fix the cut $P$ and the choice of $i_1, i_2$. Since $x_3 = 0$, this already determines $P(x_1e^{i_1} + x_2e^{i_2} + x_3e^{i_3})$. From the definition of $P_{\{i_1, i_2, i_3\}}$, $P_{\{i_1, i_2, i_3\}}(x) = 3$ if and only if $i_3 = P(x_1e^{i_1} + x_2e^{i_2} + x_3e^{i_3})$. Since $i_3$ is a random element of $[K] \setminus \{i_1, i_2\}$, this happens with probability at most $1 / (K - 2)$ and we have completed our proof sketch.

\paragraph{Lower Bound for $\tilde \tau_{3, n}^*$.} The final component in our construction is a lower bound for $\ttau_{3, n}^*$. To lower bound $\ttau_{3, n}^*$, it is enough to construct a weighted undirected graph with vertices $\Delta_{3, n}$ such that the CKR relaxation value (with respect to the solution in which each vertex is assigned the vector representation of itself) is small and every non-opposite cut has large cost. In other words, we need to construct an integrality gap for Multiway Cut with 3 terminals but only against non-opposite cuts, instead of against all 3-way cuts.

We will set the edge set $E_{3, n}$ to be all two consecutive points on grid lines, i.e., $E_{3, n} = \{(x, y) \mid x,y \in \Delta_{3, n}, \|x - y\|_1 = 2/n\}$. Fortunately for us, Karger \etal~\cite{KKSTY04} gave an elegant characterization of 3-way cuts of the graph $(\Delta_{3, n}, E_{3, n})$, which ultimately led to their $12/11 - \varepsilon$ ratio integrality gap for Multiway Cut with 3 terminals. By extending this, we arrive at a similar characterization for all non-opposite cuts of the graph, which also eventually enables us to construct the integrality gap against such cuts.

\subsection{Organization of the paper}

The rest of the paper is organized as follows. We start by formally defining our notation in Section~\ref{sec:prelim}. In Section~\ref{sec:3}, we prove the lower bound of $6/5 - O(1/n)$ on $\ttau_{3, n}$. In Section~\ref{sec:rootk}, we use this to construct an integrality gap of ratio $6/5 - O(1/\sqrt{k})$, which essentially contains all of our main ideas. Then, in Section~\ref{sec:k}, we show how to nail down the dependency on $k$ and achieve the integrality ratio of $6 / (5 + \frac{1}{k - 1}) - \varepsilon$, completing the proof of our main result. Finally, we provide some discussion regarding our result and open questions in Section~\ref{sec:open}.

\section{Preliminaries and Notation} \label{sec:prelim}

From now on, our graphs will have vertex sets being $\Delta_{k,n} \subseteq \Delta_k$ for some $k, n$, and the edge sets being $E_{k, n} := \{(x, y) \mid  x,y \in \Delta_{k,n}, \|x - y\|_1 = 2/n\}$. The terminals are naturally associated with the $k$ vertices of the simplex. To avoid confusion, we reserve the word \emph{vertices} for the simplex vertices and we instead refer to the vertices of such graphs as \emph{points}.

We refer to the LP solution in which each point is assigned to itself as the \emph{canonical LP solution} and we use $\lpc(w)$ to denote its value with respect to a weight function $w: E_{k, n} \rightarrow \Rz$, i.e., $$\lpc(w) = \sum_{(x, y) \in E_{k, n}} w(x, y) \cdot \frac{1}{2} \|x - y\|_1 = \frac{1}{n} \sum_{(x, y) \in E_{k, n}} w(x, y),$$ where the latter equality comes from the fact that $\|x - y\|_1 = 2/n$ for all $(x, y) \in E_{k, n}$. Moreover, for every $k$-way cut $P$, we denote its cost with respect to $w$ by $\cost(P, w)$; in other words, $$\cost(P, w) = \sum_{(x, y) \in E_{k, n}} w(x, y) \cdot \ind[P(x) \ne P(y)],$$ where $\ind[P(x) \ne P(y)]$ is the indicator variable of the event $P(x) \ne P(y)$.

\section{A Lower Bound on $\ttau_{3, n}^*$} \label{sec:3}

In this section, we will construct an integrality gap on the graph $(\Delta_{3, n}, E_{3, n})$ such that its CKR LP value is small but every non-opposite cut has large cost. As our graph is already defined, we only need to define the weights of the edges. The main result of this section can be stated as follows.

\begin{lemma} \label{lem:3sim}
For every $n$ divisible by 3, there exists $w: E_{3, n} \rightarrow \Rz$ such that
\begin{itemize}
\item The value of the canonical LP solution of $(\Delta_{3, n}, E_{3, n}, w)$ is at most $5/6 + O(1/n)$.
\item The cost of every non-opposite cut of $(\Delta_{3, n}, E_{3, n}, w)$ is at least one.
\end{itemize}
\end{lemma}

To prove Lemma~\ref{lem:3sim}, we will first characterize candidate optimal non-opposite cuts.

\subsection{A Characterization of Non-Opposite Cuts of $\Delta_{3, n}$}

To characterize the 3-way cuts, Karger \etal~\cite{KKSTY04} consider the augmented version of $(\Delta_{3, n}, E_{3, n})$, in which each vertex of the simplex has an edge heading out infinitely. They then look at the dual graph of this graph. The augmentation creates three outer faces; the vertices in the dual graph corresponding to these faces are named $O_1, O_2$ and $O_3$ where $O_1, O_2$ and $O_3$ are opposite to $e^1, e^2$ and $e^3$ respectively. For convenience, we disregard the edges among $O_1, O_2, O_3$ in the dual. Figure~\ref{fig:primal-dual} contains illustrations of the augmented graph and its planar dual. The figure is a recreation of Figure 1 and Figure 2 in~\cite{KKSTY04}.

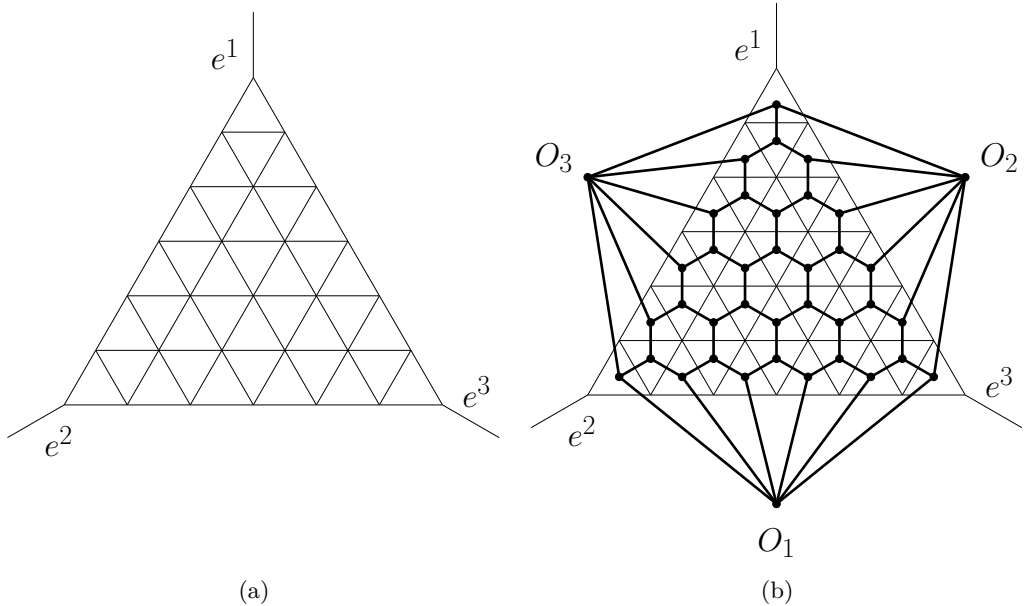
\begin{figure}[h!]
    \centering
    \begin{subfigure}[b]{0.4\textwidth}
        \resizebox{\textwidth}{!}{\input{primal.tikz}}
        \caption{}
        \label{fig:primal}
    \end{subfigure}
    ~ 
    \begin{subfigure}[b]{0.4\textwidth}
    	\resizebox{\textwidth}{!}{\input{dual.tikz}}
    	\caption{}
    	\label{fig:dual}
    \end{subfigure}
    \caption{The augmented $(\Delta_{3, 6}, E_{3, 6})$ and its dual.}
    \label{fig:primal-dual}
\end{figure}

A cut in the original graph can now be viewed as a collection of edges in the dual graph; an edge in the dual graph between two faces corresponding to the shared edge of the faces being cut. With this interpretation, Karger \etal give the following characterization for candidate optimal cuts of any weight function.

\begin{observation}[\cite{KKSTY04}] \label{obs:char-3way}
For any $w: E_{3, n} \rightarrow \Rz$, there exists a least-cost\footnote{Note that it is possible that there are cuts not of these forms that also achieve the same (optimal) cost.} 3-way cut $P$ that is of one of the following forms:
\begin{itemize}
\item $P$ contains three non-intersecting paths from a triangle to $O_1, O_2$ and $O_3$. Such $P$ is called a \emph{ball cut}.
\item $P$ contains two non-intersecting paths among $O_1, O_2$ and $O_3$. Such $P$ is called a \emph{2-corner cut}.
\end{itemize}
\end{observation}

Examples of a ball cut and a 2-corner cut can be found in Figure~\ref{fig:ballcut} and Figure~\ref{fig:2cornercut} respectively. These figures are reproduced from Figure 3 and Figure 4 of~\cite{KKSTY04} respectively. 

\begin{figure}[h!]
    \centering
    \begin{subfigure}[b]{0.3\textwidth}
        \resizebox{\textwidth}{!}{\input{ball.tikz}}
        \caption{A ball cut}
        \label{fig:ballcut}
    \end{subfigure}
    ~ 
    \begin{subfigure}[b]{0.3\textwidth}
    	\resizebox{\textwidth}{!}{\input{2-corner.tikz}}
    	\caption{A 2-corner cut}
    	\label{fig:2cornercut}
    \end{subfigure}
    ~
    \begin{subfigure}[b]{0.3\textwidth}
    	\resizebox{\textwidth}{!}{\input{3-corner.tikz}}
    	\caption{A 3-corner cut}
    	\label{fig:3cornercut}
    \end{subfigure}
    \caption{Examples of a ball cut, a 2-corner cut and a 3-corner cut of $(\Delta_{3, 6}, E_{3, 6})$.}
    \label{fig:cuts}
\end{figure}
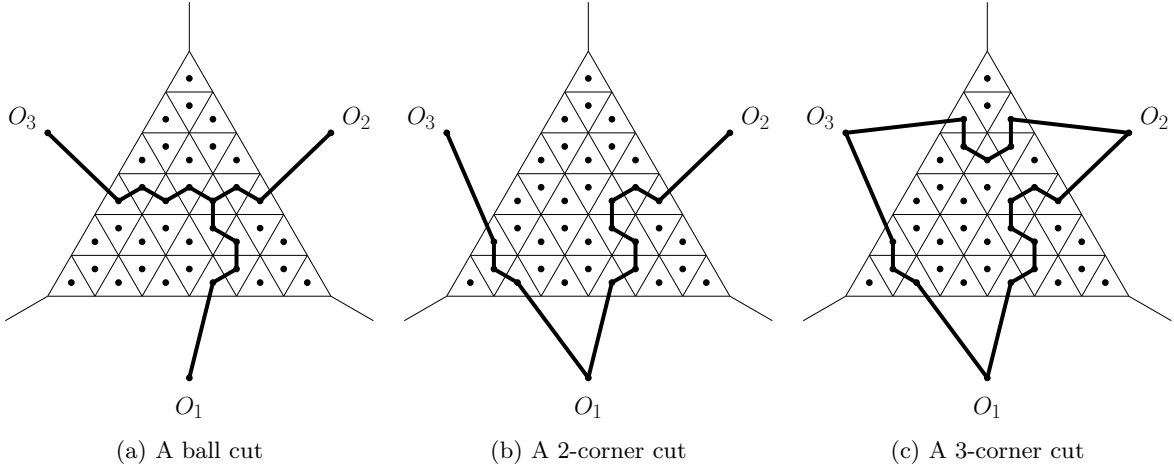

The observation implies that, if one wants to prove that every cut incurs large cost with respect to $w$, then it is enough to show that all ball cuts and 2-corner cuts have large cost according to $w$. This is much more convenient to work with and it turns out that finding such $w$ can be formulated as a linear program of size polynomial in $n$. Once this observation was made, Karger \etal then solved the LP for large $n$ and used the solution as an inspiration to come up with the $12/11 - \varepsilon$ gap for Multiway Cut with three terminals. 

We make the following analogous observation for non-opposite cuts. This observation allows us to modify Karger \etal's LP to express the problem of finding an integrality gap against non-opposite cuts on $(\Delta_{3, n}, E_{3, n})$. We then solve the LP and use it as an inspiration for our explicit gap construction. Since the LP will not appear in our proof, we refer interested readers to Section 3.1 of~\cite{KKSTY04} for more details on the LP.

\begin{observation} \label{obs:char-nonop}
For any $w: E_{3, n} \rightarrow \Rz$, there exists a least-cost non-opposite cut $P$ that is of one of the following forms:
\begin{itemize}
\item $P$ is a ball cut.
\item $P$ contains three non-intersecting paths among $O_1, O_2$ and $O_3$. Such $P$ is called a \emph{3-corner cut}.
\end{itemize}
\end{observation}

An example of a 3-corner cut can be found in Figure~\ref{fig:3cornercut}. For brevity, we will sometimes abbreviate  3-corner cuts as simply \emph{corner cuts}. Note that we reserve this abbreviation only for 3-corner cuts and we will never shorten 2-corner cuts. The proof of our observation is quite simple and is given below.

\begin{proofof}[Observation~\ref{obs:char-nonop}]
Let $P$ be any non-opposite cut of $(\Delta_{3, n}, E_{3, n})$ (not necessary a ball cut or a corner cut). To prove Observation~\ref{obs:char-nonop}, we will provide a sequence of transformations of $P$ such that all uncut edges in $P$ remain uncut and that, at the end of the transformation, $P$ becomes either a ball cut or a corner cut.

For convenience, let $G_P$ be the graph $(\Delta_{3, n}, E_{3, n})$ after all the edges cut by $P$ are removed. Observe that every point in a connected component of $G_P$ is assigned by $P$ to the same cluster. The transformations can be described as follows:
\begin{itemize}
\item For each connected component $C$ of $G_P$ that is assigned to 4, if $C$ intersects with all the three sides of the triangle, then do nothing. Otherwise, reassign the whole component to either $1, 2$ or $3$, whichever one still maintains the non-oppositeness of $P$; note that since $C$ does not intersect with all three sides of the triangle, this is always possible.
\item As long as there exists a component $C$ of $G_P$ that is assigned to $i \in [3]$ but does not contain $e^i$, reassign $C$ to the same cluster as one of its neighbors. Moreover, it is again easy to check that we can always pick the neighboring cluster such that $P$ remains non-opposite.
\end{itemize}

Clearly, any edge that was initially uncut remains uncut in the final $P$. Moreover, since the second transformation reduces the number of connected components of $G_P$ by at least one, the sequence of transformations always halts. Finally, note that, in the resulting $P$, all points assigned to any $i \in [4]$ form a connected component. Moreover, if any point is assigned to 4, then the corresponding connected component must contain at least one point from each of the three sides. It is now easy to see that, if no point is assigned to 4, then the final cut is a ball cut. On the other hand, if at least one point is assigned to 4, the final cut is a corner cut.
\end{proofof}

\subsection{The Integrality Gap} \label{subsec:int-nonop}

With the characterization of non-opposite cuts in place, we are now ready to describe our integrality gap and prove Lemma~\ref{lem:3sim}. We remark here that, although we have so far focused only on Karger \etal's result, our construction will be more similar to that of Cunningham and Tang~\cite{CT99}, in that we will partition the simplex into ``corner triangles'' and ``middle hexagon''. We would like to stress however that, even with these similarities, these two gaps are significantly different and nothing carries over from there as a blackbox.

\begin{proofof}[Lemma~\ref{lem:3sim}]
The gap can be described as follows. First, we divide the vertex set into \emph{corner triangles} $T_1, T_2, T_3$ and a \emph{middle hexagon} $H$ as follows. The middle hexagon contains all points $x \in \Delta_{3, n}$ such that all the coordinates of $x$ do not exceed 2/3, whereas the corner triangle $T_i$ contains all the points $x$'s such that $x_i \geqs 2/3$. Note that this is not a partition since the middle hexagon $H$ and each corner triangle $T_i$ still share the line $x_i = 2/3$, but this notation will be more convenient for us.

Every edge in the middle hexagon including its border is assigned weight $\rho := 1/(2n)$. Moreover, for each $i$, every non-border edge in $T_i$ that is not parallel to the opposite side of $e^i$ is also assigned weight $\rho$, whereas the non-border edges parallel to the opposite side of $e^i$ are assigned weight zero. Finally, for each of the two borders of $T_i$ containing $e^i$, we define its weight as follows. The edge closest to $e^i$ is assigned weight $(n/3)\rho$, the second closest is assigned $(n/3 - 1)\rho$, and so on. An illustration of the construction is shown in Figure~\ref{fig:gap}.

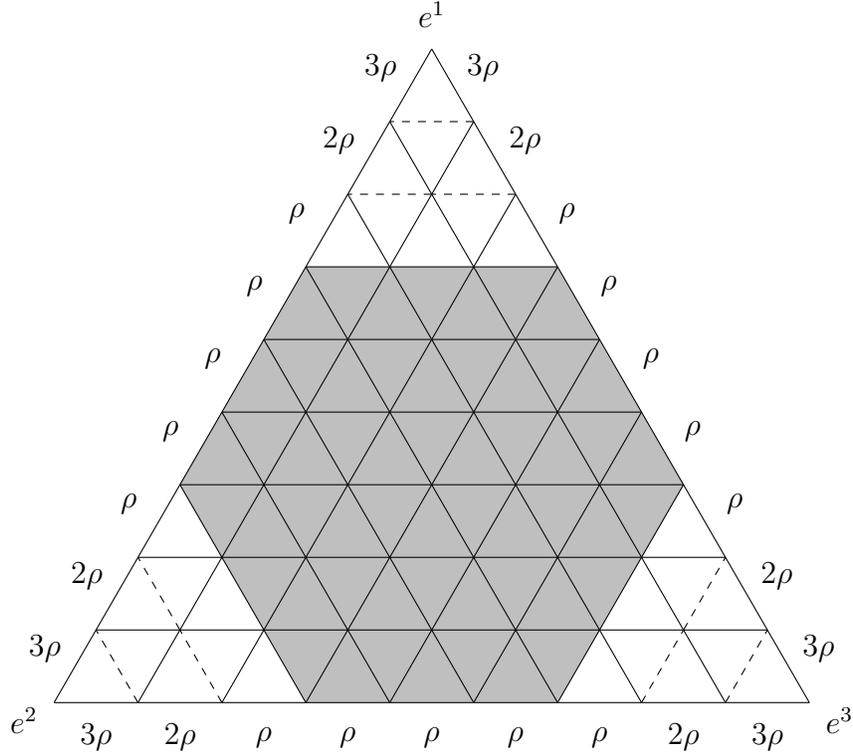
\begin{figure}[h!]
    \centering
    \resizebox{0.7\textwidth}{!}{\input{gap.tikz}}
   	\caption{The constructed integrality gap when $n = 9$. The middle hexagon $H$ is the shaded area. The weights of the edges on the borders of the triangle are shown. For the non-border edges, the dashed edges have weight zero whereas the rest of the edges all have weight $\rho$.}
    \label{fig:gap}
\end{figure}

Let us now compute the value of the canonical LP solution with respect to $w$. Recall that the value $\lpc(w)$ is simply $\frac{1}{n} \sum_{(x, y) \in E_{3, n}} w(x, y)$. Note that due to symmetry the sum of all the weights is three times the sum of the weights of all $(x, y) \in E_{3, n}$ that are parallel to $(e^1, e^2)$. The total weight on the line $(e^1, e^2)$ is 
\begin{align*}
\left((n/3) \rho + (n/3 - 1) \rho + \cdots + \rho \right) + (\rho + \cdots + \rho) + \left(\rho + \cdots + (n/3)\rho \right) = (n^2/9 + 2n/3) \rho
\end{align*}
and the total weight on the non-border edges parallel to $(e^1, e^2)$ is simply
\begin{align*}
(n - 1)\rho + \cdots + (n/3) \rho = (4n^2/9 - n/3) \rho.
\end{align*}
Hence, the total weight of all edges parallel to $(e^1, e^2)$ is $(5n^2/9 + n/3) \rho$.
As a result, we have
\begin{align*}
\lpc(w) = \frac{1}{n} \cdot 3 \cdot \left(\frac{5n^2}{9} + \frac{n}{3}\right) \cdot \frac{1}{2n} = \frac{5}{6} + O(1/n),
\end{align*}
as desired.

We will next prove that, for any non-opposite cut $P$, $\cost(P, w) \geqs 1$. From Observation~\ref{obs:char-nonop}, we can assume that $P$ is either a ball cut or a corner cut. We will prove that $\cost(P, w) \geqs 1$ by a potential function argument. Recall that each node in the dual graph is either $O_1, O_2, O_3$ or a triangle. We represent each triangle by its middle point (median). For each $i$, we define the potential function $\Phi_i$ on $O_i$ and all the triangles as follows.
\begin{align*}
\Phi_i(F) =
\begin{cases}
0 & \text{ if } F = O_i, \\
(4n/3)\rho & \text{ if } x \in T_i, \\
(n/3 + n(x_i - x_w))\rho & \text{ if } x \in T_j, \\
(n/3 + n(x_i - x_j))\rho & \text{ if } x \in T_w, \\
\lceil 2nx_i \rceil\rho & \text{ if } x \in H,
\end{cases}
\end{align*}
where $\{i,j,w\} = \{1,2,3\}$ and $x = (x_1, x_2, x_3)$ is the middle point of $F$. An example illustrating the potential function can be found in Figure~\ref{fig:pot}.

\begin{figure}[h!]
    \centering
    \resizebox{0.7\textwidth}{!}{\input{potential.tikz}}
   	\caption{The potential function $\Phi_1$ when $n = 9$.}
    \label{fig:pot}
\end{figure}
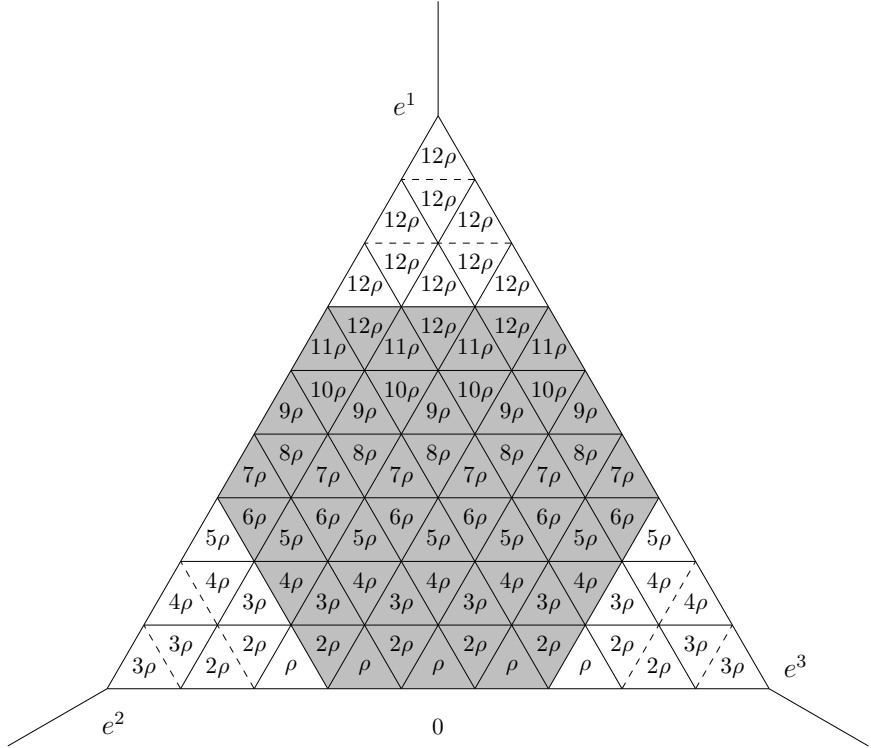

It is easy to check that, for any two triangles $F_1, F_2$ that share an edge, the difference $|\Phi_i(F_1) - \Phi_i(F_2)|$ is no more than the weight of the shared edge. The same remains true even when one of $F_1, F_2$ is $O_i$. This implies that, in the dual graph, the shortest path from $O_i$ to any triangle $F$ is at least $\Phi_i(F)$. With these observations in mind, we are ready to show that $\cost(P, w) \geqs 1$. Let us consider the two cases based on whether $P$ is a corner cut or a ball cut.

\emph{Case 1.} Suppose that $P$ is a corner cut. Observe that, for any $j \in [3] \setminus \{i\}$ and for any triangle $F$ sharing an edge with the outer face $O_j$, we always have $\Phi_i(F)$ plus the weight of the shared edge being at least $(2n/3)\rho = 1/3$. Hence, the shortest path from $O_i$ to $O_j$ in the dual graph has weight at least 1/3. Since $P$ contains three paths among $O_1, O_2, O_3$, the cost of $P$ is at least 1.

\emph{Case 2.} Suppose that $P$ is a ball cut. Let $F$ be the triangle which the three paths to $O_1, O_2, O_3$ in $P$ originate from. To show that $P$ has cost at least one, it is enough to show that the total length of the shortest paths from $F$ to $O_1, O_2$ and $O_3$ is at least one. Since these shortest paths are lower bounded by the potential functions, we only need to prove that $\Phi_1(F) + \Phi_2(F) + \Phi_3(F) \geqs 1$.

Let $x$ be the middle point of $F$. If $x \in H$, then 
\begin{align*}
\Phi_1(F) + \Phi_2(F) + \Phi_3(F) = \lceil 2nx_1 \rceil\rho + \lceil 2nx_2 \rceil\rho + \lceil 2nx_3 \rceil\rho \geqs (2nx_1)\rho + (2nx_2)\rho + (2nx_3)\rho = 1.
\end{align*}

On the other hand, if $x \notin H$, assume without loss of generality that $x \in T_1$. We have
\begin{align*}
\Phi_1(F) + \Phi_2(F) + \Phi_3(F) = (4n/3)\rho + \left(n/3 + n(x_2 - x_3) \right)\rho + \left(n/3 + n(x_3 - x_2) \right)\rho = (2n)\rho = 1.
\end{align*}

In all cases, we have $\cost(P, w) \geqs 1$, thus completing the proof of Lemma~\ref{lem:3sim}.
\end{proofof}

\section{An Integrality Gap of $6/5 - O(1/\sqrt{k})$} \label{sec:rootk}

We will now use the integrality gap from Lemma~\ref{lem:3sim} to construct an integrality gap for Multiway Cut with $k$ terminals, for $k \geqs 3$. The vertex set and the edge set of the graph will be $\Delta_{k, n} \subseteq \Delta_k$ and $E_{k, n} = \{(x, y) \in \Delta_{k, n} \times \Delta_{k, n} \mid \|x - y\|_1 = 2/n\}$ respectively. The goal of this section is to prove the following lemma.

\begin{lemma} \label{lem:rootk}
For every $n$ divisible by 3 and every $k \geqs 3$, there exists $\ow: E_{k, n} \rightarrow \Rz$ such that
\begin{itemize}
\item The value of the canonical LP solution of $(\Delta_{k, n}, E_{k, n}, w)$ is at most $5/6 + O(1/n)$.
\item The cost of every $k$-way cut of the graph is at least $1 - O(n/k)$.
\end{itemize}
\end{lemma}

By setting $n = \sqrt{k}$, we have an integrality gap of $6/5 - O(1/\sqrt{k})$. While the dependency on $k$ here is still worse than that in Theorem~\ref{thm:main}, the construction and proof encapsulate the main ideas of this paper. The gap from this section will finally be slightly tweaked in the next section to yield the right dependency on $k$.

The core of the proof of Lemma~\ref{lem:rootk} is the following proposition, which is essentially equivalent to $\tilde \tau^*_{3, n} - O(n/k) \leqs \tau^*_k$, albeit stated in more convenient integrality gap terminologies.

\begin{proposition} \label{prop:proj}
Let $P$ be any $k$-way cut of $\Delta_{k, n}$ for any $k \geqs 3$. For any $\{i_1, i_2, i_3\} \subseteq [k]$, let $P_{\{i_1, i_2, i_3\}}$ be as defined in the introduction. If $i_1, i_2, i_3$ are 3 randomly selected distinct elements of $[k]$, then $$\Pr_{i_1, i_2, i_3}[P_{\{i_1, i_2, i_3\}} \text{ is non-opposite}] \geqs 1 - O(n/k).$$
\end{proposition}

Before we prove the proposition, let us first use it to prove Lemma~\ref{lem:rootk}.

\begin{proofof}[Lemma~\ref{lem:rootk}]
Our gap is simply created by embedding the gap from Section~\ref{sec:3} to every triangular face of $\Delta_{k, n}$. More specifically, for each $\{i_1, i_2, i_3\} \in \binom{[k]}{3}$, the weights $\ow_{\{i_1, i_2, i_3\}}: E_{k, n} \rightarrow \Rz$ of the triangular integrality gap embedded to the face induced by $e^{i_1}, e^{i_2}, e^{i_3}$ can be defined as
\begin{align*}
\ow_{\{i_1, i_2, i_3\}}(x, y) =
\begin{cases}
w((x_{i_1}, x_{i_2}, x_{i_3}), (y_{i_1}, y_{i_2}, y_{i_3})) & \text{ if } \supp(x), \supp(y) \subseteq \{i_1, i_2, i_3\}, \\
0 & \text{ otherwise,}
\end{cases}
\end{align*}
where $w$ is the weight function from Lemma~\ref{lem:3sim}. Finally, the overall weighting scheme is just the average over all $\ow_{\{i_1, i_2, i_3\}}$'s, i.e., $$\ow = \E_{\{i_1, i_2, i_3\} \in \binom{[k]}{3}} \left[\ow_{\{i_1, i_2, i_3\}}\right].$$

Let us now compute the value of the canonical LP solution of $(\Delta_{k, n}, E_{k, n}, \ow)$. From the definition of $\ow$, $\lpc(\ow)$ is simply the average of $\lpc(\ow_{\{i_1, i_2, i_3\}})$ over all $\{i_1, i_2, i_3\} \in \binom{[k]}{3}$. Since $\lpc(\ow_{\{i_1, i_2, i_3\}}) = \lpc(w) = 5/6 + O(1/n)$, we can conclude that $\lpc(\ow) = 5/6 + O(1/n)$ as desired.

Next, we will argue that the cost of every $k$-way cut of the graph is at least $1 - O(n/k)$. Consider any $k$-way cut $P$. Before we formally prove that $\cost(P, \ow) \geqs 1 - O(n/k)$, let us give a short intuition behind the proof. First, by Proposition~\ref{prop:proj}, for at least $1 - O(n/k)$ fraction of $\{i_1, i_2, i_3\} \in \binom{[k]}{3}$, $P_{\{i_1, i_2, i_3\}}$ is non-opposite. From the second condition of Lemma~\ref{lem:3sim}, this means that the cost of $P$ with respect to $\ow_{\{i_1, i_2, i_3\}}$ is at least one. This implies that the total cost of $P$ with respect to $\ow$ is at least $1 - O(n/k)$.

More formally, the cost of $P$ can be written as
\begin{align*}
\cost(P, \ow) &= \E_{\{i_1, i_2, i_3\} \in \binom{[k]}{3}} [\cost(P, \ow_{\{i_1, i_2, i_3\}})] \\
(\text{from definitions of } \ow_{\{i_1, i_2, i_3\}}, P_{\{i_1, i_2, i_3\}}) &\geqs \E_{\{i_1, i_2, i_3\} \in \binom{[k]}{3}} [\cost(P_{\{i_1, i_2, i_3\}}, w)].
\end{align*}
Let $E_{\{i_1, i_2, i_3\}}$ be the event that $P_{\{i_1, i_2, i_3\}}$ is non-opposite. We can then lower bound the last expression by
\begin{align*}
\E_{\{i_1, i_2, i_3\} \in \binom{[k]}{3}} [\cost(P_{\{i_1, i_2, i_3\}}, w)] &\geqs \Pr[E_{\{i_1, i_2, i_3\}}] \cdot \E_{\{i_1, i_2, i_3\} \in \binom{[k]}{3}} [\cost(P_{\{i_1, i_2, i_3\}}, w) \mid E_{\{i_1, i_2, i_3\}}] \\
(\text{from Proposition}~\ref{prop:proj}) &\geqs \left(1 - O(n/ k)\right) \cdot \E_{\{i_1, i_2, i_3\} \in \binom{[k]}{3}} [\cost(P_{\{i_1, i_2, i_3\}}, w) \mid E_{\{i_1, i_2, i_3\}}] \\
&\geqs 1 - O(n/ k),
\end{align*}
where the last inequality comes from the second property of $w$ in Lemma~\ref{lem:3sim}.
\end{proofof}

We will next prove Proposition~\ref{prop:proj}; the proof is essentially the same as the proof of $\tilde \tau^*_{3, n} - O(n/k) \leqs \tau^*_k$ sketched earlier in the introduction.

\begin{proofof}[Proposition~\ref{prop:proj}]
By definition of non-opposite cuts of $\Delta_{3, n}$, if $P_{\{i_1, i_2, i_3\}}$ is not non-opposite, then there exists a point $x$ on one of the sides of $\Delta_{3, n}$ (i.e. $|\supp(x)| = 2$) such that $P_{\{i_1, i_2, i_3\}}(x) \notin \supp(x) \cup \{4\}$. Hence, by union bound, we have
\begin{align*}
\Pr_{i_1, i_2, i_3}[P_{\{i_1, i_2, i_3\}} \text{ is not non-opposite}] \leqs \sum_{x \in \Delta_{3, n}: \atop |\supp(x)| = 2} \Pr_{i_1, i_2, i_3}[P_{\{i_1, i_2, i_3\}}(x) \notin \supp(x) \cup \{4\}].
\end{align*}

Consider an $x \in \Delta_{3, n}$ with $|\supp(x)| = 2$. Assume without loss of generality that $\supp(x) = \{1, 2\}$. If $P_{\{i_1, i_2, i_3\}}(x) \notin \supp(x) \cup \{4\}$, then $P_{\{i_1, i_2, i_3\}}(x) = 3$. The latter occurs if and only if $P(x_1e^{i_1} + x_2e^{i_2} + x_3e^{i_3}) = i_3$. Fix $i_1, i_2$; this already determines $P(x_1e^{i_1} + x_2e^{i_2} + x_3e^{i_3})$. Since $i_3$ is now a random element from $[k] \setminus \{i_1, i_2\}$, the probability that $P(x_1e^{i_1} + x_2e^{i_2} + x_3e^{i_3}) = i_3$ is at most $1 / (k - 2)$. Hence, we have
\begin{align*}
\Pr_{i_1, i_2, i_3}[P_{\{i_1, i_2, i_3\}} \text{ is not non-opposite}] \leqs \sum_{x \in \Delta_{3, n}: \atop |\supp(x)| = 2} \frac{1}{k - 2} = \frac{3(n - 1)}{k - 2} = O(n/k),
\end{align*}
completing the proof of the proposition.
\end{proofof}

\section{Getting the Right Dependency on $k$} \label{sec:k}

We will modify our gap from the previous section to arrive at the following quantitatively better bound.

\begin{theorem} \label{thm:main-refined}
For every $n$ divisible by 3 and every $k \geqs 3$, there exists $\tilde w: E_{k, n} \rightarrow \Rz$ such that
\begin{itemize}
\item The value of the canonical LP solution of $(\Delta_{k, n}, E_{k, n}, \tilde w)$ is at most $(5 + \frac{1}{k - 1})/6 + O(1/n)$.
\item The cost of every $k$-way cut of the graph is at least one.
\end{itemize}
\end{theorem}

Clearly, by taking $n$ sufficiently large (i.e. $n = \Omega(1/\varepsilon)$), Theorem~\ref{thm:main-refined} implies Theorem~\ref{thm:main} since the value of the canonical LP solution is no more than the value of the optimal LP solution.

The rest of the section is arranged as follows. First, in the first subsection, we refine our analysis from Section~\ref{sec:rootk} to get improved bounds for the costs of $k$-way cuts against the weight function. Unfortunately, this will not be enough to prove Theorem~\ref{thm:main-refined} and we will have to provide another simple gap to complement worst cases for the gap from Section~\ref{sec:rootk}; this is done in Subsection~\ref{subsec:edgegap}. Finally, we take an appropriate linear combination of the two gaps and prove Theorem~\ref{thm:main-refined} in the last subsection.

\subsection{Refined Analysis for the Gap from Section~\ref{sec:rootk}} \label{subsec:refined}

There are two main points of the analysis from the previous section that need to be tightened in order to get the bound in Theorem~\ref{thm:main-refined}. First, notice that the bound in Proposition~\ref{prop:proj} can be improved if, for each two vertices $e^i, e^j$ of the simplex, the points on the line between $e^i$ and $e^j$ are assigned to few clusters. For instance, in the extreme case where every point on the line from $e^i$ to $e^j$ is assigned to either cluster $i$ or cluster $j$, $P_{\{i_1, i_2, i_3\}}$ is always non-opposite for every $i_1, i_2, i_3$. The following proposition gives a more refined bound on the probability that $P_{\{i_1, i_2, i_3\}}$ is non-opposite based on the number of clusters that the points on the line from $e^i$ to $e^j$ are assigned to for every $i, j \in [k]$.

\begin{proposition} \label{prop:proj-refined}
Let $P$ be any $k$-way cut of $\Delta_{k, n}$ for any $k \geqs 3$. For each $\{i, j\} \subseteq [k]$, let us define $D_{i, j}(P)$ to be the set of all clusters that the points on the line between $e^i$ and $e^j$ (inclusive) are assigned to, i.e., $$D_{i, j}(P) = \{P(x) \mid x \in \Delta_{k, n}, \supp(x) \subseteq \{i, j\}\}.$$ Finally, let $D(P)$ be the average of the size of $D_{i, j}(P)$ among all $\{i, j\} \subseteq [k]$, i.e., $D(P) = \E_{\{i, j\} \subseteq [k]}[|D_{i, j}(P)|]$. If $i_1, i_2, i_3$ are three randomly selected distinct elements of $[k]$, then $$\Pr_{i_1, i_2, i_3}[P_{\{i_1, i_2, i_3\}} \text{ is non-opposite}] \geqs 1 - 3(D(P) - 2)/(k - 2).$$
\end{proposition}

\begin{proof}
We want to bound the probability that $P_{\{i_1, i_2, i_3\}}$ is not non-opposite. Again, this happens only when there exists $x$ such that $\supp(x) \subseteq \{i_1, i_2, i_3\}$ and $P(x) \in \{i_1, i_2, i_3\} \setminus \supp(x)$. Note that such $x$ can only have $|\supp(x)| = 2$. Thus, by union bound, we have
\begin{align*}
&\Pr_{i_1, i_2, i_3}[P_{\{i_1, i_2, i_3\}} \text{ is not non-opposite}] \\
&\leqs \sum_{\{j_1, j_2\} \subseteq \{i_1, i_2, i_3\}}\Pr_{i_1, i_2, i_3}[\exists x, \supp(x) = \{j_1, j_2\} \text{ and } P(x) \in \{i_1, i_2, i_3\} \setminus \{j_1, j_2\}] \\
&= 3 \Pr_{i_1, i_2, i_3}[\exists x, \supp(x) = \{i_1, i_2\} \text{ and } P(x) = i_3],
\end{align*}
where the last equality comes from symmetry. Now, let us fix $i_1, i_2$. There is an $x$ with $\supp(x) = \{i_1, i_2\}$ such that $P(x) = i_3$ if and only if $i_3 \in D_{i_1, i_2}(P)$. Observe that $i_1, i_2 \in D_{i_1, i_2}(P)$. Hence, the probability that $i_3$, a uniformly random element from $[k] \setminus \{i_1, i_2\}$, lies in $D_{i_1, i_2}(P)$ is exactly $(|D_{i_1, i_2}(P)| - 2)/(k - 2)$. Thus,
\begin{align*}
\Pr_{i_1, i_2, i_3}[\exists x, \supp(x) = \{i_1, i_2\} \text{ and } P(x) = i_3] &= \E_{i_1, i_2}[\Pr_{i_3}[\exists x, \supp(x) = \{i_1, i_2\} \text{ and } P(x) = i_3]] \\
&= \E_{i_1, i_2}[(|D_{i_1, i_2}(P)| - 2)/(k - 2)] \\
&= (D(P) - 2)/(k - 2).
\end{align*}

Hence, $\Pr_{i_1, i_2, i_3}[P_{\{i_1, i_2, i_3\}} \text{ is non-opposite}] \geqs 1 - 3(D(P) - 2)/(k - 2)$ as desired.
\end{proof}

Another point in our analysis that was not tight was that, in the proof of Lemma~\ref{lem:rootk}, we only used the fact that $\cost(P, \ow_{\{i_1, i_2, i_3\}})$ is at least one when $P_{\{i_1, i_2, i_3\}}$ is non-opposite. However, even when $P_{\{i_1, i_2, i_3\}}$ is not non-opposite, $\cost(P, \ow_{\{i_1, i_2, i_3\}})$ is non-zero and we should include this in our argument, too. To do so, we need the following observation.

\begin{observation} \label{obs:3way}
Let $w$ be as defined in Lemma~\ref{lem:3sim}. The cost of any 3-way cut of $(\Delta_{3, n}, E_{3, n})$ with respect to $w$ is at least 2/3. 
\end{observation}

To see that the observation is true, recall Observation~\ref{obs:char-3way}; to prove that any 3-way cut of $(\Delta_{3, n}, E_{3, n})$ has cost at least 2/3, we only need to consider ball cuts and 2-corner cuts. However, we have shown in Subsection~\ref{subsec:int-nonop} that the cost of any ball cut and 3-corner cut with respect to $w$ is at least one. This immediately implies that the cost of any 2-corner cut is at least 2/3, meaning that Observation~\ref{obs:3way} is indeed true.

With the above proposition and the observation, we can now prove a refined version of Lemma~\ref{lem:rootk}, as stated formally below. The proof is very similar to that of the original lemma.

\begin{lemma} \label{lem:rootk-refined}
For $\ow$ as defined in Lemma~\ref{lem:rootk} and for any $k$-way cut $P$, we have $\cost(P, \ow) \geqs 1 - (D(P) - 2) / (k - 2)$.
\end{lemma}

\begin{proof}
Recall that we have
\begin{align*}
\cost(P, \ow) \geqs \E_{\{i_1, i_2, i_3\} \subseteq [k]} [\cost(P_{\{i_1, i_2, i_3\}}, w)].
\end{align*}
Again, let $E_{\{i_1, i_2, i_3\}}$ be the event that $P_{\{i_1, i_2, i_3\}}$ is non-opposite. We can then lower bound $\cost(P, \ow)$ as follows.
\begin{align*}
\E_{\{i_1, i_2, i_3\} \subseteq [k]} [\cost(P_{\{i_1, i_2, i_3\}}, w)] =& \Pr[E_{\{i_1, i_2, i_3\}}] \cdot \E_{\{i_1, i_2, i_3\} \subseteq [k]} [\cost(P_{\{i_1, i_2, i_3\}}, w) \mid E_{\{i_1, i_2, i_3\}}] + \\
&\Pr[\neg E_{\{i_1, i_2, i_3\}}] \cdot \E_{\{i_1, i_2, i_3\} \subseteq [k]} [\cost(P_{\{i_1, i_2, i_3\}}, w) \mid \neg E_{\{i_1, i_2, i_3\}}] \\
(\text{from Lemma}~\ref{lem:3sim} \text{ and Observation}~\ref{obs:3way}) \geqs& \Pr[E_{\{i_1, i_2, i_3\}}] + \frac{2}{3} \cdot \Pr[\neg E_{\{i_1, i_2, i_3\}}] \\
=& 1 - \frac{1}{3} \cdot \Pr[\neg E_{\{i_1, i_2, i_3\}}] \\
(\text{from Proposition}~\ref{prop:proj-refined}) \geqs& 1 - \frac{D(P) - 2}{k - 2}.
\end{align*}
\end{proof}

\subsection{An Integrality Gap Against Large $D(P)$} \label{subsec:edgegap}

As seen in the previous subsection, the weight $\ow$ works well against cuts $P$ with small $D(P)$ but it does not work well against $P$ with large $D(P)$. As a result, we need to construct a different integrality gap instance that deals with this case. Formally, the new gap will have the following properties.

\begin{lemma} \label{lem:edgegap}
For every $n \geqs 2$ and every $k \geqs 2$, there exists $w': E_{k, n} \rightarrow \Rz$ such that
\begin{itemize}
\item The value of the canonical LP solution of $(\Delta_{k, n}, E_{k, n}, w')$ is exactly one.
\item The cost of every $k$-way cut $P$ of the graph is at least $D(P) - 1$.
\end{itemize}
\end{lemma}

\begin{proof}
The gap construction is intuitive; since $D(P)$ being large means that many edges on the lines between two vertices of the simplex are cut, we simply define the gap so that we place all the weights equally on such edges. Specifically, define $w'(x, y)$ for all $(x, y) \in E_{k, n}$ as follows.
\begin{align*}
w'(x, y) =
\begin{cases}
\frac{1}{\binom{k}{2}} & \text{ if } |\supp(x) \cup \supp(y)| = 2, \\
0 & \text{ otherwise.}
\end{cases}
\end{align*}

Clearly, $\lpc(w') = 1$ since each line from $e^i$ to $e^j$ contributes $1/\binom{k}{2}$ to the value of the solution.

Now, consider any $k$-way cut $P$. The cost of the cut can be written as
\begin{align*}
\cost(P, w') &= \sum_{(x, y) \in E_{k, n}} w'(x, y) \cdot \ind[P(x) \ne P(y)] \\
&= \sum_{(x, y) \in E_{k, n}: \atop |\supp(x) \cup \supp(y)| = 2} \frac{1}{\binom{k}{2}} \cdot \ind[P(x) \ne P(y)] \\
&= \frac{1}{\binom{k}{2}} \sum_{\{i, j\} \subseteq [k]} \sum_{(x, y) \in E_{k, n}: \atop \supp(x) \cup \supp(y) = \{i, j\}} \ind[P(x) \ne P(y)] \\
&= \E_{\{i, j\} \subseteq [k]} \left[\sum_{(x, y) \in E_{k, n}: \atop \supp(x) \cup \supp(y) = \{i, j\}} \ind[P(x) \ne P(y)]\right].
\end{align*}
The term in the expectation is simply the number of edges on the line between $e^i$ and $e^j$ cut by $P$. Since we know that the points on this line are assigned to $|D_{i, j}(P)|$ different clusters, we can lower bound the number of edges cut by $|D_{i,j}(P)| - 1$, which gives
\begin{align*}
\cost(P, w') \geqs \E_{\{i, j\} \subseteq [k]} \left[|D_{i, j}(P)| - 1\right] = D(P) - 1.
\end{align*}
This completes the proof of Lemma~\ref{lem:edgegap}.
\end{proof}

\subsection{Putting Things Together}

It is now easy to prove our main theorem by simply picking the weight $\tw$ to be an appropriate linear combination of $\ow$ and $w'$.

\begin{proofof}[Theorem~\ref{thm:main-refined}]
Let $\ow$ and $w'$ be as defined in Lemma~\ref{lem:rootk} and Lemma~\ref{lem:edgegap} respectively. We define $\tw$ to be $$\tw = \left(\frac{k - 2}{k - 1}\right) \ow + \left(\frac{1}{k - 1}\right) w'.$$

The value of the canonical LP solution of $(\Delta_{k, n}, E_{k, n}, \tw)$ is simply
\begin{align*}
\lpc(\tw) &= \left(\frac{k - 2}{k - 1}\right) \lpc(\ow) + \left(\frac{1}{k - 1}\right) \lpc(w') \\
&= \left(\frac{k - 2}{k - 1}\right) (5/6 + O(1/n)) + \left(\frac{1}{k - 1}\right) \\
&= \left(5 + \frac{1}{k - 1}\right) / 6 + O(1/n).
\end{align*}

Moreover, for any $k$-way cut $P$, we have
\begin{align*}
\cost(P, \tw) &= \left(\frac{k - 2}{k - 1}\right) \cost(P, \ow) + \left(\frac{1}{k - 1}\right) \cost(P, w') \\
(\text{from Lemma}~\ref{lem:rootk-refined} \text{ and Lemma}~\ref{lem:edgegap}) &\geqs \left(\frac{k - 2}{k - 1}\right)\left(1 - \frac{D(P) - 2}{k - 2}\right) + \left(\frac{1}{k - 1}\right)\left(D(P) - 1\right) \\
&= 1.
\end{align*}

Hence, $\ow$ satisfies the two properties in Theorem~\ref{thm:main-refined}.
\end{proofof}

\section{Conclusion and Discussions} \label{sec:open}

In this paper, we construct integrality gap instances of ratio $6/(5 + \frac{1}{k - 1}) - \varepsilon$ for every $\varepsilon > 0$ for the C{\u{a}}linescu-Karloff-Rabani linear program relaxation of Multiway Cut with $k \geqs 3$ terminals. Combined with the result of Manokaran \etal~\cite{MNRS08}, our result implies that it is UG-hard to approximate Multiway Cut with $k$ terminals to within $6/(5 + \frac{1}{k - 1}) - \varepsilon$ ratio of the optimum. Our gap construction is based on the observation that $\ttau^*_{3, n} - O(n/K)$ lower bounds $\tau^*_K$; with this observation, we extend Karger \etal's characterization of 3-way cuts~\cite{KKSTY04} to non-opposite cuts, which ultimately leads to a lower bound of $6/5 - O(1/n)$ for $\ttau^*_{3, n}$. We also observe that Freund and Karloff's integrality gap instance~\cite{FK00} is subsumed by our approach.

Even with our result, the approximability of Multiway Cut is far from resolved as the best known approximation algorithm by Sharma and Vondr\'{a}k~\cite{SV14} only achieves approximation ratio of $1.2965$ as $k$ goes to infinity whereas our gap is only $1.2$. It may be tempting to try to prove a stronger lower bound for $\ttau^*_{3, n}$ in order to construct improved integrality gaps. Unfortunately, this is impossible since a slight modification of Karger \etal's algorithm for Multiway Cut with $3$ terminals shows that $\ttau^*_3 \leqs 1.2$. For completeness, we include the modified algorithm in Appendix~\ref{app:rounding}. With this in mind, it is likely that, to construct better integrality gaps, we need to exploit geometric properties of higher-dimensional simplexes. This seems much harder to do as all the known gaps, including ours, only deal with the two-dimensional simplex $\Delta_3$, whose nice properties allow the cuts to be characterized in a convenient manner.

Finally, we would like to note that the proof in Section~\ref{sec:k} can easily be converted to a proof of $\ttau^*_3 - O(1/K) \leqs \tau^*_K$. Notice that $\ttau^*_3$ appears in the inequality as opposed to its discretized approximation $\ttau^*_{3, n}$ as sketched in the introduction. Unfortunately, we do not know of any way of extending this proof to show a direct relation between $\ttau^*_k$ and $\tau^*_K$ for $k > 3$. Such relation may be crucial if one wants to construct an integrality gap instance by proving a lower bound for $\ttau^*_k$.

\bibliographystyle{alpha}
\bibliography{main}

\appendix

\section{The Freund-Karloff Gap, Revisited} \label{app:FK}

In this section, we will rephrase the Freund-Karloff integrality gap~\cite{FK00} in terms of a lower bound for $\ttau^*_{3,2}$. In $\Delta_{3, 2}$, there are only six vertices: the three vertices of the simplex and the three middle points. The Freund-Karloff gap sets the weight of the edges between a simplex vertex and a middle point to be 1/6 and sets the weight of the edges among the middle points to be 1/4.

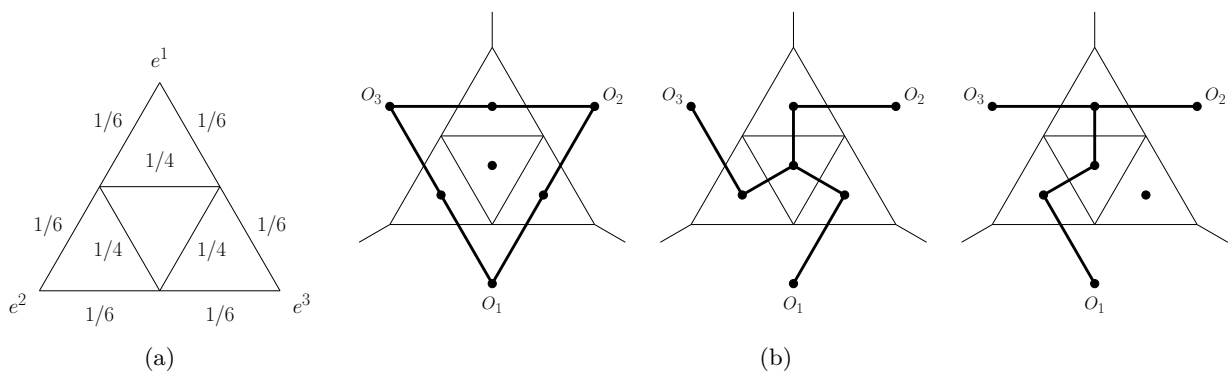
\begin{figure}[h!]
    \centering
    \begin{subfigure}[b]{0.25\textwidth}
        \resizebox{\textwidth}{!}{\input{fk-gap.tikz}}
        \caption{}
        \label{fig:fk-gap}
    \end{subfigure}
    ~ 
    \begin{subfigure}[b]{0.7\textwidth}
    	\begin{tabular}{c c c}
    	\resizebox{0.31\textwidth}{!}{\input{fk-corner.tikz}} & 
    	\resizebox{0.31\textwidth}{!}{\input{fk-ball-1.tikz}} &
    	\resizebox{0.31\textwidth}{!}{\input{fk-ball-2.tikz}}
    	\end{tabular}
    	\caption{}
    	\label{fig:fk-ball-corner}
    \end{subfigure}
    \caption{The Freund-Karloff Integrality Gap and Candidate Optimal Non-Opposite Cuts. Figure~\ref{fig:fk-gap} demonstrates the Freund-Karloff integrality gap whereas Figure~\ref{fig:fk-ball-corner} shows the only three possible ball and corner cuts (up to rotations and reflections).}
\end{figure}

The illustrations of the gap can be found in Figure~\ref{fig:fk-gap}. Clearly, the canonical LP solution has value 7/8. Moreover, discounting symmetry, there are only three ball and corner cuts, as shown in Figure~\ref{fig:fk-ball-corner}. Obviously, every cut has value at least one. Hence, this is an 8/7 lower bound on $\ttau^*_{3,2}$ as intended.

Finally, we note that, by using the technique from Section~\ref{sec:k}, one can arrive at $8/(7 + \frac{1}{k - 1})$ integrality gap for Multiway Cut with $k$ terminals, which is exactly the result of Freund and Karloff.

\section{An Inequality Between $\ttau^*_{k, n}$ and $\tau^*_K$} \label{app:tk}

In this section, we prove a relation between $\ttau^*_{k, n}$ and $\tau^*_K$ for $K > k$. This generalizes the inequality between $\ttau^*_{3, n}$ and $\tau^*_K$, whose proof sketch was presented in the introduction.

\begin{lemma}
For every $K > k$, $\ttau^*_{k, n} - O(kn/(K - k)) \leqs \tau^*_K$.
\end{lemma}

\begin{proof}
Suppose that $\cP$ is the distribution on $K$-way cuts such that $\tau_K(\cP) = \tau^*_K$. We define a distribution $\tilde \cP$ on non-opposite cuts of $\Delta_k$ by a sampling process for $\tilde P \sim \tilde \cP$ as follows. First, sample $P \sim \cP$. Then, randomly select an injection $f: [k] \rightarrow [K]$. In short, we will set $\tilde P$ to be the restriction of $P$ on the face induced by $e^{f(1)}, \dots, e^{f(k)}$ and ``fix'' $\tilde P$ so that it is non-opposite. 

For convenience, let $x \circ f$ denote $\sum_{i \in [k]} x_i e^{f(i)} \in \Delta_K$ for every $x \in \Delta_k$. $\tilde P$ can then be defined as follows.
\begin{align*}
\tilde P(x) =
\begin{cases}
f^{-1}(P(x \circ f)) & \text{ if } P(x \circ f) \in \supp(x \circ f), \\
k + 1 & \text{ otherwise.}
\end{cases}
\end{align*}

Since $\supp(x \circ f) = f(\supp(x))$, it is clear from the definition above that $\tilde P$ is indeed a non-opposite cut of $\Delta_k$. We are left to only argue that $\tau_{k, n}(\tilde \cP) \leqs \tau_K(\cP) + O(k/K)$. Consider any two different points $x, y \in \Delta_{k, n}$. Observe that, if $\tilde P(x) \ne \tilde P(y)$, then either $P(x \circ f) \ne P(y \circ f)$ or $P(x \circ f) = P(y \circ f) \in \supp(x \circ f) \oplus \supp(y \circ f)$ where $\oplus$ denotes the symmetric difference between the two sets. The latter also implies that either $P(x \circ f) \in f([k] \setminus \supp(x))$ or $P(y \circ f) \in f([k] \setminus \supp(y))$. By union bound, we have
\begin{align*}
\Pr_{\tP \sim \tcP}[\tP(x) \ne \tP(y)] &\leqs \Pr_{P \sim \cP, f}[P(x \circ f) \ne P(y \circ f)] + \Pr_{P \sim \cP, f}[P(x \circ f) = P(y \circ f) \in \supp(x \circ f) \oplus \supp(y \circ f)]\\
&\leqs \tau^*_K \cdot \frac{1}{2} \|x - y\|_1 + \Pr_{P \sim \cP, f}[P(x \circ f) \in f([k] \setminus \supp(x))] + \Pr_{P \sim \cP, f}[P(y \circ f) \in f([k] \setminus \supp(y))].
\end{align*}

Let us now bound $\Pr_{P \sim \cP, f}[P(x \circ f) \in f([k] \setminus \supp(x))]$. Fix $P$ and fix $g = f|_{\supp(x)}$. Note that fixing $g$ already determines $x \circ f$ and hence $P(x \circ f)$. If $P(x \circ f)$ is in $g(\supp(x)) = f(\supp(x))$, then clearly $P(x \circ f) \notin f([k] \setminus \supp(x))$. On the other hand, if $P(x \circ f) \notin f(\supp(x))$, then $P(x \circ f) \in f([k] \setminus \supp(x))$ if and only if $P(x \circ f) = f(i)$ for some $i \in [k] \setminus \supp(x)$. For a fixed $i$, the probability that $P(x \circ f) = f(i)$ is only $1 / (K - \supp(x)) \leqs 1 / (K - k)$. Hence, by union bound, $\Pr_{P \sim \cP, f}[P(x \circ f) \in f([k] \setminus \supp(x))] \leqs k / (K - k)$.

A similar bound can be derived for $\Pr_{P \sim \cP, f}[P(y \circ f) \in f([k] \setminus \supp(y))]$, which implies that 
\begin{align*}
\Pr_{\tP \sim \tcP}[\tP(x) \ne \tP(y)]
&\leqs \tau^*_K \cdot \frac{1}{2} \|x - y\|_1 + 2k/(K - k).
\end{align*}

Finally, note that, for all $x \ne y \in \Delta_{k, n}$, the distance $\|x - y\|_1$ is at least $2/n$. Hence, we have $\Pr_{\tP \sim \tcP}[\tP(x) \ne \tP(y)] \leqs (\tau^*_K + O(kn/(K - k))) \cdot \frac{1}{2} \|x - y\|_1$
, meaning that $\tilde \tau^*_k - O(kn/(K - k)) \leqs \tau^*_K$ as desired.
\end{proof}

\section{A Rounding Scheme for Non-Opposite Cuts of $\Delta_3$} \label{app:rounding}

In this section, we provide a probability distribution $\cP$ on non-opposite cuts of $\Delta_3$ such that $\tau(\cP) \leqs 6/5$, thereby proving that $\ttau^*_3 \leqs 6/5$. The probability distribution is a slight modification of the algorithm for Multiway Cut with 3 terminals by Karger \etal~\cite{KKSTY04}. 

A cut $P \sim \cP$ is drawn as follows:
\begin{itemize}
\item With probability 1/5, we select $r \in [2/3, 1)$ uniformly at random. Let $P$ be the 3-corner cut in which each $x \in \Delta_3$ is assigned to $i$ if and only if $x_i > r$ and is assigned to $4$ if none of its coordinates exceeds $r$.
\item With probability 4/5, we pick $P$ to be a ball cut by the following process. With probability 1/2, let $r$ be a random point on the line between $(2/3, 1/3, 0)$ and $(0, 2/3, 1/3)$. Otherwise, let $r$ be a random point on the line between $(2/3, 0, 1/3)$ and $(0, 1/3, 2/3)$. Consider segments starting from $r$ which are parallel to the sides of $\Delta_3$. There are six such segments with two of them ending on each side of the triangle. For each side, randomly pick one of the two segments. Finally, let $P$ be the cut induced by the three segments picked.
\end{itemize}
We note that the distribution is the same as that in Theorem 4.3 of~\cite{KKSTY04} except that here the corner cut is chosen with probability 1/5 instead of 3/11 as in~\cite{KKSTY04}; for illustrations of the cuts, please refer to Figure 5 in~\cite{KKSTY04}. Since the proof remains essentially the same as that in the mentioned theorem, we do not provide it here.

%\begin{lemma}
%$\ttau^*_3 \leqs 6/5$.
%\end{lemma}

\end{document}

%% file: oppositecut.tikz
\begin{tikzpicture}
\draw [color=black] (5.0000000000, -2.8867513459) -- (-5.0000000000, -2.8867513459);
\draw [color=black] (5.0000000000, -2.8867513459) -- (0.0000000000, 5.7735026919);
\draw [color=black] (-5.0000000000, -2.8867513459) -- (0.0000000000, 5.7735026919);
\node at (0.0,6.81273317644) {\Huge $e^1$};
\node at (-5.9,-3.40636658822) {\Huge $e^2$};
\node at (5.9,-3.40636658822) {\Huge $e^3$};
\draw [color=black] (-3.5000000000, -0.2886751346) -- (-1.0000000000, -2.8867513459);
\draw [color=black] (3.0000000000, 0.5773502692) -- (1.0000000000, -2.8867513459);
\node at (-3.16666666667,-2.02072594216) {\Huge 2};
\node at (3.0,-1.73205080757) {\Huge 3};
\node at (0.0,0.57735026919) {\Huge 1};
\end{tikzpicture}

%% file: nonop-ball.tikz
\begin{tikzpicture}
\draw [color=black] (5.0000000000, -2.8867513459) -- (-5.0000000000, -2.8867513459);
\draw [color=black] (5.0000000000, -2.8867513459) -- (0.0000000000, 5.7735026919);
\draw [color=black] (-5.0000000000, -2.8867513459) -- (0.0000000000, 5.7735026919);
\node at (0.0,6.81273317644) {\Huge $e^1$};
\node at (-5.9,-3.40636658822) {\Huge $e^2$};
\node at (5.9,-3.40636658822) {\Huge $e^3$};
\draw [color=black] (-3.2500000000, 0.1443375673) -- (0.5000000000, -0.2886751346);
\draw [color=black] (3.0000000000, 0.5773502692) -- (0.5000000000, -0.2886751346);
\draw [color=black] (0.0000000000, -2.8867513459) -- (0.5000000000, -0.2886751346);
\node at (-0.0833333333333,2.16506350946) {\Huge 1};
\node at (-2.75,-1.87638837487) {\Huge 2};
\node at (2.66666666667,-1.73205080757) {\Huge 3};
\end{tikzpicture}

%% file: nonop-corner.tikz
\begin{tikzpicture}
\draw [color=black] (5.0000000000, -2.8867513459) -- (-5.0000000000, -2.8867513459);
\draw [color=black] (5.0000000000, -2.8867513459) -- (0.0000000000, 5.7735026919);
\draw [color=black] (-5.0000000000, -2.8867513459) -- (0.0000000000, 5.7735026919);
\node at (0.0,6.81273317644) {\Huge $e^1$};
\node at (-5.9,-3.40636658822) {\Huge $e^2$};
\node at (5.9,-3.40636658822) {\Huge $e^3$};
\draw [color=black] (-3.5000000000, -0.2886751346) -- (-1.0000000000, -2.8867513459);
\draw [color=black] (3.0000000000, 0.5773502692) -- (1.0000000000, -2.8867513459);
\draw [color=black] (-2.0000000000, 2.3094010768) -- (1.2500000000, 3.6084391824);
\node at (-0.25,3.89711431703) {\Huge 1};
\node at (-3.16666666667,-2.02072594216) {\Huge 2};
\node at (3.0,-1.73205080757) {\Huge 3};
\node at (0,0) {\Huge 4};
\end{tikzpicture}

%% file: primal.tikz
\begin{tikzpicture}
\draw [color=black] (5.0000000000, -2.8867513459) -- (3.3333333333, -2.8867513459);
\draw [color=black] (5.0000000000, -2.8867513459) -- (4.1666666667, -1.4433756730);
\draw [color=black] (3.3333333333, -2.8867513459) -- (1.6666666667, -2.8867513459);
\draw [color=black] (3.3333333333, -2.8867513459) -- (4.1666666667, -1.4433756730);
\draw [color=black] (3.3333333333, -2.8867513459) -- (2.5000000000, -1.4433756730);
\draw [color=black] (1.6666666667, -2.8867513459) -- (0.0000000000, -2.8867513459);
\draw [color=black] (1.6666666667, -2.8867513459) -- (2.5000000000, -1.4433756730);
\draw [color=black] (1.6666666667, -2.8867513459) -- (0.8333333333, -1.4433756730);
\draw [color=black] (0.0000000000, -2.8867513459) -- (-1.6666666667, -2.8867513459);
\draw [color=black] (0.0000000000, -2.8867513459) -- (0.8333333333, -1.4433756730);
\draw [color=black] (0.0000000000, -2.8867513459) -- (-0.8333333333, -1.4433756730);
\draw [color=black] (-1.6666666667, -2.8867513459) -- (-3.3333333333, -2.8867513459);
\draw [color=black] (-1.6666666667, -2.8867513459) -- (-0.8333333333, -1.4433756730);
\draw [color=black] (-1.6666666667, -2.8867513459) -- (-2.5000000000, -1.4433756730);
\draw [color=black] (-3.3333333333, -2.8867513459) -- (-5.0000000000, -2.8867513459);
\draw [color=black] (-3.3333333333, -2.8867513459) -- (-2.5000000000, -1.4433756730);
\draw [color=black] (-3.3333333333, -2.8867513459) -- (-4.1666666667, -1.4433756730);
\draw [color=black] (-5.0000000000, -2.8867513459) -- (-4.1666666667, -1.4433756730);
\draw [color=black] (4.1666666667, -1.4433756730) -- (2.5000000000, -1.4433756730);
\draw [color=black] (4.1666666667, -1.4433756730) -- (3.3333333333, -0.0000000000);
\draw [color=black] (2.5000000000, -1.4433756730) -- (0.8333333333, -1.4433756730);
\draw [color=black] (2.5000000000, -1.4433756730) -- (3.3333333333, -0.0000000000);
\draw [color=black] (2.5000000000, -1.4433756730) -- (1.6666666667, -0.0000000000);
\draw [color=black] (0.8333333333, -1.4433756730) -- (-0.8333333333, -1.4433756730);
\draw [color=black] (0.8333333333, -1.4433756730) -- (1.6666666667, -0.0000000000);
\draw [color=black] (0.8333333333, -1.4433756730) -- (0.0000000000, -0.0000000000);
\draw [color=black] (-0.8333333333, -1.4433756730) -- (-2.5000000000, -1.4433756730);
\draw [color=black] (-0.8333333333, -1.4433756730) -- (0.0000000000, -0.0000000000);
\draw [color=black] (-0.8333333333, -1.4433756730) -- (-1.6666666667, -0.0000000000);
\draw [color=black] (-2.5000000000, -1.4433756730) -- (-4.1666666667, -1.4433756730);
\draw [color=black] (-2.5000000000, -1.4433756730) -- (-1.6666666667, -0.0000000000);
\draw [color=black] (-2.5000000000, -1.4433756730) -- (-3.3333333333, -0.0000000000);
\draw [color=black] (-4.1666666667, -1.4433756730) -- (-3.3333333333, -0.0000000000);
\draw [color=black] (3.3333333333, -0.0000000000) -- (1.6666666667, -0.0000000000);
\draw [color=black] (3.3333333333, -0.0000000000) -- (2.5000000000, 1.4433756730);
\draw [color=black] (1.6666666667, -0.0000000000) -- (0.0000000000, -0.0000000000);
\draw [color=black] (1.6666666667, -0.0000000000) -- (2.5000000000, 1.4433756730);
\draw [color=black] (1.6666666667, -0.0000000000) -- (0.8333333333, 1.4433756730);
\draw [color=black] (0.0000000000, -0.0000000000) -- (-1.6666666667, -0.0000000000);
\draw [color=black] (0.0000000000, -0.0000000000) -- (0.8333333333, 1.4433756730);
\draw [color=black] (0.0000000000, -0.0000000000) -- (-0.8333333333, 1.4433756730);
\draw [color=black] (-1.6666666667, -0.0000000000) -- (-3.3333333333, -0.0000000000);
\draw [color=black] (-1.6666666667, -0.0000000000) -- (-0.8333333333, 1.4433756730);
\draw [color=black] (-1.6666666667, -0.0000000000) -- (-2.5000000000, 1.4433756730);
\draw [color=black] (-3.3333333333, -0.0000000000) -- (-2.5000000000, 1.4433756730);
\draw [color=black] (2.5000000000, 1.4433756730) -- (0.8333333333, 1.4433756730);
\draw [color=black] (2.5000000000, 1.4433756730) -- (1.6666666667, 2.8867513459);
\draw [color=black] (0.8333333333, 1.4433756730) -- (-0.8333333333, 1.4433756730);
\draw [color=black] (0.8333333333, 1.4433756730) -- (1.6666666667, 2.8867513459);
\draw [color=black] (0.8333333333, 1.4433756730) -- (0.0000000000, 2.8867513459);
\draw [color=black] (-0.8333333333, 1.4433756730) -- (-2.5000000000, 1.4433756730);
\draw [color=black] (-0.8333333333, 1.4433756730) -- (0.0000000000, 2.8867513459);
\draw [color=black] (-0.8333333333, 1.4433756730) -- (-1.6666666667, 2.8867513459);
\draw [color=black] (-2.5000000000, 1.4433756730) -- (-1.6666666667, 2.8867513459);
\draw [color=black] (1.6666666667, 2.8867513459) -- (0.0000000000, 2.8867513459);
\draw [color=black] (1.6666666667, 2.8867513459) -- (0.8333333333, 4.3301270189);
\draw [color=black] (0.0000000000, 2.8867513459) -- (-1.6666666667, 2.8867513459);
\draw [color=black] (0.0000000000, 2.8867513459) -- (0.8333333333, 4.3301270189);
\draw [color=black] (0.0000000000, 2.8867513459) -- (-0.8333333333, 4.3301270189);
\draw [color=black] (-1.6666666667, 2.8867513459) -- (-0.8333333333, 4.3301270189);
\draw [color=black] (0.8333333333, 4.3301270189) -- (-0.8333333333, 4.3301270189);
\draw [color=black] (0.8333333333, 4.3301270189) -- (0.0000000000, 5.7735026919);
\draw [color=black] (-0.8333333333, 4.3301270189) -- (0.0000000000, 5.7735026919);
\draw [color=black] (0, 5.7735026919) -- (0.0, 7.50555349947);
\draw [color=black] (-5.0, -2.88675134595) -- (-6.5, -3.75277674973);
\draw [color=black] (5.0, -2.88675134595) -- (6.5, -3.75277674973);
\node at (-0.75,6.37972047455) {\Huge $e^1$};
\node at (-5.15,-3.83937929011) {\Huge $e^2$};
\node at (5.9,-2.54034118443) {\Huge $e^3$};
\node[opacity=0] at (0.0,-6.81273317644) {t};
\end{tikzpicture}

%% file: dual.tikz
\begin{tikzpicture}
\draw [color=black] (5.0000000000, -2.8867513459) -- (3.3333333333, -2.8867513459);
\draw [color=black] (5.0000000000, -2.8867513459) -- (4.1666666667, -1.4433756730);
\draw [color=black] (3.3333333333, -2.8867513459) -- (1.6666666667, -2.8867513459);
\draw [color=black] (3.3333333333, -2.8867513459) -- (4.1666666667, -1.4433756730);
\draw [color=black] (3.3333333333, -2.8867513459) -- (2.5000000000, -1.4433756730);
\draw [color=black] (1.6666666667, -2.8867513459) -- (0.0000000000, -2.8867513459);
\draw [color=black] (1.6666666667, -2.8867513459) -- (2.5000000000, -1.4433756730);
\draw [color=black] (1.6666666667, -2.8867513459) -- (0.8333333333, -1.4433756730);
\draw [color=black] (0.0000000000, -2.8867513459) -- (-1.6666666667, -2.8867513459);
\draw [color=black] (0.0000000000, -2.8867513459) -- (0.8333333333, -1.4433756730);
\draw [color=black] (0.0000000000, -2.8867513459) -- (-0.8333333333, -1.4433756730);
\draw [color=black] (-1.6666666667, -2.8867513459) -- (-3.3333333333, -2.8867513459);
\draw [color=black] (-1.6666666667, -2.8867513459) -- (-0.8333333333, -1.4433756730);
\draw [color=black] (-1.6666666667, -2.8867513459) -- (-2.5000000000, -1.4433756730);
\draw [color=black] (-3.3333333333, -2.8867513459) -- (-5.0000000000, -2.8867513459);
\draw [color=black] (-3.3333333333, -2.8867513459) -- (-2.5000000000, -1.4433756730);
\draw [color=black] (-3.3333333333, -2.8867513459) -- (-4.1666666667, -1.4433756730);
\draw [color=black] (-5.0000000000, -2.8867513459) -- (-4.1666666667, -1.4433756730);
\draw [color=black] (4.1666666667, -1.4433756730) -- (2.5000000000, -1.4433756730);
\draw [color=black] (4.1666666667, -1.4433756730) -- (3.3333333333, -0.0000000000);
\draw [color=black] (2.5000000000, -1.4433756730) -- (0.8333333333, -1.4433756730);
\draw [color=black] (2.5000000000, -1.4433756730) -- (3.3333333333, -0.0000000000);
\draw [color=black] (2.5000000000, -1.4433756730) -- (1.6666666667, -0.0000000000);
\draw [color=black] (0.8333333333, -1.4433756730) -- (-0.8333333333, -1.4433756730);
\draw [color=black] (0.8333333333, -1.4433756730) -- (1.6666666667, -0.0000000000);
\draw [color=black] (0.8333333333, -1.4433756730) -- (0.0000000000, -0.0000000000);
\draw [color=black] (-0.8333333333, -1.4433756730) -- (-2.5000000000, -1.4433756730);
\draw [color=black] (-0.8333333333, -1.4433756730) -- (0.0000000000, -0.0000000000);
\draw [color=black] (-0.8333333333, -1.4433756730) -- (-1.6666666667, -0.0000000000);
\draw [color=black] (-2.5000000000, -1.4433756730) -- (-4.1666666667, -1.4433756730);
\draw [color=black] (-2.5000000000, -1.4433756730) -- (-1.6666666667, -0.0000000000);
\draw [color=black] (-2.5000000000, -1.4433756730) -- (-3.3333333333, -0.0000000000);
\draw [color=black] (-4.1666666667, -1.4433756730) -- (-3.3333333333, -0.0000000000);
\draw [color=black] (3.3333333333, -0.0000000000) -- (1.6666666667, -0.0000000000);
\draw [color=black] (3.3333333333, -0.0000000000) -- (2.5000000000, 1.4433756730);
\draw [color=black] (1.6666666667, -0.0000000000) -- (0.0000000000, -0.0000000000);
\draw [color=black] (1.6666666667, -0.0000000000) -- (2.5000000000, 1.4433756730);
\draw [color=black] (1.6666666667, -0.0000000000) -- (0.8333333333, 1.4433756730);
\draw [color=black] (0.0000000000, -0.0000000000) -- (-1.6666666667, -0.0000000000);
\draw [color=black] (0.0000000000, -0.0000000000) -- (0.8333333333, 1.4433756730);
\draw [color=black] (0.0000000000, -0.0000000000) -- (-0.8333333333, 1.4433756730);
\draw [color=black] (-1.6666666667, -0.0000000000) -- (-3.3333333333, -0.0000000000);
\draw [color=black] (-1.6666666667, -0.0000000000) -- (-0.8333333333, 1.4433756730);
\draw [color=black] (-1.6666666667, -0.0000000000) -- (-2.5000000000, 1.4433756730);
\draw [color=black] (-3.3333333333, -0.0000000000) -- (-2.5000000000, 1.4433756730);
\draw [color=black] (2.5000000000, 1.4433756730) -- (0.8333333333, 1.4433756730);
\draw [color=black] (2.5000000000, 1.4433756730) -- (1.6666666667, 2.8867513459);
\draw [color=black] (0.8333333333, 1.4433756730) -- (-0.8333333333, 1.4433756730);
\draw [color=black] (0.8333333333, 1.4433756730) -- (1.6666666667, 2.8867513459);
\draw [color=black] (0.8333333333, 1.4433756730) -- (0.0000000000, 2.8867513459);
\draw [color=black] (-0.8333333333, 1.4433756730) -- (-2.5000000000, 1.4433756730);
\draw [color=black] (-0.8333333333, 1.4433756730) -- (0.0000000000, 2.8867513459);
\draw [color=black] (-0.8333333333, 1.4433756730) -- (-1.6666666667, 2.8867513459);
\draw [color=black] (-2.5000000000, 1.4433756730) -- (-1.6666666667, 2.8867513459);
\draw [color=black] (1.6666666667, 2.8867513459) -- (0.0000000000, 2.8867513459);
\draw [color=black] (1.6666666667, 2.8867513459) -- (0.8333333333, 4.3301270189);
\draw [color=black] (0.0000000000, 2.8867513459) -- (-1.6666666667, 2.8867513459);
\draw [color=black] (0.0000000000, 2.8867513459) -- (0.8333333333, 4.3301270189);
\draw [color=black] (0.0000000000, 2.8867513459) -- (-0.8333333333, 4.3301270189);
\draw [color=black] (-1.6666666667, 2.8867513459) -- (-0.8333333333, 4.3301270189);
\draw [color=black] (0.8333333333, 4.3301270189) -- (-0.8333333333, 4.3301270189);
\draw [color=black] (0.8333333333, 4.3301270189) -- (0.0000000000, 5.7735026919);
\draw [color=black] (-0.8333333333, 4.3301270189) -- (0.0000000000, 5.7735026919);
\draw [color=black] (0, 5.7735026919) -- (0.0, 7.50555349947);
\draw [color=black] (-5.0, -2.88675134595) -- (-6.5, -3.75277674973);
\draw [color=black] (5.0, -2.88675134595) -- (6.5, -3.75277674973);
\node at (-0.75,6.37972047455) {\Huge $e^1$};
\node at (-5.15,-3.83937929011) {\Huge $e^2$};
\node at (5.9,-2.54034118443) {\Huge $e^3$};
\draw[fill] (0.0000000000, -5.7735026919) circle [radius=0.1000000000];
\node at (0.0,-6.81273317644) {\Huge $O_1$};
\draw[fill] (5.0000000000, 2.8867513459) circle [radius=0.1000000000];
\node at (5.9,3.40636658822) {\Huge $O_2$};
\draw[fill] (-5.0000000000, 2.8867513459) circle [radius=0.1000000000];
\node at (-5.9,3.40636658822) {\Huge $O_3$};
\draw[fill] (4.1666666667, -2.4056261216) circle [radius=0.1000000000];
\draw[fill] (2.5000000000, -2.4056261216) circle [radius=0.1000000000];
\draw[fill] (3.3333333333, -1.9245008973) circle [radius=0.1000000000];
\draw[fill] (0.8333333333, -2.4056261216) circle [radius=0.1000000000];
\draw[fill] (1.6666666667, -1.9245008973) circle [radius=0.1000000000];
\draw[fill] (-0.8333333333, -2.4056261216) circle [radius=0.1000000000];
\draw[fill] (0.0000000000, -1.9245008973) circle [radius=0.1000000000];
\draw[fill] (-2.5000000000, -2.4056261216) circle [radius=0.1000000000];
\draw[fill] (-1.6666666667, -1.9245008973) circle [radius=0.1000000000];
\draw[fill] (-4.1666666667, -2.4056261216) circle [radius=0.1000000000];
\draw[fill] (-3.3333333333, -1.9245008973) circle [radius=0.1000000000];
\draw[fill] (3.3333333333, -0.9622504486) circle [radius=0.1000000000];
\draw[fill] (1.6666666667, -0.9622504486) circle [radius=0.1000000000];
\draw[fill] (2.5000000000, -0.4811252243) circle [radius=0.1000000000];
\draw[fill] (0.0000000000, -0.9622504486) circle [radius=0.1000000000];
\draw[fill] (0.8333333333, -0.4811252243) circle [radius=0.1000000000];
\draw[fill] (-1.6666666667, -0.9622504486) circle [radius=0.1000000000];
\draw[fill] (-0.8333333333, -0.4811252243) circle [radius=0.1000000000];
\draw[fill] (-3.3333333333, -0.9622504486) circle [radius=0.1000000000];
\draw[fill] (-2.5000000000, -0.4811252243) circle [radius=0.1000000000];
\draw[fill] (2.5000000000, 0.4811252243) circle [radius=0.1000000000];
\draw[fill] (0.8333333333, 0.4811252243) circle [radius=0.1000000000];
\draw[fill] (1.6666666667, 0.9622504486) circle [radius=0.1000000000];
\draw[fill] (-0.8333333333, 0.4811252243) circle [radius=0.1000000000];
\draw[fill] (0.0000000000, 0.9622504486) circle [radius=0.1000000000];
\draw[fill] (-2.5000000000, 0.4811252243) circle [radius=0.1000000000];
\draw[fill] (-1.6666666667, 0.9622504486) circle [radius=0.1000000000];
\draw[fill] (1.6666666667, 1.9245008973) circle [radius=0.1000000000];
\draw[fill] (0.0000000000, 1.9245008973) circle [radius=0.1000000000];
\draw[fill] (0.8333333333, 2.4056261216) circle [radius=0.1000000000];
\draw[fill] (-1.6666666667, 1.9245008973) circle [radius=0.1000000000];
\draw[fill] (-0.8333333333, 2.4056261216) circle [radius=0.1000000000];
\draw[fill] (0.8333333333, 3.3678765703) circle [radius=0.1000000000];
\draw[fill] (-0.8333333333, 3.3678765703) circle [radius=0.1000000000];
\draw[fill] (0.0000000000, 3.8490017946) circle [radius=0.1000000000];
\draw[fill] (0.0000000000, 4.8112522432) circle [radius=0.1000000000];
\draw [color=black, line width=2] (4.1666666667, -2.4056261216) -- (0.0000000000, -5.7735026919);
\draw [color=black, line width=2] (4.1666666667, -2.4056261216) -- (5.0000000000, 2.8867513459);
\draw [color=black, line width=2] (2.5000000000, -2.4056261216) -- (0.0000000000, -5.7735026919);
\draw [color=black, line width=2] (4.1666666667, -2.4056261216) -- (3.3333333333, -1.9245008973);
\draw [color=black, line width=2] (2.5000000000, -2.4056261216) -- (3.3333333333, -1.9245008973);
\draw [color=black, line width=2] (0.8333333333, -2.4056261216) -- (0.0000000000, -5.7735026919);
\draw [color=black, line width=2] (2.5000000000, -2.4056261216) -- (1.6666666667, -1.9245008973);
\draw [color=black, line width=2] (0.8333333333, -2.4056261216) -- (1.6666666667, -1.9245008973);
\draw [color=black, line width=2] (-0.8333333333, -2.4056261216) -- (0.0000000000, -5.7735026919);
\draw [color=black, line width=2] (0.8333333333, -2.4056261216) -- (0.0000000000, -1.9245008973);
\draw [color=black, line width=2] (-0.8333333333, -2.4056261216) -- (0.0000000000, -1.9245008973);
\draw [color=black, line width=2] (-2.5000000000, -2.4056261216) -- (0.0000000000, -5.7735026919);
\draw [color=black, line width=2] (-0.8333333333, -2.4056261216) -- (-1.6666666667, -1.9245008973);
\draw [color=black, line width=2] (-2.5000000000, -2.4056261216) -- (-1.6666666667, -1.9245008973);
\draw [color=black, line width=2] (-4.1666666667, -2.4056261216) -- (0.0000000000, -5.7735026919);
\draw [color=black, line width=2] (-2.5000000000, -2.4056261216) -- (-3.3333333333, -1.9245008973);
\draw [color=black, line width=2] (-4.1666666667, -2.4056261216) -- (-3.3333333333, -1.9245008973);
\draw [color=black, line width=2] (-4.1666666667, -2.4056261216) -- (-5.0000000000, 2.8867513459);
\draw [color=black, line width=2] (3.3333333333, -1.9245008973) -- (3.3333333333, -0.9622504486);
\draw [color=black, line width=2] (3.3333333333, -0.9622504486) -- (5.0000000000, 2.8867513459);
\draw [color=black, line width=2] (1.6666666667, -1.9245008973) -- (1.6666666667, -0.9622504486);
\draw [color=black, line width=2] (3.3333333333, -0.9622504486) -- (2.5000000000, -0.4811252243);
\draw [color=black, line width=2] (1.6666666667, -0.9622504486) -- (2.5000000000, -0.4811252243);
\draw [color=black, line width=2] (0.0000000000, -1.9245008973) -- (0.0000000000, -0.9622504486);
\draw [color=black, line width=2] (1.6666666667, -0.9622504486) -- (0.8333333333, -0.4811252243);
\draw [color=black, line width=2] (0.0000000000, -0.9622504486) -- (0.8333333333, -0.4811252243);
\draw [color=black, line width=2] (-1.6666666667, -1.9245008973) -- (-1.6666666667, -0.9622504486);
\draw [color=black, line width=2] (0.0000000000, -0.9622504486) -- (-0.8333333333, -0.4811252243);
\draw [color=black, line width=2] (-1.6666666667, -0.9622504486) -- (-0.8333333333, -0.4811252243);
\draw [color=black, line width=2] (-3.3333333333, -1.9245008973) -- (-3.3333333333, -0.9622504486);
\draw [color=black, line width=2] (-1.6666666667, -0.9622504486) -- (-2.5000000000, -0.4811252243);
\draw [color=black, line width=2] (-3.3333333333, -0.9622504486) -- (-2.5000000000, -0.4811252243);
\draw [color=black, line width=2] (-3.3333333333, -0.9622504486) -- (-5.0000000000, 2.8867513459);
\draw [color=black, line width=2] (2.5000000000, -0.4811252243) -- (2.5000000000, 0.4811252243);
\draw [color=black, line width=2] (2.5000000000, 0.4811252243) -- (5.0000000000, 2.8867513459);
\draw [color=black, line width=2] (0.8333333333, -0.4811252243) -- (0.8333333333, 0.4811252243);
\draw [color=black, line width=2] (2.5000000000, 0.4811252243) -- (1.6666666667, 0.9622504486);
\draw [color=black, line width=2] (0.8333333333, 0.4811252243) -- (1.6666666667, 0.9622504486);
\draw [color=black, line width=2] (-0.8333333333, -0.4811252243) -- (-0.8333333333, 0.4811252243);
\draw [color=black, line width=2] (0.8333333333, 0.4811252243) -- (0.0000000000, 0.9622504486);
\draw [color=black, line width=2] (-0.8333333333, 0.4811252243) -- (0.0000000000, 0.9622504486);
\draw [color=black, line width=2] (-2.5000000000, -0.4811252243) -- (-2.5000000000, 0.4811252243);
\draw [color=black, line width=2] (-0.8333333333, 0.4811252243) -- (-1.6666666667, 0.9622504486);
\draw [color=black, line width=2] (-2.5000000000, 0.4811252243) -- (-1.6666666667, 0.9622504486);
\draw [color=black, line width=2] (-2.5000000000, 0.4811252243) -- (-5.0000000000, 2.8867513459);
\draw [color=black, line width=2] (1.6666666667, 0.9622504486) -- (1.6666666667, 1.9245008973);
\draw [color=black, line width=2] (1.6666666667, 1.9245008973) -- (5.0000000000, 2.8867513459);
\draw [color=black, line width=2] (0.0000000000, 0.9622504486) -- (0.0000000000, 1.9245008973);
\draw [color=black, line width=2] (1.6666666667, 1.9245008973) -- (0.8333333333, 2.4056261216);
\draw [color=black, line width=2] (0.0000000000, 1.9245008973) -- (0.8333333333, 2.4056261216);
\draw [color=black, line width=2] (-1.6666666667, 0.9622504486) -- (-1.6666666667, 1.9245008973);
\draw [color=black, line width=2] (0.0000000000, 1.9245008973) -- (-0.8333333333, 2.4056261216);
\draw [color=black, line width=2] (-1.6666666667, 1.9245008973) -- (-0.8333333333, 2.4056261216);
\draw [color=black, line width=2] (-1.6666666667, 1.9245008973) -- (-5.0000000000, 2.8867513459);
\draw [color=black, line width=2] (0.8333333333, 2.4056261216) -- (0.8333333333, 3.3678765703);
\draw [color=black, line width=2] (0.8333333333, 3.3678765703) -- (5.0000000000, 2.8867513459);
\draw [color=black, line width=2] (-0.8333333333, 2.4056261216) -- (-0.8333333333, 3.3678765703);
\draw [color=black, line width=2] (0.8333333333, 3.3678765703) -- (0.0000000000, 3.8490017946);
\draw [color=black, line width=2] (-0.8333333333, 3.3678765703) -- (0.0000000000, 3.8490017946);
\draw [color=black, line width=2] (-0.8333333333, 3.3678765703) -- (-5.0000000000, 2.8867513459);
\draw [color=black, line width=2] (0.0000000000, 3.8490017946) -- (0.0000000000, 4.8112522432);
\draw [color=black, line width=2] (0.0000000000, 4.8112522432) -- (5.0000000000, 2.8867513459);
\draw [color=black, line width=2] (0.0000000000, 4.8112522432) -- (-5.0000000000, 2.8867513459);
\end{tikzpicture}

%% file: ball.tikz
\begin{tikzpicture}
\draw [color=black] (5.0000000000, -2.8867513459) -- (3.3333333333, -2.8867513459);
\draw [color=black] (5.0000000000, -2.8867513459) -- (4.1666666667, -1.4433756730);
\draw [color=black] (3.3333333333, -2.8867513459) -- (1.6666666667, -2.8867513459);
\draw [color=black] (3.3333333333, -2.8867513459) -- (4.1666666667, -1.4433756730);
\draw [color=black] (3.3333333333, -2.8867513459) -- (2.5000000000, -1.4433756730);
\draw [color=black] (1.6666666667, -2.8867513459) -- (0.0000000000, -2.8867513459);
\draw [color=black] (1.6666666667, -2.8867513459) -- (2.5000000000, -1.4433756730);
\draw [color=black] (1.6666666667, -2.8867513459) -- (0.8333333333, -1.4433756730);
\draw [color=black] (0.0000000000, -2.8867513459) -- (-1.6666666667, -2.8867513459);
\draw [color=black] (0.0000000000, -2.8867513459) -- (0.8333333333, -1.4433756730);
\draw [color=black] (0.0000000000, -2.8867513459) -- (-0.8333333333, -1.4433756730);
\draw [color=black] (-1.6666666667, -2.8867513459) -- (-3.3333333333, -2.8867513459);
\draw [color=black] (-1.6666666667, -2.8867513459) -- (-0.8333333333, -1.4433756730);
\draw [color=black] (-1.6666666667, -2.8867513459) -- (-2.5000000000, -1.4433756730);
\draw [color=black] (-3.3333333333, -2.8867513459) -- (-5.0000000000, -2.8867513459);
\draw [color=black] (-3.3333333333, -2.8867513459) -- (-2.5000000000, -1.4433756730);
\draw [color=black] (-3.3333333333, -2.8867513459) -- (-4.1666666667, -1.4433756730);
\draw [color=black] (-5.0000000000, -2.8867513459) -- (-4.1666666667, -1.4433756730);
\draw [color=black] (4.1666666667, -1.4433756730) -- (2.5000000000, -1.4433756730);
\draw [color=black] (4.1666666667, -1.4433756730) -- (3.3333333333, -0.0000000000);
\draw [color=black] (2.5000000000, -1.4433756730) -- (0.8333333333, -1.4433756730);
\draw [color=black] (2.5000000000, -1.4433756730) -- (3.3333333333, -0.0000000000);
\draw [color=black] (2.5000000000, -1.4433756730) -- (1.6666666667, -0.0000000000);
\draw [color=black] (0.8333333333, -1.4433756730) -- (-0.8333333333, -1.4433756730);
\draw [color=black] (0.8333333333, -1.4433756730) -- (1.6666666667, -0.0000000000);
\draw [color=black] (0.8333333333, -1.4433756730) -- (0.0000000000, -0.0000000000);
\draw [color=black] (-0.8333333333, -1.4433756730) -- (-2.5000000000, -1.4433756730);
\draw [color=black] (-0.8333333333, -1.4433756730) -- (0.0000000000, -0.0000000000);
\draw [color=black] (-0.8333333333, -1.4433756730) -- (-1.6666666667, -0.0000000000);
\draw [color=black] (-2.5000000000, -1.4433756730) -- (-4.1666666667, -1.4433756730);
\draw [color=black] (-2.5000000000, -1.4433756730) -- (-1.6666666667, -0.0000000000);
\draw [color=black] (-2.5000000000, -1.4433756730) -- (-3.3333333333, -0.0000000000);
\draw [color=black] (-4.1666666667, -1.4433756730) -- (-3.3333333333, -0.0000000000);
\draw [color=black] (3.3333333333, -0.0000000000) -- (1.6666666667, -0.0000000000);
\draw [color=black] (3.3333333333, -0.0000000000) -- (2.5000000000, 1.4433756730);
\draw [color=black] (1.6666666667, -0.0000000000) -- (0.0000000000, -0.0000000000);
\draw [color=black] (1.6666666667, -0.0000000000) -- (2.5000000000, 1.4433756730);
\draw [color=black] (1.6666666667, -0.0000000000) -- (0.8333333333, 1.4433756730);
\draw [color=black] (0.0000000000, -0.0000000000) -- (-1.6666666667, -0.0000000000);
\draw [color=black] (0.0000000000, -0.0000000000) -- (0.8333333333, 1.4433756730);
\draw [color=black] (0.0000000000, -0.0000000000) -- (-0.8333333333, 1.4433756730);
\draw [color=black] (-1.6666666667, -0.0000000000) -- (-3.3333333333, -0.0000000000);
\draw [color=black] (-1.6666666667, -0.0000000000) -- (-0.8333333333, 1.4433756730);
\draw [color=black] (-1.6666666667, -0.0000000000) -- (-2.5000000000, 1.4433756730);
\draw [color=black] (-3.3333333333, -0.0000000000) -- (-2.5000000000, 1.4433756730);
\draw [color=black] (2.5000000000, 1.4433756730) -- (0.8333333333, 1.4433756730);
\draw [color=black] (2.5000000000, 1.4433756730) -- (1.6666666667, 2.8867513459);
\draw [color=black] (0.8333333333, 1.4433756730) -- (-0.8333333333, 1.4433756730);
\draw [color=black] (0.8333333333, 1.4433756730) -- (1.6666666667, 2.8867513459);
\draw [color=black] (0.8333333333, 1.4433756730) -- (0.0000000000, 2.8867513459);
\draw [color=black] (-0.8333333333, 1.4433756730) -- (-2.5000000000, 1.4433756730);
\draw [color=black] (-0.8333333333, 1.4433756730) -- (0.0000000000, 2.8867513459);
\draw [color=black] (-0.8333333333, 1.4433756730) -- (-1.6666666667, 2.8867513459);
\draw [color=black] (-2.5000000000, 1.4433756730) -- (-1.6666666667, 2.8867513459);
\draw [color=black] (1.6666666667, 2.8867513459) -- (0.0000000000, 2.8867513459);
\draw [color=black] (1.6666666667, 2.8867513459) -- (0.8333333333, 4.3301270189);
\draw [color=black] (0.0000000000, 2.8867513459) -- (-1.6666666667, 2.8867513459);
\draw [color=black] (0.0000000000, 2.8867513459) -- (0.8333333333, 4.3301270189);
\draw [color=black] (0.0000000000, 2.8867513459) -- (-0.8333333333, 4.3301270189);
\draw [color=black] (-1.6666666667, 2.8867513459) -- (-0.8333333333, 4.3301270189);
\draw [color=black] (0.8333333333, 4.3301270189) -- (-0.8333333333, 4.3301270189);
\draw [color=black] (0.8333333333, 4.3301270189) -- (0.0000000000, 5.7735026919);
\draw [color=black] (-0.8333333333, 4.3301270189) -- (0.0000000000, 5.7735026919);
\draw [color=black] (0, 5.7735026919) -- (0.0, 7.50555349947);
\draw [color=black] (-5.0, -2.88675134595) -- (-6.5, -3.75277674973);
\draw [color=black] (5.0, -2.88675134595) -- (6.5, -3.75277674973);
\draw[fill] (0.0000000000, -5.7735026919) circle [radius=0.1000000000];
\node at (0.0,-6.81273317644) {\Huge $O_1$};
\draw[fill] (5.0000000000, 2.8867513459) circle [radius=0.1000000000];
\node at (5.9,3.40636658822) {\Huge $O_2$};
\draw[fill] (-5.0000000000, 2.8867513459) circle [radius=0.1000000000];
\node at (-5.9,3.40636658822) {\Huge $O_3$};
\draw[fill] (4.1666666667, -2.4056261216) circle [radius=0.1000000000];
\draw[fill] (2.5000000000, -2.4056261216) circle [radius=0.1000000000];
\draw[fill] (3.3333333333, -1.9245008973) circle [radius=0.1000000000];
\draw[fill] (0.8333333333, -2.4056261216) circle [radius=0.1000000000];
\draw[fill] (1.6666666667, -1.9245008973) circle [radius=0.1000000000];
\draw[fill] (-0.8333333333, -2.4056261216) circle [radius=0.1000000000];
\draw[fill] (0.0000000000, -1.9245008973) circle [radius=0.1000000000];
\draw[fill] (-2.5000000000, -2.4056261216) circle [radius=0.1000000000];
\draw[fill] (-1.6666666667, -1.9245008973) circle [radius=0.1000000000];
\draw[fill] (-4.1666666667, -2.4056261216) circle [radius=0.1000000000];
\draw[fill] (-3.3333333333, -1.9245008973) circle [radius=0.1000000000];
\draw[fill] (3.3333333333, -0.9622504486) circle [radius=0.1000000000];
\draw[fill] (1.6666666667, -0.9622504486) circle [radius=0.1000000000];
\draw[fill] (2.5000000000, -0.4811252243) circle [radius=0.1000000000];
\draw[fill] (0.0000000000, -0.9622504486) circle [radius=0.1000000000];
\draw[fill] (0.8333333333, -0.4811252243) circle [radius=0.1000000000];
\draw[fill] (-1.6666666667, -0.9622504486) circle [radius=0.1000000000];
\draw[fill] (-0.8333333333, -0.4811252243) circle [radius=0.1000000000];
\draw[fill] (-3.3333333333, -0.9622504486) circle [radius=0.1000000000];
\draw[fill] (-2.5000000000, -0.4811252243) circle [radius=0.1000000000];
\draw[fill] (2.5000000000, 0.4811252243) circle [radius=0.1000000000];
\draw[fill] (0.8333333333, 0.4811252243) circle [radius=0.1000000000];
\draw[fill] (1.6666666667, 0.9622504486) circle [radius=0.1000000000];
\draw[fill] (-0.8333333333, 0.4811252243) circle [radius=0.1000000000];
\draw[fill] (0.0000000000, 0.9622504486) circle [radius=0.1000000000];
\draw[fill] (-2.5000000000, 0.4811252243) circle [radius=0.1000000000];
\draw[fill] (-1.6666666667, 0.9622504486) circle [radius=0.1000000000];
\draw[fill] (1.6666666667, 1.9245008973) circle [radius=0.1000000000];
\draw[fill] (0.0000000000, 1.9245008973) circle [radius=0.1000000000];
\draw[fill] (0.8333333333, 2.4056261216) circle [radius=0.1000000000];
\draw[fill] (-1.6666666667, 1.9245008973) circle [radius=0.1000000000];
\draw[fill] (-0.8333333333, 2.4056261216) circle [radius=0.1000000000];
\draw[fill] (0.8333333333, 3.3678765703) circle [radius=0.1000000000];
\draw[fill] (-0.8333333333, 3.3678765703) circle [radius=0.1000000000];
\draw[fill] (0.0000000000, 3.8490017946) circle [radius=0.1000000000];
\draw[fill] (0.0000000000, 4.8112522432) circle [radius=0.1000000000];
\draw [color=black, line width=4] (0.0000000000, -5.7735026919) -- (0.8333333333, -2.4056261216);
\draw [color=black, line width=4] (0.8333333333, -2.4056261216) -- (1.6666666667, -1.9245008973);
\draw [color=black, line width=4] (1.6666666667, -1.9245008973) -- (1.6666666667, -0.9622504486);
\draw [color=black, line width=4] (1.6666666667, -0.9622504486) -- (0.8333333333, -0.4811252243);
\draw [color=black, line width=4] (0.8333333333, -0.4811252243) -- (0.8333333333, 0.4811252243);
\draw [color=black, line width=4] (5.0000000000, 2.8867513459) -- (2.5000000000, 0.4811252243);
\draw [color=black, line width=4] (2.5000000000, 0.4811252243) -- (1.6666666667, 0.9622504486);
\draw [color=black, line width=4] (1.6666666667, 0.9622504486) -- (0.8333333333, 0.4811252243);
\draw [color=black, line width=4] (-5.0000000000, 2.8867513459) -- (-2.5000000000, 0.4811252243);
\draw [color=black, line width=4] (-2.5000000000, 0.4811252243) -- (-1.6666666667, 0.9622504486);
\draw [color=black, line width=4] (-1.6666666667, 0.9622504486) -- (-0.8333333333, 0.4811252243);
\draw [color=black, line width=4] (-0.8333333333, 0.4811252243) -- (0.0000000000, 0.9622504486);
\draw [color=black, line width=4] (0.0000000000, 0.9622504486) -- (0.8333333333, 0.4811252243);
\end{tikzpicture}

%% file: 2-corner.tikz
\begin{tikzpicture}
\draw [color=black] (5.0000000000, -2.8867513459) -- (3.3333333333, -2.8867513459);
\draw [color=black] (5.0000000000, -2.8867513459) -- (4.1666666667, -1.4433756730);
\draw [color=black] (3.3333333333, -2.8867513459) -- (1.6666666667, -2.8867513459);
\draw [color=black] (3.3333333333, -2.8867513459) -- (4.1666666667, -1.4433756730);
\draw [color=black] (3.3333333333, -2.8867513459) -- (2.5000000000, -1.4433756730);
\draw [color=black] (1.6666666667, -2.8867513459) -- (0.0000000000, -2.8867513459);
\draw [color=black] (1.6666666667, -2.8867513459) -- (2.5000000000, -1.4433756730);
\draw [color=black] (1.6666666667, -2.8867513459) -- (0.8333333333, -1.4433756730);
\draw [color=black] (0.0000000000, -2.8867513459) -- (-1.6666666667, -2.8867513459);
\draw [color=black] (0.0000000000, -2.8867513459) -- (0.8333333333, -1.4433756730);
\draw [color=black] (0.0000000000, -2.8867513459) -- (-0.8333333333, -1.4433756730);
\draw [color=black] (-1.6666666667, -2.8867513459) -- (-3.3333333333, -2.8867513459);
\draw [color=black] (-1.6666666667, -2.8867513459) -- (-0.8333333333, -1.4433756730);
\draw [color=black] (-1.6666666667, -2.8867513459) -- (-2.5000000000, -1.4433756730);
\draw [color=black] (-3.3333333333, -2.8867513459) -- (-5.0000000000, -2.8867513459);
\draw [color=black] (-3.3333333333, -2.8867513459) -- (-2.5000000000, -1.4433756730);
\draw [color=black] (-3.3333333333, -2.8867513459) -- (-4.1666666667, -1.4433756730);
\draw [color=black] (-5.0000000000, -2.8867513459) -- (-4.1666666667, -1.4433756730);
\draw [color=black] (4.1666666667, -1.4433756730) -- (2.5000000000, -1.4433756730);
\draw [color=black] (4.1666666667, -1.4433756730) -- (3.3333333333, -0.0000000000);
\draw [color=black] (2.5000000000, -1.4433756730) -- (0.8333333333, -1.4433756730);
\draw [color=black] (2.5000000000, -1.4433756730) -- (3.3333333333, -0.0000000000);
\draw [color=black] (2.5000000000, -1.4433756730) -- (1.6666666667, -0.0000000000);
\draw [color=black] (0.8333333333, -1.4433756730) -- (-0.8333333333, -1.4433756730);
\draw [color=black] (0.8333333333, -1.4433756730) -- (1.6666666667, -0.0000000000);
\draw [color=black] (0.8333333333, -1.4433756730) -- (0.0000000000, -0.0000000000);
\draw [color=black] (-0.8333333333, -1.4433756730) -- (-2.5000000000, -1.4433756730);
\draw [color=black] (-0.8333333333, -1.4433756730) -- (0.0000000000, -0.0000000000);
\draw [color=black] (-0.8333333333, -1.4433756730) -- (-1.6666666667, -0.0000000000);
\draw [color=black] (-2.5000000000, -1.4433756730) -- (-4.1666666667, -1.4433756730);
\draw [color=black] (-2.5000000000, -1.4433756730) -- (-1.6666666667, -0.0000000000);
\draw [color=black] (-2.5000000000, -1.4433756730) -- (-3.3333333333, -0.0000000000);
\draw [color=black] (-4.1666666667, -1.4433756730) -- (-3.3333333333, -0.0000000000);
\draw [color=black] (3.3333333333, -0.0000000000) -- (1.6666666667, -0.0000000000);
\draw [color=black] (3.3333333333, -0.0000000000) -- (2.5000000000, 1.4433756730);
\draw [color=black] (1.6666666667, -0.0000000000) -- (0.0000000000, -0.0000000000);
\draw [color=black] (1.6666666667, -0.0000000000) -- (2.5000000000, 1.4433756730);
\draw [color=black] (1.6666666667, -0.0000000000) -- (0.8333333333, 1.4433756730);
\draw [color=black] (0.0000000000, -0.0000000000) -- (-1.6666666667, -0.0000000000);
\draw [color=black] (0.0000000000, -0.0000000000) -- (0.8333333333, 1.4433756730);
\draw [color=black] (0.0000000000, -0.0000000000) -- (-0.8333333333, 1.4433756730);
\draw [color=black] (-1.6666666667, -0.0000000000) -- (-3.3333333333, -0.0000000000);
\draw [color=black] (-1.6666666667, -0.0000000000) -- (-0.8333333333, 1.4433756730);
\draw [color=black] (-1.6666666667, -0.0000000000) -- (-2.5000000000, 1.4433756730);
\draw [color=black] (-3.3333333333, -0.0000000000) -- (-2.5000000000, 1.4433756730);
\draw [color=black] (2.5000000000, 1.4433756730) -- (0.8333333333, 1.4433756730);
\draw [color=black] (2.5000000000, 1.4433756730) -- (1.6666666667, 2.8867513459);
\draw [color=black] (0.8333333333, 1.4433756730) -- (-0.8333333333, 1.4433756730);
\draw [color=black] (0.8333333333, 1.4433756730) -- (1.6666666667, 2.8867513459);
\draw [color=black] (0.8333333333, 1.4433756730) -- (0.0000000000, 2.8867513459);
\draw [color=black] (-0.8333333333, 1.4433756730) -- (-2.5000000000, 1.4433756730);
\draw [color=black] (-0.8333333333, 1.4433756730) -- (0.0000000000, 2.8867513459);
\draw [color=black] (-0.8333333333, 1.4433756730) -- (-1.6666666667, 2.8867513459);
\draw [color=black] (-2.5000000000, 1.4433756730) -- (-1.6666666667, 2.8867513459);
\draw [color=black] (1.6666666667, 2.8867513459) -- (0.0000000000, 2.8867513459);
\draw [color=black] (1.6666666667, 2.8867513459) -- (0.8333333333, 4.3301270189);
\draw [color=black] (0.0000000000, 2.8867513459) -- (-1.6666666667, 2.8867513459);
\draw [color=black] (0.0000000000, 2.8867513459) -- (0.8333333333, 4.3301270189);
\draw [color=black] (0.0000000000, 2.8867513459) -- (-0.8333333333, 4.3301270189);
\draw [color=black] (-1.6666666667, 2.8867513459) -- (-0.8333333333, 4.3301270189);
\draw [color=black] (0.8333333333, 4.3301270189) -- (-0.8333333333, 4.3301270189);
\draw [color=black] (0.8333333333, 4.3301270189) -- (0.0000000000, 5.7735026919);
\draw [color=black] (-0.8333333333, 4.3301270189) -- (0.0000000000, 5.7735026919);
\draw [color=black] (0, 5.7735026919) -- (0.0, 7.50555349947);
\draw [color=black] (-5.0, -2.88675134595) -- (-6.5, -3.75277674973);
\draw [color=black] (5.0, -2.88675134595) -- (6.5, -3.75277674973);
\draw[fill] (0.0000000000, -5.7735026919) circle [radius=0.1000000000];
\node at (0.0,-6.81273317644) {\Huge $O_1$};
\draw[fill] (5.0000000000, 2.8867513459) circle [radius=0.1000000000];
\node at (5.9,3.40636658822) {\Huge $O_2$};
\draw[fill] (-5.0000000000, 2.8867513459) circle [radius=0.1000000000];
\node at (-5.9,3.40636658822) {\Huge $O_3$};
\draw[fill] (4.1666666667, -2.4056261216) circle [radius=0.1000000000];
\draw[fill] (2.5000000000, -2.4056261216) circle [radius=0.1000000000];
\draw[fill] (3.3333333333, -1.9245008973) circle [radius=0.1000000000];
\draw[fill] (0.8333333333, -2.4056261216) circle [radius=0.1000000000];
\draw[fill] (1.6666666667, -1.9245008973) circle [radius=0.1000000000];
\draw[fill] (-0.8333333333, -2.4056261216) circle [radius=0.1000000000];
\draw[fill] (0.0000000000, -1.9245008973) circle [radius=0.1000000000];
\draw[fill] (-2.5000000000, -2.4056261216) circle [radius=0.1000000000];
\draw[fill] (-1.6666666667, -1.9245008973) circle [radius=0.1000000000];
\draw[fill] (-4.1666666667, -2.4056261216) circle [radius=0.1000000000];
\draw[fill] (-3.3333333333, -1.9245008973) circle [radius=0.1000000000];
\draw[fill] (3.3333333333, -0.9622504486) circle [radius=0.1000000000];
\draw[fill] (1.6666666667, -0.9622504486) circle [radius=0.1000000000];
\draw[fill] (2.5000000000, -0.4811252243) circle [radius=0.1000000000];
\draw[fill] (0.0000000000, -0.9622504486) circle [radius=0.1000000000];
\draw[fill] (0.8333333333, -0.4811252243) circle [radius=0.1000000000];
\draw[fill] (-1.6666666667, -0.9622504486) circle [radius=0.1000000000];
\draw[fill] (-0.8333333333, -0.4811252243) circle [radius=0.1000000000];
\draw[fill] (-3.3333333333, -0.9622504486) circle [radius=0.1000000000];
\draw[fill] (-2.5000000000, -0.4811252243) circle [radius=0.1000000000];
\draw[fill] (2.5000000000, 0.4811252243) circle [radius=0.1000000000];
\draw[fill] (0.8333333333, 0.4811252243) circle [radius=0.1000000000];
\draw[fill] (1.6666666667, 0.9622504486) circle [radius=0.1000000000];
\draw[fill] (-0.8333333333, 0.4811252243) circle [radius=0.1000000000];
\draw[fill] (0.0000000000, 0.9622504486) circle [radius=0.1000000000];
\draw[fill] (-2.5000000000, 0.4811252243) circle [radius=0.1000000000];
\draw[fill] (-1.6666666667, 0.9622504486) circle [radius=0.1000000000];
\draw[fill] (1.6666666667, 1.9245008973) circle [radius=0.1000000000];
\draw[fill] (0.0000000000, 1.9245008973) circle [radius=0.1000000000];
\draw[fill] (0.8333333333, 2.4056261216) circle [radius=0.1000000000];
\draw[fill] (-1.6666666667, 1.9245008973) circle [radius=0.1000000000];
\draw[fill] (-0.8333333333, 2.4056261216) circle [radius=0.1000000000];
\draw[fill] (0.8333333333, 3.3678765703) circle [radius=0.1000000000];
\draw[fill] (-0.8333333333, 3.3678765703) circle [radius=0.1000000000];
\draw[fill] (0.0000000000, 3.8490017946) circle [radius=0.1000000000];
\draw[fill] (0.0000000000, 4.8112522432) circle [radius=0.1000000000];
\draw [color=black, line width=4] (0.0000000000, -5.7735026919) -- (0.8333333333, -2.4056261216);
\draw [color=black, line width=4] (0.8333333333, -2.4056261216) -- (1.6666666667, -1.9245008973);
\draw [color=black, line width=4] (1.6666666667, -1.9245008973) -- (1.6666666667, -0.9622504486);
\draw [color=black, line width=4] (1.6666666667, -0.9622504486) -- (0.8333333333, -0.4811252243);
\draw [color=black, line width=4] (0.8333333333, -0.4811252243) -- (0.8333333333, 0.4811252243);
\draw [color=black, line width=4] (5.0000000000, 2.8867513459) -- (2.5000000000, 0.4811252243);
\draw [color=black, line width=4] (2.5000000000, 0.4811252243) -- (1.6666666667, 0.9622504486);
\draw [color=black, line width=4] (1.6666666667, 0.9622504486) -- (0.8333333333, 0.4811252243);
\draw [color=black, line width=4] (0.0000000000, -5.7735026919) -- (-2.5000000000, -2.4056261216);
\draw [color=black, line width=4] (-2.5000000000, -2.4056261216) -- (-3.3333333333, -1.9245008973);
\draw [color=black, line width=4] (-3.3333333333, -1.9245008973) -- (-3.3333333333, -0.9622504486);
\draw [color=black, line width=4] (-5.0000000000, 2.8867513459) -- (-3.3333333333, -0.9622504486);
\end{tikzpicture}

%% file: 3-corner.tikz
\begin{tikzpicture}
\draw [color=black] (5.0000000000, -2.8867513459) -- (3.3333333333, -2.8867513459);
\draw [color=black] (5.0000000000, -2.8867513459) -- (4.1666666667, -1.4433756730);
\draw [color=black] (3.3333333333, -2.8867513459) -- (1.6666666667, -2.8867513459);
\draw [color=black] (3.3333333333, -2.8867513459) -- (4.1666666667, -1.4433756730);
\draw [color=black] (3.3333333333, -2.8867513459) -- (2.5000000000, -1.4433756730);
\draw [color=black] (1.6666666667, -2.8867513459) -- (0.0000000000, -2.8867513459);
\draw [color=black] (1.6666666667, -2.8867513459) -- (2.5000000000, -1.4433756730);
\draw [color=black] (1.6666666667, -2.8867513459) -- (0.8333333333, -1.4433756730);
\draw [color=black] (0.0000000000, -2.8867513459) -- (-1.6666666667, -2.8867513459);
\draw [color=black] (0.0000000000, -2.8867513459) -- (0.8333333333, -1.4433756730);
\draw [color=black] (0.0000000000, -2.8867513459) -- (-0.8333333333, -1.4433756730);
\draw [color=black] (-1.6666666667, -2.8867513459) -- (-3.3333333333, -2.8867513459);
\draw [color=black] (-1.6666666667, -2.8867513459) -- (-0.8333333333, -1.4433756730);
\draw [color=black] (-1.6666666667, -2.8867513459) -- (-2.5000000000, -1.4433756730);
\draw [color=black] (-3.3333333333, -2.8867513459) -- (-5.0000000000, -2.8867513459);
\draw [color=black] (-3.3333333333, -2.8867513459) -- (-2.5000000000, -1.4433756730);
\draw [color=black] (-3.3333333333, -2.8867513459) -- (-4.1666666667, -1.4433756730);
\draw [color=black] (-5.0000000000, -2.8867513459) -- (-4.1666666667, -1.4433756730);
\draw [color=black] (4.1666666667, -1.4433756730) -- (2.5000000000, -1.4433756730);
\draw [color=black] (4.1666666667, -1.4433756730) -- (3.3333333333, -0.0000000000);
\draw [color=black] (2.5000000000, -1.4433756730) -- (0.8333333333, -1.4433756730);
\draw [color=black] (2.5000000000, -1.4433756730) -- (3.3333333333, -0.0000000000);
\draw [color=black] (2.5000000000, -1.4433756730) -- (1.6666666667, -0.0000000000);
\draw [color=black] (0.8333333333, -1.4433756730) -- (-0.8333333333, -1.4433756730);
\draw [color=black] (0.8333333333, -1.4433756730) -- (1.6666666667, -0.0000000000);
\draw [color=black] (0.8333333333, -1.4433756730) -- (0.0000000000, -0.0000000000);
\draw [color=black] (-0.8333333333, -1.4433756730) -- (-2.5000000000, -1.4433756730);
\draw [color=black] (-0.8333333333, -1.4433756730) -- (0.0000000000, -0.0000000000);
\draw [color=black] (-0.8333333333, -1.4433756730) -- (-1.6666666667, -0.0000000000);
\draw [color=black] (-2.5000000000, -1.4433756730) -- (-4.1666666667, -1.4433756730);
\draw [color=black] (-2.5000000000, -1.4433756730) -- (-1.6666666667, -0.0000000000);
\draw [color=black] (-2.5000000000, -1.4433756730) -- (-3.3333333333, -0.0000000000);
\draw [color=black] (-4.1666666667, -1.4433756730) -- (-3.3333333333, -0.0000000000);
\draw [color=black] (3.3333333333, -0.0000000000) -- (1.6666666667, -0.0000000000);
\draw [color=black] (3.3333333333, -0.0000000000) -- (2.5000000000, 1.4433756730);
\draw [color=black] (1.6666666667, -0.0000000000) -- (0.0000000000, -0.0000000000);
\draw [color=black] (1.6666666667, -0.0000000000) -- (2.5000000000, 1.4433756730);
\draw [color=black] (1.6666666667, -0.0000000000) -- (0.8333333333, 1.4433756730);
\draw [color=black] (0.0000000000, -0.0000000000) -- (-1.6666666667, -0.0000000000);
\draw [color=black] (0.0000000000, -0.0000000000) -- (0.8333333333, 1.4433756730);
\draw [color=black] (0.0000000000, -0.0000000000) -- (-0.8333333333, 1.4433756730);
\draw [color=black] (-1.6666666667, -0.0000000000) -- (-3.3333333333, -0.0000000000);
\draw [color=black] (-1.6666666667, -0.0000000000) -- (-0.8333333333, 1.4433756730);
\draw [color=black] (-1.6666666667, -0.0000000000) -- (-2.5000000000, 1.4433756730);
\draw [color=black] (-3.3333333333, -0.0000000000) -- (-2.5000000000, 1.4433756730);
\draw [color=black] (2.5000000000, 1.4433756730) -- (0.8333333333, 1.4433756730);
\draw [color=black] (2.5000000000, 1.4433756730) -- (1.6666666667, 2.8867513459);
\draw [color=black] (0.8333333333, 1.4433756730) -- (-0.8333333333, 1.4433756730);
\draw [color=black] (0.8333333333, 1.4433756730) -- (1.6666666667, 2.8867513459);
\draw [color=black] (0.8333333333, 1.4433756730) -- (0.0000000000, 2.8867513459);
\draw [color=black] (-0.8333333333, 1.4433756730) -- (-2.5000000000, 1.4433756730);
\draw [color=black] (-0.8333333333, 1.4433756730) -- (0.0000000000, 2.8867513459);
\draw [color=black] (-0.8333333333, 1.4433756730) -- (-1.6666666667, 2.8867513459);
\draw [color=black] (-2.5000000000, 1.4433756730) -- (-1.6666666667, 2.8867513459);
\draw [color=black] (1.6666666667, 2.8867513459) -- (0.0000000000, 2.8867513459);
\draw [color=black] (1.6666666667, 2.8867513459) -- (0.8333333333, 4.3301270189);
\draw [color=black] (0.0000000000, 2.8867513459) -- (-1.6666666667, 2.8867513459);
\draw [color=black] (0.0000000000, 2.8867513459) -- (0.8333333333, 4.3301270189);
\draw [color=black] (0.0000000000, 2.8867513459) -- (-0.8333333333, 4.3301270189);
\draw [color=black] (-1.6666666667, 2.8867513459) -- (-0.8333333333, 4.3301270189);
\draw [color=black] (0.8333333333, 4.3301270189) -- (-0.8333333333, 4.3301270189);
\draw [color=black] (0.8333333333, 4.3301270189) -- (0.0000000000, 5.7735026919);
\draw [color=black] (-0.8333333333, 4.3301270189) -- (0.0000000000, 5.7735026919);
\draw [color=black] (0, 5.7735026919) -- (0.0, 7.50555349947);
\draw [color=black] (-5.0, -2.88675134595) -- (-6.5, -3.75277674973);
\draw [color=black] (5.0, -2.88675134595) -- (6.5, -3.75277674973);
\draw[fill] (0.0000000000, -5.7735026919) circle [radius=0.1000000000];
\node at (0.0,-6.81273317644) {\Huge $O_1$};
\draw[fill] (5.0000000000, 2.8867513459) circle [radius=0.1000000000];
\node at (5.9,3.40636658822) {\Huge $O_2$};
\draw[fill] (-5.0000000000, 2.8867513459) circle [radius=0.1000000000];
\node at (-5.9,3.40636658822) {\Huge $O_3$};
\draw[fill] (4.1666666667, -2.4056261216) circle [radius=0.1000000000];
\draw[fill] (2.5000000000, -2.4056261216) circle [radius=0.1000000000];
\draw[fill] (3.3333333333, -1.9245008973) circle [radius=0.1000000000];
\draw[fill] (0.8333333333, -2.4056261216) circle [radius=0.1000000000];
\draw[fill] (1.6666666667, -1.9245008973) circle [radius=0.1000000000];
\draw[fill] (-0.8333333333, -2.4056261216) circle [radius=0.1000000000];
\draw[fill] (0.0000000000, -1.9245008973) circle [radius=0.1000000000];
\draw[fill] (-2.5000000000, -2.4056261216) circle [radius=0.1000000000];
\draw[fill] (-1.6666666667, -1.9245008973) circle [radius=0.1000000000];
\draw[fill] (-4.1666666667, -2.4056261216) circle [radius=0.1000000000];
\draw[fill] (-3.3333333333, -1.9245008973) circle [radius=0.1000000000];
\draw[fill] (3.3333333333, -0.9622504486) circle [radius=0.1000000000];
\draw[fill] (1.6666666667, -0.9622504486) circle [radius=0.1000000000];
\draw[fill] (2.5000000000, -0.4811252243) circle [radius=0.1000000000];
\draw[fill] (0.0000000000, -0.9622504486) circle [radius=0.1000000000];
\draw[fill] (0.8333333333, -0.4811252243) circle [radius=0.1000000000];
\draw[fill] (-1.6666666667, -0.9622504486) circle [radius=0.1000000000];
\draw[fill] (-0.8333333333, -0.4811252243) circle [radius=0.1000000000];
\draw[fill] (-3.3333333333, -0.9622504486) circle [radius=0.1000000000];
\draw[fill] (-2.5000000000, -0.4811252243) circle [radius=0.1000000000];
\draw[fill] (2.5000000000, 0.4811252243) circle [radius=0.1000000000];
\draw[fill] (0.8333333333, 0.4811252243) circle [radius=0.1000000000];
\draw[fill] (1.6666666667, 0.9622504486) circle [radius=0.1000000000];
\draw[fill] (-0.8333333333, 0.4811252243) circle [radius=0.1000000000];
\draw[fill] (0.0000000000, 0.9622504486) circle [radius=0.1000000000];
\draw[fill] (-2.5000000000, 0.4811252243) circle [radius=0.1000000000];
\draw[fill] (-1.6666666667, 0.9622504486) circle [radius=0.1000000000];
\draw[fill] (1.6666666667, 1.9245008973) circle [radius=0.1000000000];
\draw[fill] (0.0000000000, 1.9245008973) circle [radius=0.1000000000];
\draw[fill] (0.8333333333, 2.4056261216) circle [radius=0.1000000000];
\draw[fill] (-1.6666666667, 1.9245008973) circle [radius=0.1000000000];
\draw[fill] (-0.8333333333, 2.4056261216) circle [radius=0.1000000000];
\draw[fill] (0.8333333333, 3.3678765703) circle [radius=0.1000000000];
\draw[fill] (-0.8333333333, 3.3678765703) circle [radius=0.1000000000];
\draw[fill] (0.0000000000, 3.8490017946) circle [radius=0.1000000000];
\draw[fill] (0.0000000000, 4.8112522432) circle [radius=0.1000000000];
\draw [color=black, line width=4] (0.0000000000, -5.7735026919) -- (0.8333333333, -2.4056261216);
\draw [color=black, line width=4] (0.8333333333, -2.4056261216) -- (1.6666666667, -1.9245008973);
\draw [color=black, line width=4] (1.6666666667, -1.9245008973) -- (1.6666666667, -0.9622504486);
\draw [color=black, line width=4] (1.6666666667, -0.9622504486) -- (0.8333333333, -0.4811252243);
\draw [color=black, line width=4] (0.8333333333, -0.4811252243) -- (0.8333333333, 0.4811252243);
\draw [color=black, line width=4] (5.0000000000, 2.8867513459) -- (2.5000000000, 0.4811252243);
\draw [color=black, line width=4] (2.5000000000, 0.4811252243) -- (1.6666666667, 0.9622504486);
\draw [color=black, line width=4] (1.6666666667, 0.9622504486) -- (0.8333333333, 0.4811252243);
\draw [color=black, line width=4] (0.0000000000, -5.7735026919) -- (-2.5000000000, -2.4056261216);
\draw [color=black, line width=4] (-2.5000000000, -2.4056261216) -- (-3.3333333333, -1.9245008973);
\draw [color=black, line width=4] (-3.3333333333, -1.9245008973) -- (-3.3333333333, -0.9622504486);
\draw [color=black, line width=4] (-5.0000000000, 2.8867513459) -- (-3.3333333333, -0.9622504486);
\draw [color=black, line width=4] (-5.0000000000, 2.8867513459) -- (-0.8333333333, 3.3678765703);
\draw [color=black, line width=4] (-0.8333333333, 3.3678765703) -- (-0.8333333333, 2.4056261216);
\draw [color=black, line width=4] (-0.8333333333, 2.4056261216) -- (0.0000000000, 1.9245008973);
\draw [color=black, line width=4] (0.0000000000, 1.9245008973) -- (0.8333333333, 2.4056261216);
\draw [color=black, line width=4] (0.8333333333, 2.4056261216) -- (0.8333333333, 3.3678765703);
\draw [color=black, line width=4] (5.0000000000, 2.8867513459) -- (0.8333333333, 3.3678765703);
\end{tikzpicture}

%% file: gap.tikz
\begin{tikzpicture}
\path [fill=lightgray] (1.6666666667, -2.8867513459) -- (0.5555555556, -2.8867513459) -- (1.1111111111, -1.9245008973) -- (1.6666666667, -2.8867513459);
\path [fill=lightgray] (1.6666666667, -2.8867513459) -- (2.2222222222, -1.9245008973) -- (1.1111111111, -1.9245008973) -- (1.6666666667, -2.8867513459);
\path [fill=lightgray] (0.5555555556, -2.8867513459) -- (-0.5555555556, -2.8867513459) -- (0.0000000000, -1.9245008973) -- (0.5555555556, -2.8867513459);
\path [fill=lightgray] (0.5555555556, -2.8867513459) -- (1.1111111111, -1.9245008973) -- (0.0000000000, -1.9245008973) -- (0.5555555556, -2.8867513459);
\path [fill=lightgray] (-0.5555555556, -2.8867513459) -- (-1.6666666667, -2.8867513459) -- (-1.1111111111, -1.9245008973) -- (-0.5555555556, -2.8867513459);
\path [fill=lightgray] (-0.5555555556, -2.8867513459) -- (0.0000000000, -1.9245008973) -- (-1.1111111111, -1.9245008973) -- (-0.5555555556, -2.8867513459);
\path [fill=lightgray] (-1.6666666667, -2.8867513459) -- (-1.1111111111, -1.9245008973) -- (-2.2222222222, -1.9245008973) -- (-1.6666666667, -2.8867513459);
\path [fill=lightgray] (2.2222222222, -1.9245008973) -- (1.1111111111, -1.9245008973) -- (1.6666666667, -0.9622504486) -- (2.2222222222, -1.9245008973);
\path [fill=lightgray] (2.2222222222, -1.9245008973) -- (2.7777777778, -0.9622504486) -- (1.6666666667, -0.9622504486) -- (2.2222222222, -1.9245008973);
\path [fill=lightgray] (1.1111111111, -1.9245008973) -- (0.0000000000, -1.9245008973) -- (0.5555555556, -0.9622504486) -- (1.1111111111, -1.9245008973);
\path [fill=lightgray] (1.1111111111, -1.9245008973) -- (1.6666666667, -0.9622504486) -- (0.5555555556, -0.9622504486) -- (1.1111111111, -1.9245008973);
\path [fill=lightgray] (0.0000000000, -1.9245008973) -- (-1.1111111111, -1.9245008973) -- (-0.5555555556, -0.9622504486) -- (0.0000000000, -1.9245008973);
\path [fill=lightgray] (0.0000000000, -1.9245008973) -- (0.5555555556, -0.9622504486) -- (-0.5555555556, -0.9622504486) -- (0.0000000000, -1.9245008973);
\path [fill=lightgray] (-1.1111111111, -1.9245008973) -- (-2.2222222222, -1.9245008973) -- (-1.6666666667, -0.9622504486) -- (-1.1111111111, -1.9245008973);
\path [fill=lightgray] (-1.1111111111, -1.9245008973) -- (-0.5555555556, -0.9622504486) -- (-1.6666666667, -0.9622504486) -- (-1.1111111111, -1.9245008973);
\path [fill=lightgray] (-2.2222222222, -1.9245008973) -- (-1.6666666667, -0.9622504486) -- (-2.7777777778, -0.9622504486) -- (-2.2222222222, -1.9245008973);
\path [fill=lightgray] (2.7777777778, -0.9622504486) -- (1.6666666667, -0.9622504486) -- (2.2222222222, -0.0000000000) -- (2.7777777778, -0.9622504486);
\path [fill=lightgray] (2.7777777778, -0.9622504486) -- (3.3333333333, -0.0000000000) -- (2.2222222222, -0.0000000000) -- (2.7777777778, -0.9622504486);
\path [fill=lightgray] (1.6666666667, -0.9622504486) -- (0.5555555556, -0.9622504486) -- (1.1111111111, -0.0000000000) -- (1.6666666667, -0.9622504486);
\path [fill=lightgray] (1.6666666667, -0.9622504486) -- (2.2222222222, -0.0000000000) -- (1.1111111111, -0.0000000000) -- (1.6666666667, -0.9622504486);
\path [fill=lightgray] (0.5555555556, -0.9622504486) -- (-0.5555555556, -0.9622504486) -- (0.0000000000, -0.0000000000) -- (0.5555555556, -0.9622504486);
\path [fill=lightgray] (0.5555555556, -0.9622504486) -- (1.1111111111, -0.0000000000) -- (0.0000000000, -0.0000000000) -- (0.5555555556, -0.9622504486);
\path [fill=lightgray] (-0.5555555556, -0.9622504486) -- (-1.6666666667, -0.9622504486) -- (-1.1111111111, -0.0000000000) -- (-0.5555555556, -0.9622504486);
\path [fill=lightgray] (-0.5555555556, -0.9622504486) -- (0.0000000000, -0.0000000000) -- (-1.1111111111, -0.0000000000) -- (-0.5555555556, -0.9622504486);
\path [fill=lightgray] (-1.6666666667, -0.9622504486) -- (-2.7777777778, -0.9622504486) -- (-2.2222222222, -0.0000000000) -- (-1.6666666667, -0.9622504486);
\path [fill=lightgray] (-1.6666666667, -0.9622504486) -- (-1.1111111111, -0.0000000000) -- (-2.2222222222, -0.0000000000) -- (-1.6666666667, -0.9622504486);
\path [fill=lightgray] (-2.7777777778, -0.9622504486) -- (-2.2222222222, -0.0000000000) -- (-3.3333333333, -0.0000000000) -- (-2.7777777778, -0.9622504486);
\path [fill=lightgray] (3.3333333333, -0.0000000000) -- (2.2222222222, -0.0000000000) -- (2.7777777778, 0.9622504486) -- (3.3333333333, -0.0000000000);
\path [fill=lightgray] (2.2222222222, -0.0000000000) -- (1.1111111111, -0.0000000000) -- (1.6666666667, 0.9622504486) -- (2.2222222222, -0.0000000000);
\path [fill=lightgray] (2.2222222222, -0.0000000000) -- (2.7777777778, 0.9622504486) -- (1.6666666667, 0.9622504486) -- (2.2222222222, -0.0000000000);
\path [fill=lightgray] (1.1111111111, -0.0000000000) -- (0.0000000000, -0.0000000000) -- (0.5555555556, 0.9622504486) -- (1.1111111111, -0.0000000000);
\path [fill=lightgray] (1.1111111111, -0.0000000000) -- (1.6666666667, 0.9622504486) -- (0.5555555556, 0.9622504486) -- (1.1111111111, -0.0000000000);
\path [fill=lightgray] (0.0000000000, -0.0000000000) -- (-1.1111111111, -0.0000000000) -- (-0.5555555556, 0.9622504486) -- (0.0000000000, -0.0000000000);
\path [fill=lightgray] (0.0000000000, -0.0000000000) -- (0.5555555556, 0.9622504486) -- (-0.5555555556, 0.9622504486) -- (0.0000000000, -0.0000000000);
\path [fill=lightgray] (-1.1111111111, -0.0000000000) -- (-2.2222222222, -0.0000000000) -- (-1.6666666667, 0.9622504486) -- (-1.1111111111, -0.0000000000);
\path [fill=lightgray] (-1.1111111111, -0.0000000000) -- (-0.5555555556, 0.9622504486) -- (-1.6666666667, 0.9622504486) -- (-1.1111111111, -0.0000000000);
\path [fill=lightgray] (-2.2222222222, -0.0000000000) -- (-3.3333333333, -0.0000000000) -- (-2.7777777778, 0.9622504486) -- (-2.2222222222, -0.0000000000);
\path [fill=lightgray] (-2.2222222222, -0.0000000000) -- (-1.6666666667, 0.9622504486) -- (-2.7777777778, 0.9622504486) -- (-2.2222222222, -0.0000000000);
\path [fill=lightgray] (2.7777777778, 0.9622504486) -- (1.6666666667, 0.9622504486) -- (2.2222222222, 1.9245008973) -- (2.7777777778, 0.9622504486);
\path [fill=lightgray] (1.6666666667, 0.9622504486) -- (0.5555555556, 0.9622504486) -- (1.1111111111, 1.9245008973) -- (1.6666666667, 0.9622504486);
\path [fill=lightgray] (1.6666666667, 0.9622504486) -- (2.2222222222, 1.9245008973) -- (1.1111111111, 1.9245008973) -- (1.6666666667, 0.9622504486);
\path [fill=lightgray] (0.5555555556, 0.9622504486) -- (-0.5555555556, 0.9622504486) -- (0.0000000000, 1.9245008973) -- (0.5555555556, 0.9622504486);
\path [fill=lightgray] (0.5555555556, 0.9622504486) -- (1.1111111111, 1.9245008973) -- (0.0000000000, 1.9245008973) -- (0.5555555556, 0.9622504486);
\path [fill=lightgray] (-0.5555555556, 0.9622504486) -- (-1.6666666667, 0.9622504486) -- (-1.1111111111, 1.9245008973) -- (-0.5555555556, 0.9622504486);
\path [fill=lightgray] (-0.5555555556, 0.9622504486) -- (0.0000000000, 1.9245008973) -- (-1.1111111111, 1.9245008973) -- (-0.5555555556, 0.9622504486);
\path [fill=lightgray] (-1.6666666667, 0.9622504486) -- (-2.7777777778, 0.9622504486) -- (-2.2222222222, 1.9245008973) -- (-1.6666666667, 0.9622504486);
\path [fill=lightgray] (-1.6666666667, 0.9622504486) -- (-1.1111111111, 1.9245008973) -- (-2.2222222222, 1.9245008973) -- (-1.6666666667, 0.9622504486);
\path [fill=lightgray] (2.2222222222, 1.9245008973) -- (1.1111111111, 1.9245008973) -- (1.6666666667, 2.8867513459) -- (2.2222222222, 1.9245008973);
\path [fill=lightgray] (1.1111111111, 1.9245008973) -- (0.0000000000, 1.9245008973) -- (0.5555555556, 2.8867513459) -- (1.1111111111, 1.9245008973);
\path [fill=lightgray] (1.1111111111, 1.9245008973) -- (1.6666666667, 2.8867513459) -- (0.5555555556, 2.8867513459) -- (1.1111111111, 1.9245008973);
\path [fill=lightgray] (0.0000000000, 1.9245008973) -- (-1.1111111111, 1.9245008973) -- (-0.5555555556, 2.8867513459) -- (0.0000000000, 1.9245008973);
\path [fill=lightgray] (0.0000000000, 1.9245008973) -- (0.5555555556, 2.8867513459) -- (-0.5555555556, 2.8867513459) -- (0.0000000000, 1.9245008973);
\path [fill=lightgray] (-1.1111111111, 1.9245008973) -- (-2.2222222222, 1.9245008973) -- (-1.6666666667, 2.8867513459) -- (-1.1111111111, 1.9245008973);
\path [fill=lightgray] (-1.1111111111, 1.9245008973) -- (-0.5555555556, 2.8867513459) -- (-1.6666666667, 2.8867513459) -- (-1.1111111111, 1.9245008973);
\draw [color=black] (5.0000000000, -2.8867513459) -- (3.8888888889, -2.8867513459);
\draw [color=black] (5.0000000000, -2.8867513459) -- (4.4444444444, -1.9245008973);
\draw [color=black] (3.8888888889, -2.8867513459) -- (2.7777777778, -2.8867513459);
\draw [dashed] (3.8888888889, -2.8867513459) -- (4.4444444444, -1.9245008973);
\draw [color=black] (3.8888888889, -2.8867513459) -- (3.3333333333, -1.9245008973);
\draw [color=black] (2.7777777778, -2.8867513459) -- (1.6666666667, -2.8867513459);
\draw [dashed] (2.7777777778, -2.8867513459) -- (3.3333333333, -1.9245008973);
\draw [color=black] (2.7777777778, -2.8867513459) -- (2.2222222222, -1.9245008973);
\draw [color=black] (1.6666666667, -2.8867513459) -- (0.5555555556, -2.8867513459);
\draw [color=black] (1.6666666667, -2.8867513459) -- (2.2222222222, -1.9245008973);
\draw [color=black] (1.6666666667, -2.8867513459) -- (1.1111111111, -1.9245008973);
\draw [color=black] (0.5555555556, -2.8867513459) -- (-0.5555555556, -2.8867513459);
\draw [color=black] (0.5555555556, -2.8867513459) -- (1.1111111111, -1.9245008973);
\draw [color=black] (0.5555555556, -2.8867513459) -- (0.0000000000, -1.9245008973);
\draw [color=black] (-0.5555555556, -2.8867513459) -- (-1.6666666667, -2.8867513459);
\draw [color=black] (-0.5555555556, -2.8867513459) -- (0.0000000000, -1.9245008973);
\draw [color=black] (-0.5555555556, -2.8867513459) -- (-1.1111111111, -1.9245008973);
\draw [color=black] (-1.6666666667, -2.8867513459) -- (-2.7777777778, -2.8867513459);
\draw [color=black] (-1.6666666667, -2.8867513459) -- (-1.1111111111, -1.9245008973);
\draw [color=black] (-1.6666666667, -2.8867513459) -- (-2.2222222222, -1.9245008973);
\draw [color=black] (-2.7777777778, -2.8867513459) -- (-3.8888888889, -2.8867513459);
\draw [color=black] (-2.7777777778, -2.8867513459) -- (-2.2222222222, -1.9245008973);
\draw [dashed] (-2.7777777778, -2.8867513459) -- (-3.3333333333, -1.9245008973);
\draw [color=black] (-3.8888888889, -2.8867513459) -- (-5.0000000000, -2.8867513459);
\draw [color=black] (-3.8888888889, -2.8867513459) -- (-3.3333333333, -1.9245008973);
\draw [dashed] (-3.8888888889, -2.8867513459) -- (-4.4444444444, -1.9245008973);
\draw [color=black] (-5.0000000000, -2.8867513459) -- (-4.4444444444, -1.9245008973);
\draw [color=black] (4.4444444444, -1.9245008973) -- (3.3333333333, -1.9245008973);
\draw [color=black] (4.4444444444, -1.9245008973) -- (3.8888888889, -0.9622504486);
\draw [color=black] (3.3333333333, -1.9245008973) -- (2.2222222222, -1.9245008973);
\draw [dashed] (3.3333333333, -1.9245008973) -- (3.8888888889, -0.9622504486);
\draw [color=black] (3.3333333333, -1.9245008973) -- (2.7777777778, -0.9622504486);
\draw [color=black] (2.2222222222, -1.9245008973) -- (1.1111111111, -1.9245008973);
\draw [color=black] (2.2222222222, -1.9245008973) -- (2.7777777778, -0.9622504486);
\draw [color=black] (2.2222222222, -1.9245008973) -- (1.6666666667, -0.9622504486);
\draw [color=black] (1.1111111111, -1.9245008973) -- (0.0000000000, -1.9245008973);
\draw [color=black] (1.1111111111, -1.9245008973) -- (1.6666666667, -0.9622504486);
\draw [color=black] (1.1111111111, -1.9245008973) -- (0.5555555556, -0.9622504486);
\draw [color=black] (0.0000000000, -1.9245008973) -- (-1.1111111111, -1.9245008973);
\draw [color=black] (0.0000000000, -1.9245008973) -- (0.5555555556, -0.9622504486);
\draw [color=black] (0.0000000000, -1.9245008973) -- (-0.5555555556, -0.9622504486);
\draw [color=black] (-1.1111111111, -1.9245008973) -- (-2.2222222222, -1.9245008973);
\draw [color=black] (-1.1111111111, -1.9245008973) -- (-0.5555555556, -0.9622504486);
\draw [color=black] (-1.1111111111, -1.9245008973) -- (-1.6666666667, -0.9622504486);
\draw [color=black] (-2.2222222222, -1.9245008973) -- (-3.3333333333, -1.9245008973);
\draw [color=black] (-2.2222222222, -1.9245008973) -- (-1.6666666667, -0.9622504486);
\draw [color=black] (-2.2222222222, -1.9245008973) -- (-2.7777777778, -0.9622504486);
\draw [color=black] (-3.3333333333, -1.9245008973) -- (-4.4444444444, -1.9245008973);
\draw [color=black] (-3.3333333333, -1.9245008973) -- (-2.7777777778, -0.9622504486);
\draw [dashed] (-3.3333333333, -1.9245008973) -- (-3.8888888889, -0.9622504486);
\draw [color=black] (-4.4444444444, -1.9245008973) -- (-3.8888888889, -0.9622504486);
\draw [color=black] (3.8888888889, -0.9622504486) -- (2.7777777778, -0.9622504486);
\draw [color=black] (3.8888888889, -0.9622504486) -- (3.3333333333, -0.0000000000);
\draw [color=black] (2.7777777778, -0.9622504486) -- (1.6666666667, -0.9622504486);
\draw [color=black] (2.7777777778, -0.9622504486) -- (3.3333333333, -0.0000000000);
\draw [color=black] (2.7777777778, -0.9622504486) -- (2.2222222222, -0.0000000000);
\draw [color=black] (1.6666666667, -0.9622504486) -- (0.5555555556, -0.9622504486);
\draw [color=black] (1.6666666667, -0.9622504486) -- (2.2222222222, -0.0000000000);
\draw [color=black] (1.6666666667, -0.9622504486) -- (1.1111111111, -0.0000000000);
\draw [color=black] (0.5555555556, -0.9622504486) -- (-0.5555555556, -0.9622504486);
\draw [color=black] (0.5555555556, -0.9622504486) -- (1.1111111111, -0.0000000000);
\draw [color=black] (0.5555555556, -0.9622504486) -- (0.0000000000, -0.0000000000);
\draw [color=black] (-0.5555555556, -0.9622504486) -- (-1.6666666667, -0.9622504486);
\draw [color=black] (-0.5555555556, -0.9622504486) -- (0.0000000000, -0.0000000000);
\draw [color=black] (-0.5555555556, -0.9622504486) -- (-1.1111111111, -0.0000000000);
\draw [color=black] (-1.6666666667, -0.9622504486) -- (-2.7777777778, -0.9622504486);
\draw [color=black] (-1.6666666667, -0.9622504486) -- (-1.1111111111, -0.0000000000);
\draw [color=black] (-1.6666666667, -0.9622504486) -- (-2.2222222222, -0.0000000000);
\draw [color=black] (-2.7777777778, -0.9622504486) -- (-3.8888888889, -0.9622504486);
\draw [color=black] (-2.7777777778, -0.9622504486) -- (-2.2222222222, -0.0000000000);
\draw [color=black] (-2.7777777778, -0.9622504486) -- (-3.3333333333, -0.0000000000);
\draw [color=black] (-3.8888888889, -0.9622504486) -- (-3.3333333333, -0.0000000000);
\draw [color=black] (3.3333333333, -0.0000000000) -- (2.2222222222, -0.0000000000);
\draw [color=black] (3.3333333333, -0.0000000000) -- (2.7777777778, 0.9622504486);
\draw [color=black] (2.2222222222, -0.0000000000) -- (1.1111111111, -0.0000000000);
\draw [color=black] (2.2222222222, -0.0000000000) -- (2.7777777778, 0.9622504486);
\draw [color=black] (2.2222222222, -0.0000000000) -- (1.6666666667, 0.9622504486);
\draw [color=black] (1.1111111111, -0.0000000000) -- (0.0000000000, -0.0000000000);
\draw [color=black] (1.1111111111, -0.0000000000) -- (1.6666666667, 0.9622504486);
\draw [color=black] (1.1111111111, -0.0000000000) -- (0.5555555556, 0.9622504486);
\draw [color=black] (0.0000000000, -0.0000000000) -- (-1.1111111111, -0.0000000000);
\draw [color=black] (0.0000000000, -0.0000000000) -- (0.5555555556, 0.9622504486);
\draw [color=black] (0.0000000000, -0.0000000000) -- (-0.5555555556, 0.9622504486);
\draw [color=black] (-1.1111111111, -0.0000000000) -- (-2.2222222222, -0.0000000000);
\draw [color=black] (-1.1111111111, -0.0000000000) -- (-0.5555555556, 0.9622504486);
\draw [color=black] (-1.1111111111, -0.0000000000) -- (-1.6666666667, 0.9622504486);
\draw [color=black] (-2.2222222222, -0.0000000000) -- (-3.3333333333, -0.0000000000);
\draw [color=black] (-2.2222222222, -0.0000000000) -- (-1.6666666667, 0.9622504486);
\draw [color=black] (-2.2222222222, -0.0000000000) -- (-2.7777777778, 0.9622504486);
\draw [color=black] (-3.3333333333, -0.0000000000) -- (-2.7777777778, 0.9622504486);
\draw [color=black] (2.7777777778, 0.9622504486) -- (1.6666666667, 0.9622504486);
\draw [color=black] (2.7777777778, 0.9622504486) -- (2.2222222222, 1.9245008973);
\draw [color=black] (1.6666666667, 0.9622504486) -- (0.5555555556, 0.9622504486);
\draw [color=black] (1.6666666667, 0.9622504486) -- (2.2222222222, 1.9245008973);
\draw [color=black] (1.6666666667, 0.9622504486) -- (1.1111111111, 1.9245008973);
\draw [color=black] (0.5555555556, 0.9622504486) -- (-0.5555555556, 0.9622504486);
\draw [color=black] (0.5555555556, 0.9622504486) -- (1.1111111111, 1.9245008973);
\draw [color=black] (0.5555555556, 0.9622504486) -- (0.0000000000, 1.9245008973);
\draw [color=black] (-0.5555555556, 0.9622504486) -- (-1.6666666667, 0.9622504486);
\draw [color=black] (-0.5555555556, 0.9622504486) -- (0.0000000000, 1.9245008973);
\draw [color=black] (-0.5555555556, 0.9622504486) -- (-1.1111111111, 1.9245008973);
\draw [color=black] (-1.6666666667, 0.9622504486) -- (-2.7777777778, 0.9622504486);
\draw [color=black] (-1.6666666667, 0.9622504486) -- (-1.1111111111, 1.9245008973);
\draw [color=black] (-1.6666666667, 0.9622504486) -- (-2.2222222222, 1.9245008973);
\draw [color=black] (-2.7777777778, 0.9622504486) -- (-2.2222222222, 1.9245008973);
\draw [color=black] (2.2222222222, 1.9245008973) -- (1.1111111111, 1.9245008973);
\draw [color=black] (2.2222222222, 1.9245008973) -- (1.6666666667, 2.8867513459);
\draw [color=black] (1.1111111111, 1.9245008973) -- (0.0000000000, 1.9245008973);
\draw [color=black] (1.1111111111, 1.9245008973) -- (1.6666666667, 2.8867513459);
\draw [color=black] (1.1111111111, 1.9245008973) -- (0.5555555556, 2.8867513459);
\draw [color=black] (0.0000000000, 1.9245008973) -- (-1.1111111111, 1.9245008973);
\draw [color=black] (0.0000000000, 1.9245008973) -- (0.5555555556, 2.8867513459);
\draw [color=black] (0.0000000000, 1.9245008973) -- (-0.5555555556, 2.8867513459);
\draw [color=black] (-1.1111111111, 1.9245008973) -- (-2.2222222222, 1.9245008973);
\draw [color=black] (-1.1111111111, 1.9245008973) -- (-0.5555555556, 2.8867513459);
\draw [color=black] (-1.1111111111, 1.9245008973) -- (-1.6666666667, 2.8867513459);
\draw [color=black] (-2.2222222222, 1.9245008973) -- (-1.6666666667, 2.8867513459);
\draw [color=black] (1.6666666667, 2.8867513459) -- (0.5555555556, 2.8867513459);
\draw [color=black] (1.6666666667, 2.8867513459) -- (1.1111111111, 3.8490017946);
\draw [color=black] (0.5555555556, 2.8867513459) -- (-0.5555555556, 2.8867513459);
\draw [color=black] (0.5555555556, 2.8867513459) -- (1.1111111111, 3.8490017946);
\draw [color=black] (0.5555555556, 2.8867513459) -- (0.0000000000, 3.8490017946);
\draw [color=black] (-0.5555555556, 2.8867513459) -- (-1.6666666667, 2.8867513459);
\draw [color=black] (-0.5555555556, 2.8867513459) -- (0.0000000000, 3.8490017946);
\draw [color=black] (-0.5555555556, 2.8867513459) -- (-1.1111111111, 3.8490017946);
\draw [color=black] (-1.6666666667, 2.8867513459) -- (-1.1111111111, 3.8490017946);
\draw [dashed] (1.1111111111, 3.8490017946) -- (0.0000000000, 3.8490017946);
\draw [color=black] (1.1111111111, 3.8490017946) -- (0.5555555556, 4.8112522432);
\draw [dashed] (0.0000000000, 3.8490017946) -- (-1.1111111111, 3.8490017946);
\draw [color=black] (0.0000000000, 3.8490017946) -- (0.5555555556, 4.8112522432);
\draw [color=black] (0.0000000000, 3.8490017946) -- (-0.5555555556, 4.8112522432);
\draw [color=black] (-1.1111111111, 3.8490017946) -- (-0.5555555556, 4.8112522432);
\draw [dashed] (0.5555555556, 4.8112522432) -- (-0.5555555556, 4.8112522432);
\draw [color=black] (0.5555555556, 4.8112522432) -- (0.0000000000, 5.7735026919);
\draw [color=black] (-0.5555555556, 4.8112522432) -- (0.0000000000, 5.7735026919);
\node at (0.0,6.23538290725) {\large $e^1$};
\node at (-5.4,-3.11769145362) {\large $e^2$};
\node at (5.4,-3.11769145362) {\large $e^3$};
\node at (4.44444444444,-3.3486315613) {\large $3\rho$};
\node at (5.12222222222,-2.17468601395) {\large $3\rho$};
\node at (3.33333333333,-3.3486315613) {\large $2\rho$};
\node at (2.22222222222,-3.3486315613) {\large $\rho$};
\node at (1.11111111111,-3.3486315613) {\large $\rho$};
\node at (0.0,-3.3486315613) {\large $\rho$};
\node at (-1.11111111111,-3.3486315613) {\large $\rho$};
\node at (-2.22222222222,-3.3486315613) {\large $\rho$};
\node at (-3.33333333333,-3.3486315613) {\large $2\rho$};
\node at (-4.44444444444,-3.3486315613) {\large $3\rho$};
\node at (-5.12222222222,-2.17468601395) {\large $3\rho$};
\node at (4.56666666667,-1.2124355653) {\large $2\rho$};
\node at (-4.56666666667,-1.2124355653) {\large $2\rho$};
\node at (4.01111111111,-0.250185116649) {\large $\rho$};
\node at (-4.01111111111,-0.250185116649) {\large $\rho$};
\node at (3.45555555556,0.712065332001) {\large $\rho$};
\node at (-3.45555555556,0.712065332001) {\large $\rho$};
\node at (2.9,1.67431578065) {\large $\rho$};
\node at (-2.9,1.67431578065) {\large $\rho$};
\node at (2.34444444444,2.6365662293) {\large $\rho$};
\node at (-2.34444444444,2.6365662293) {\large $\rho$};
\node at (1.78888888889,3.59881667795) {\large $\rho$};
\node at (-1.78888888889,3.59881667795) {\large $\rho$};
\node at (1.23333333333,4.5610671266) {\large $2\rho$};
\node at (-1.23333333333,4.5610671266) {\large $2\rho$};
\node at (0.677777777778,5.52331757525) {\large $3\rho$};
\node at (-0.677777777778,5.52331757525) {\large $3\rho$};
\end{tikzpicture}

%% file: potential.tikz
\begin{tikzpicture}
\path [fill=lightgray] (1.6666666667, -2.8867513459) -- (0.5555555556, -2.8867513459) -- (1.1111111111, -1.9245008973) -- (1.6666666667, -2.8867513459);
\path [fill=lightgray] (1.6666666667, -2.8867513459) -- (2.2222222222, -1.9245008973) -- (1.1111111111, -1.9245008973) -- (1.6666666667, -2.8867513459);
\path [fill=lightgray] (0.5555555556, -2.8867513459) -- (-0.5555555556, -2.8867513459) -- (0.0000000000, -1.9245008973) -- (0.5555555556, -2.8867513459);
\path [fill=lightgray] (0.5555555556, -2.8867513459) -- (1.1111111111, -1.9245008973) -- (0.0000000000, -1.9245008973) -- (0.5555555556, -2.8867513459);
\path [fill=lightgray] (-0.5555555556, -2.8867513459) -- (-1.6666666667, -2.8867513459) -- (-1.1111111111, -1.9245008973) -- (-0.5555555556, -2.8867513459);
\path [fill=lightgray] (-0.5555555556, -2.8867513459) -- (0.0000000000, -1.9245008973) -- (-1.1111111111, -1.9245008973) -- (-0.5555555556, -2.8867513459);
\path [fill=lightgray] (-1.6666666667, -2.8867513459) -- (-1.1111111111, -1.9245008973) -- (-2.2222222222, -1.9245008973) -- (-1.6666666667, -2.8867513459);
\path [fill=lightgray] (2.2222222222, -1.9245008973) -- (1.1111111111, -1.9245008973) -- (1.6666666667, -0.9622504486) -- (2.2222222222, -1.9245008973);
\path [fill=lightgray] (2.2222222222, -1.9245008973) -- (2.7777777778, -0.9622504486) -- (1.6666666667, -0.9622504486) -- (2.2222222222, -1.9245008973);
\path [fill=lightgray] (1.1111111111, -1.9245008973) -- (0.0000000000, -1.9245008973) -- (0.5555555556, -0.9622504486) -- (1.1111111111, -1.9245008973);
\path [fill=lightgray] (1.1111111111, -1.9245008973) -- (1.6666666667, -0.9622504486) -- (0.5555555556, -0.9622504486) -- (1.1111111111, -1.9245008973);
\path [fill=lightgray] (0.0000000000, -1.9245008973) -- (-1.1111111111, -1.9245008973) -- (-0.5555555556, -0.9622504486) -- (0.0000000000, -1.9245008973);
\path [fill=lightgray] (0.0000000000, -1.9245008973) -- (0.5555555556, -0.9622504486) -- (-0.5555555556, -0.9622504486) -- (0.0000000000, -1.9245008973);
\path [fill=lightgray] (-1.1111111111, -1.9245008973) -- (-2.2222222222, -1.9245008973) -- (-1.6666666667, -0.9622504486) -- (-1.1111111111, -1.9245008973);
\path [fill=lightgray] (-1.1111111111, -1.9245008973) -- (-0.5555555556, -0.9622504486) -- (-1.6666666667, -0.9622504486) -- (-1.1111111111, -1.9245008973);
\path [fill=lightgray] (-2.2222222222, -1.9245008973) -- (-1.6666666667, -0.9622504486) -- (-2.7777777778, -0.9622504486) -- (-2.2222222222, -1.9245008973);
\path [fill=lightgray] (2.7777777778, -0.9622504486) -- (1.6666666667, -0.9622504486) -- (2.2222222222, -0.0000000000) -- (2.7777777778, -0.9622504486);
\path [fill=lightgray] (2.7777777778, -0.9622504486) -- (3.3333333333, -0.0000000000) -- (2.2222222222, -0.0000000000) -- (2.7777777778, -0.9622504486);
\path [fill=lightgray] (1.6666666667, -0.9622504486) -- (0.5555555556, -0.9622504486) -- (1.1111111111, -0.0000000000) -- (1.6666666667, -0.9622504486);
\path [fill=lightgray] (1.6666666667, -0.9622504486) -- (2.2222222222, -0.0000000000) -- (1.1111111111, -0.0000000000) -- (1.6666666667, -0.9622504486);
\path [fill=lightgray] (0.5555555556, -0.9622504486) -- (-0.5555555556, -0.9622504486) -- (0.0000000000, -0.0000000000) -- (0.5555555556, -0.9622504486);
\path [fill=lightgray] (0.5555555556, -0.9622504486) -- (1.1111111111, -0.0000000000) -- (0.0000000000, -0.0000000000) -- (0.5555555556, -0.9622504486);
\path [fill=lightgray] (-0.5555555556, -0.9622504486) -- (-1.6666666667, -0.9622504486) -- (-1.1111111111, -0.0000000000) -- (-0.5555555556, -0.9622504486);
\path [fill=lightgray] (-0.5555555556, -0.9622504486) -- (0.0000000000, -0.0000000000) -- (-1.1111111111, -0.0000000000) -- (-0.5555555556, -0.9622504486);
\path [fill=lightgray] (-1.6666666667, -0.9622504486) -- (-2.7777777778, -0.9622504486) -- (-2.2222222222, -0.0000000000) -- (-1.6666666667, -0.9622504486);
\path [fill=lightgray] (-1.6666666667, -0.9622504486) -- (-1.1111111111, -0.0000000000) -- (-2.2222222222, -0.0000000000) -- (-1.6666666667, -0.9622504486);
\path [fill=lightgray] (-2.7777777778, -0.9622504486) -- (-2.2222222222, -0.0000000000) -- (-3.3333333333, -0.0000000000) -- (-2.7777777778, -0.9622504486);
\path [fill=lightgray] (3.3333333333, -0.0000000000) -- (2.2222222222, -0.0000000000) -- (2.7777777778, 0.9622504486) -- (3.3333333333, -0.0000000000);
\path [fill=lightgray] (2.2222222222, -0.0000000000) -- (1.1111111111, -0.0000000000) -- (1.6666666667, 0.9622504486) -- (2.2222222222, -0.0000000000);
\path [fill=lightgray] (2.2222222222, -0.0000000000) -- (2.7777777778, 0.9622504486) -- (1.6666666667, 0.9622504486) -- (2.2222222222, -0.0000000000);
\path [fill=lightgray] (1.1111111111, -0.0000000000) -- (0.0000000000, -0.0000000000) -- (0.5555555556, 0.9622504486) -- (1.1111111111, -0.0000000000);
\path [fill=lightgray] (1.1111111111, -0.0000000000) -- (1.6666666667, 0.9622504486) -- (0.5555555556, 0.9622504486) -- (1.1111111111, -0.0000000000);
\path [fill=lightgray] (0.0000000000, -0.0000000000) -- (-1.1111111111, -0.0000000000) -- (-0.5555555556, 0.9622504486) -- (0.0000000000, -0.0000000000);
\path [fill=lightgray] (0.0000000000, -0.0000000000) -- (0.5555555556, 0.9622504486) -- (-0.5555555556, 0.9622504486) -- (0.0000000000, -0.0000000000);
\path [fill=lightgray] (-1.1111111111, -0.0000000000) -- (-2.2222222222, -0.0000000000) -- (-1.6666666667, 0.9622504486) -- (-1.1111111111, -0.0000000000);
\path [fill=lightgray] (-1.1111111111, -0.0000000000) -- (-0.5555555556, 0.9622504486) -- (-1.6666666667, 0.9622504486) -- (-1.1111111111, -0.0000000000);
\path [fill=lightgray] (-2.2222222222, -0.0000000000) -- (-3.3333333333, -0.0000000000) -- (-2.7777777778, 0.9622504486) -- (-2.2222222222, -0.0000000000);
\path [fill=lightgray] (-2.2222222222, -0.0000000000) -- (-1.6666666667, 0.9622504486) -- (-2.7777777778, 0.9622504486) -- (-2.2222222222, -0.0000000000);
\path [fill=lightgray] (2.7777777778, 0.9622504486) -- (1.6666666667, 0.9622504486) -- (2.2222222222, 1.9245008973) -- (2.7777777778, 0.9622504486);
\path [fill=lightgray] (1.6666666667, 0.9622504486) -- (0.5555555556, 0.9622504486) -- (1.1111111111, 1.9245008973) -- (1.6666666667, 0.9622504486);
\path [fill=lightgray] (1.6666666667, 0.9622504486) -- (2.2222222222, 1.9245008973) -- (1.1111111111, 1.9245008973) -- (1.6666666667, 0.9622504486);
\path [fill=lightgray] (0.5555555556, 0.9622504486) -- (-0.5555555556, 0.9622504486) -- (0.0000000000, 1.9245008973) -- (0.5555555556, 0.9622504486);
\path [fill=lightgray] (0.5555555556, 0.9622504486) -- (1.1111111111, 1.9245008973) -- (0.0000000000, 1.9245008973) -- (0.5555555556, 0.9622504486);
\path [fill=lightgray] (-0.5555555556, 0.9622504486) -- (-1.6666666667, 0.9622504486) -- (-1.1111111111, 1.9245008973) -- (-0.5555555556, 0.9622504486);
\path [fill=lightgray] (-0.5555555556, 0.9622504486) -- (0.0000000000, 1.9245008973) -- (-1.1111111111, 1.9245008973) -- (-0.5555555556, 0.9622504486);
\path [fill=lightgray] (-1.6666666667, 0.9622504486) -- (-2.7777777778, 0.9622504486) -- (-2.2222222222, 1.9245008973) -- (-1.6666666667, 0.9622504486);
\path [fill=lightgray] (-1.6666666667, 0.9622504486) -- (-1.1111111111, 1.9245008973) -- (-2.2222222222, 1.9245008973) -- (-1.6666666667, 0.9622504486);
\path [fill=lightgray] (2.2222222222, 1.9245008973) -- (1.1111111111, 1.9245008973) -- (1.6666666667, 2.8867513459) -- (2.2222222222, 1.9245008973);
\path [fill=lightgray] (1.1111111111, 1.9245008973) -- (0.0000000000, 1.9245008973) -- (0.5555555556, 2.8867513459) -- (1.1111111111, 1.9245008973);
\path [fill=lightgray] (1.1111111111, 1.9245008973) -- (1.6666666667, 2.8867513459) -- (0.5555555556, 2.8867513459) -- (1.1111111111, 1.9245008973);
\path [fill=lightgray] (0.0000000000, 1.9245008973) -- (-1.1111111111, 1.9245008973) -- (-0.5555555556, 2.8867513459) -- (0.0000000000, 1.9245008973);
\path [fill=lightgray] (0.0000000000, 1.9245008973) -- (0.5555555556, 2.8867513459) -- (-0.5555555556, 2.8867513459) -- (0.0000000000, 1.9245008973);
\path [fill=lightgray] (-1.1111111111, 1.9245008973) -- (-2.2222222222, 1.9245008973) -- (-1.6666666667, 2.8867513459) -- (-1.1111111111, 1.9245008973);
\path [fill=lightgray] (-1.1111111111, 1.9245008973) -- (-0.5555555556, 2.8867513459) -- (-1.6666666667, 2.8867513459) -- (-1.1111111111, 1.9245008973);
\draw [color=black] (5.0000000000, -2.8867513459) -- (3.8888888889, -2.8867513459);
\draw [color=black] (5.0000000000, -2.8867513459) -- (4.4444444444, -1.9245008973);
\draw [color=black] (3.8888888889, -2.8867513459) -- (2.7777777778, -2.8867513459);
\draw [dashed] (3.8888888889, -2.8867513459) -- (4.4444444444, -1.9245008973);
\draw [color=black] (3.8888888889, -2.8867513459) -- (3.3333333333, -1.9245008973);
\draw [color=black] (2.7777777778, -2.8867513459) -- (1.6666666667, -2.8867513459);
\draw [dashed] (2.7777777778, -2.8867513459) -- (3.3333333333, -1.9245008973);
\draw [color=black] (2.7777777778, -2.8867513459) -- (2.2222222222, -1.9245008973);
\draw [color=black] (1.6666666667, -2.8867513459) -- (0.5555555556, -2.8867513459);
\draw [color=black] (1.6666666667, -2.8867513459) -- (2.2222222222, -1.9245008973);
\draw [color=black] (1.6666666667, -2.8867513459) -- (1.1111111111, -1.9245008973);
\draw [color=black] (0.5555555556, -2.8867513459) -- (-0.5555555556, -2.8867513459);
\draw [color=black] (0.5555555556, -2.8867513459) -- (1.1111111111, -1.9245008973);
\draw [color=black] (0.5555555556, -2.8867513459) -- (0.0000000000, -1.9245008973);
\draw [color=black] (-0.5555555556, -2.8867513459) -- (-1.6666666667, -2.8867513459);
\draw [color=black] (-0.5555555556, -2.8867513459) -- (0.0000000000, -1.9245008973);
\draw [color=black] (-0.5555555556, -2.8867513459) -- (-1.1111111111, -1.9245008973);
\draw [color=black] (-1.6666666667, -2.8867513459) -- (-2.7777777778, -2.8867513459);
\draw [color=black] (-1.6666666667, -2.8867513459) -- (-1.1111111111, -1.9245008973);
\draw [color=black] (-1.6666666667, -2.8867513459) -- (-2.2222222222, -1.9245008973);
\draw [color=black] (-2.7777777778, -2.8867513459) -- (-3.8888888889, -2.8867513459);
\draw [color=black] (-2.7777777778, -2.8867513459) -- (-2.2222222222, -1.9245008973);
\draw [dashed] (-2.7777777778, -2.8867513459) -- (-3.3333333333, -1.9245008973);
\draw [color=black] (-3.8888888889, -2.8867513459) -- (-5.0000000000, -2.8867513459);
\draw [color=black] (-3.8888888889, -2.8867513459) -- (-3.3333333333, -1.9245008973);
\draw [dashed] (-3.8888888889, -2.8867513459) -- (-4.4444444444, -1.9245008973);
\draw [color=black] (-5.0000000000, -2.8867513459) -- (-4.4444444444, -1.9245008973);
\draw [color=black] (4.4444444444, -1.9245008973) -- (3.3333333333, -1.9245008973);
\draw [color=black] (4.4444444444, -1.9245008973) -- (3.8888888889, -0.9622504486);
\draw [color=black] (3.3333333333, -1.9245008973) -- (2.2222222222, -1.9245008973);
\draw [dashed] (3.3333333333, -1.9245008973) -- (3.8888888889, -0.9622504486);
\draw [color=black] (3.3333333333, -1.9245008973) -- (2.7777777778, -0.9622504486);
\draw [color=black] (2.2222222222, -1.9245008973) -- (1.1111111111, -1.9245008973);
\draw [color=black] (2.2222222222, -1.9245008973) -- (2.7777777778, -0.9622504486);
\draw [color=black] (2.2222222222, -1.9245008973) -- (1.6666666667, -0.9622504486);
\draw [color=black] (1.1111111111, -1.9245008973) -- (0.0000000000, -1.9245008973);
\draw [color=black] (1.1111111111, -1.9245008973) -- (1.6666666667, -0.9622504486);
\draw [color=black] (1.1111111111, -1.9245008973) -- (0.5555555556, -0.9622504486);
\draw [color=black] (0.0000000000, -1.9245008973) -- (-1.1111111111, -1.9245008973);
\draw [color=black] (0.0000000000, -1.9245008973) -- (0.5555555556, -0.9622504486);
\draw [color=black] (0.0000000000, -1.9245008973) -- (-0.5555555556, -0.9622504486);
\draw [color=black] (-1.1111111111, -1.9245008973) -- (-2.2222222222, -1.9245008973);
\draw [color=black] (-1.1111111111, -1.9245008973) -- (-0.5555555556, -0.9622504486);
\draw [color=black] (-1.1111111111, -1.9245008973) -- (-1.6666666667, -0.9622504486);
\draw [color=black] (-2.2222222222, -1.9245008973) -- (-3.3333333333, -1.9245008973);
\draw [color=black] (-2.2222222222, -1.9245008973) -- (-1.6666666667, -0.9622504486);
\draw [color=black] (-2.2222222222, -1.9245008973) -- (-2.7777777778, -0.9622504486);
\draw [color=black] (-3.3333333333, -1.9245008973) -- (-4.4444444444, -1.9245008973);
\draw [color=black] (-3.3333333333, -1.9245008973) -- (-2.7777777778, -0.9622504486);
\draw [dashed] (-3.3333333333, -1.9245008973) -- (-3.8888888889, -0.9622504486);
\draw [color=black] (-4.4444444444, -1.9245008973) -- (-3.8888888889, -0.9622504486);
\draw [color=black] (3.8888888889, -0.9622504486) -- (2.7777777778, -0.9622504486);
\draw [color=black] (3.8888888889, -0.9622504486) -- (3.3333333333, -0.0000000000);
\draw [color=black] (2.7777777778, -0.9622504486) -- (1.6666666667, -0.9622504486);
\draw [color=black] (2.7777777778, -0.9622504486) -- (3.3333333333, -0.0000000000);
\draw [color=black] (2.7777777778, -0.9622504486) -- (2.2222222222, -0.0000000000);
\draw [color=black] (1.6666666667, -0.9622504486) -- (0.5555555556, -0.9622504486);
\draw [color=black] (1.6666666667, -0.9622504486) -- (2.2222222222, -0.0000000000);
\draw [color=black] (1.6666666667, -0.9622504486) -- (1.1111111111, -0.0000000000);
\draw [color=black] (0.5555555556, -0.9622504486) -- (-0.5555555556, -0.9622504486);
\draw [color=black] (0.5555555556, -0.9622504486) -- (1.1111111111, -0.0000000000);
\draw [color=black] (0.5555555556, -0.9622504486) -- (0.0000000000, -0.0000000000);
\draw [color=black] (-0.5555555556, -0.9622504486) -- (-1.6666666667, -0.9622504486);
\draw [color=black] (-0.5555555556, -0.9622504486) -- (0.0000000000, -0.0000000000);
\draw [color=black] (-0.5555555556, -0.9622504486) -- (-1.1111111111, -0.0000000000);
\draw [color=black] (-1.6666666667, -0.9622504486) -- (-2.7777777778, -0.9622504486);
\draw [color=black] (-1.6666666667, -0.9622504486) -- (-1.1111111111, -0.0000000000);
\draw [color=black] (-1.6666666667, -0.9622504486) -- (-2.2222222222, -0.0000000000);
\draw [color=black] (-2.7777777778, -0.9622504486) -- (-3.8888888889, -0.9622504486);
\draw [color=black] (-2.7777777778, -0.9622504486) -- (-2.2222222222, -0.0000000000);
\draw [color=black] (-2.7777777778, -0.9622504486) -- (-3.3333333333, -0.0000000000);
\draw [color=black] (-3.8888888889, -0.9622504486) -- (-3.3333333333, -0.0000000000);
\draw [color=black] (3.3333333333, -0.0000000000) -- (2.2222222222, -0.0000000000);
\draw [color=black] (3.3333333333, -0.0000000000) -- (2.7777777778, 0.9622504486);
\draw [color=black] (2.2222222222, -0.0000000000) -- (1.1111111111, -0.0000000000);
\draw [color=black] (2.2222222222, -0.0000000000) -- (2.7777777778, 0.9622504486);
\draw [color=black] (2.2222222222, -0.0000000000) -- (1.6666666667, 0.9622504486);
\draw [color=black] (1.1111111111, -0.0000000000) -- (0.0000000000, -0.0000000000);
\draw [color=black] (1.1111111111, -0.0000000000) -- (1.6666666667, 0.9622504486);
\draw [color=black] (1.1111111111, -0.0000000000) -- (0.5555555556, 0.9622504486);
\draw [color=black] (0.0000000000, -0.0000000000) -- (-1.1111111111, -0.0000000000);
\draw [color=black] (0.0000000000, -0.0000000000) -- (0.5555555556, 0.9622504486);
\draw [color=black] (0.0000000000, -0.0000000000) -- (-0.5555555556, 0.9622504486);
\draw [color=black] (-1.1111111111, -0.0000000000) -- (-2.2222222222, -0.0000000000);
\draw [color=black] (-1.1111111111, -0.0000000000) -- (-0.5555555556, 0.9622504486);
\draw [color=black] (-1.1111111111, -0.0000000000) -- (-1.6666666667, 0.9622504486);
\draw [color=black] (-2.2222222222, -0.0000000000) -- (-3.3333333333, -0.0000000000);
\draw [color=black] (-2.2222222222, -0.0000000000) -- (-1.6666666667, 0.9622504486);
\draw [color=black] (-2.2222222222, -0.0000000000) -- (-2.7777777778, 0.9622504486);
\draw [color=black] (-3.3333333333, -0.0000000000) -- (-2.7777777778, 0.9622504486);
\draw [color=black] (2.7777777778, 0.9622504486) -- (1.6666666667, 0.9622504486);
\draw [color=black] (2.7777777778, 0.9622504486) -- (2.2222222222, 1.9245008973);
\draw [color=black] (1.6666666667, 0.9622504486) -- (0.5555555556, 0.9622504486);
\draw [color=black] (1.6666666667, 0.9622504486) -- (2.2222222222, 1.9245008973);
\draw [color=black] (1.6666666667, 0.9622504486) -- (1.1111111111, 1.9245008973);
\draw [color=black] (0.5555555556, 0.9622504486) -- (-0.5555555556, 0.9622504486);
\draw [color=black] (0.5555555556, 0.9622504486) -- (1.1111111111, 1.9245008973);
\draw [color=black] (0.5555555556, 0.9622504486) -- (0.0000000000, 1.9245008973);
\draw [color=black] (-0.5555555556, 0.9622504486) -- (-1.6666666667, 0.9622504486);
\draw [color=black] (-0.5555555556, 0.9622504486) -- (0.0000000000, 1.9245008973);
\draw [color=black] (-0.5555555556, 0.9622504486) -- (-1.1111111111, 1.9245008973);
\draw [color=black] (-1.6666666667, 0.9622504486) -- (-2.7777777778, 0.9622504486);
\draw [color=black] (-1.6666666667, 0.9622504486) -- (-1.1111111111, 1.9245008973);
\draw [color=black] (-1.6666666667, 0.9622504486) -- (-2.2222222222, 1.9245008973);
\draw [color=black] (-2.7777777778, 0.9622504486) -- (-2.2222222222, 1.9245008973);
\draw [color=black] (2.2222222222, 1.9245008973) -- (1.1111111111, 1.9245008973);
\draw [color=black] (2.2222222222, 1.9245008973) -- (1.6666666667, 2.8867513459);
\draw [color=black] (1.1111111111, 1.9245008973) -- (0.0000000000, 1.9245008973);
\draw [color=black] (1.1111111111, 1.9245008973) -- (1.6666666667, 2.8867513459);
\draw [color=black] (1.1111111111, 1.9245008973) -- (0.5555555556, 2.8867513459);
\draw [color=black] (0.0000000000, 1.9245008973) -- (-1.1111111111, 1.9245008973);
\draw [color=black] (0.0000000000, 1.9245008973) -- (0.5555555556, 2.8867513459);
\draw [color=black] (0.0000000000, 1.9245008973) -- (-0.5555555556, 2.8867513459);
\draw [color=black] (-1.1111111111, 1.9245008973) -- (-2.2222222222, 1.9245008973);
\draw [color=black] (-1.1111111111, 1.9245008973) -- (-0.5555555556, 2.8867513459);
\draw [color=black] (-1.1111111111, 1.9245008973) -- (-1.6666666667, 2.8867513459);
\draw [color=black] (-2.2222222222, 1.9245008973) -- (-1.6666666667, 2.8867513459);
\draw [color=black] (1.6666666667, 2.8867513459) -- (0.5555555556, 2.8867513459);
\draw [color=black] (1.6666666667, 2.8867513459) -- (1.1111111111, 3.8490017946);
\draw [color=black] (0.5555555556, 2.8867513459) -- (-0.5555555556, 2.8867513459);
\draw [color=black] (0.5555555556, 2.8867513459) -- (1.1111111111, 3.8490017946);
\draw [color=black] (0.5555555556, 2.8867513459) -- (0.0000000000, 3.8490017946);
\draw [color=black] (-0.5555555556, 2.8867513459) -- (-1.6666666667, 2.8867513459);
\draw [color=black] (-0.5555555556, 2.8867513459) -- (0.0000000000, 3.8490017946);
\draw [color=black] (-0.5555555556, 2.8867513459) -- (-1.1111111111, 3.8490017946);
\draw [color=black] (-1.6666666667, 2.8867513459) -- (-1.1111111111, 3.8490017946);
\draw [dashed] (1.1111111111, 3.8490017946) -- (0.0000000000, 3.8490017946);
\draw [color=black] (1.1111111111, 3.8490017946) -- (0.5555555556, 4.8112522432);
\draw [dashed] (0.0000000000, 3.8490017946) -- (-1.1111111111, 3.8490017946);
\draw [color=black] (0.0000000000, 3.8490017946) -- (0.5555555556, 4.8112522432);
\draw [color=black] (0.0000000000, 3.8490017946) -- (-0.5555555556, 4.8112522432);
\draw [color=black] (-1.1111111111, 3.8490017946) -- (-0.5555555556, 4.8112522432);
\draw [dashed] (0.5555555556, 4.8112522432) -- (-0.5555555556, 4.8112522432);
\draw [color=black] (0.5555555556, 4.8112522432) -- (0.0000000000, 5.7735026919);
\draw [color=black] (-0.5555555556, 4.8112522432) -- (0.0000000000, 5.7735026919);
\draw [color=black] (0, 5.7735026919) -- (0.0, 7.50555349947);
\draw [color=black] (-5.0, -2.88675134595) -- (-6.5, -3.75277674973);
\draw [color=black] (5.0, -2.88675134595) -- (6.5, -3.75277674973);
\node at (-0.5,5.94670777265) {\large $e^1$};
\node at (-4.9,-3.40636658822) {\large $e^2$};
\node at (5.4,-2.54034118443) {\large $e^3$};
\node at (4.44444444444,-2.5660011964) {$3\rho$};
\node at (3.33333333333,-2.5660011964) {$2\rho$};
\node at (3.88888888889,-2.24525104685) {$3\rho$};
\node at (2.22222222222,-2.5660011964) {$\rho$};
\node at (2.77777777778,-2.24525104685) {$2\rho$};
\node at (1.11111111111,-2.5660011964) {$\rho$};
\node at (1.66666666667,-2.24525104685) {$2\rho$};
\node at (0.0,-2.5660011964) {$\rho$};
\node at (0.555555555556,-2.24525104685) {$2\rho$};
\node at (-1.11111111111,-2.5660011964) {$\rho$};
\node at (-0.555555555556,-2.24525104685) {$2\rho$};
\node at (-2.22222222222,-2.5660011964) {$\rho$};
\node at (-1.66666666667,-2.24525104685) {$2\rho$};
\node at (-3.33333333333,-2.5660011964) {$2\rho$};
\node at (-2.77777777778,-2.24525104685) {$2\rho$};
\node at (-4.44444444444,-2.5660011964) {$3\rho$};
\node at (-3.88888888889,-2.24525104685) {$3\rho$};
\node at (3.88888888889,-1.60375074775) {$4\rho$};
\node at (2.77777777778,-1.60375074775) {$3\rho$};
\node at (3.33333333333,-1.2830005982) {$4\rho$};
\node at (1.66666666667,-1.60375074775) {$3\rho$};
\node at (2.22222222222,-1.2830005982) {$4\rho$};
\node at (0.555555555556,-1.60375074775) {$3\rho$};
\node at (1.11111111111,-1.2830005982) {$4\rho$};
\node at (-0.555555555556,-1.60375074775) {$3\rho$};
\node at (0.0,-1.2830005982) {$4\rho$};
\node at (-1.66666666667,-1.60375074775) {$3\rho$};
\node at (-1.11111111111,-1.2830005982) {$4\rho$};
\node at (-2.77777777778,-1.60375074775) {$3\rho$};
\node at (-2.22222222222,-1.2830005982) {$4\rho$};
\node at (-3.88888888889,-1.60375074775) {$4\rho$};
\node at (-3.33333333333,-1.2830005982) {$4\rho$};
\node at (3.33333333333,-0.6415002991) {$5\rho$};
\node at (2.22222222222,-0.6415002991) {$5\rho$};
\node at (2.77777777778,-0.32075014955) {$6\rho$};
\node at (1.11111111111,-0.6415002991) {$5\rho$};
\node at (1.66666666667,-0.32075014955) {$6\rho$};
\node at (0.0,-0.6415002991) {$5\rho$};
\node at (0.555555555556,-0.32075014955) {$6\rho$};
\node at (-1.11111111111,-0.6415002991) {$5\rho$};
\node at (-0.555555555556,-0.32075014955) {$6\rho$};
\node at (-2.22222222222,-0.6415002991) {$5\rho$};
\node at (-1.66666666667,-0.32075014955) {$6\rho$};
\node at (-3.33333333333,-0.6415002991) {$5\rho$};
\node at (-2.77777777778,-0.32075014955) {$6\rho$};
\node at (2.77777777778,0.32075014955) {$7\rho$};
\node at (1.66666666667,0.32075014955) {$7\rho$};
\node at (2.22222222222,0.6415002991) {$8\rho$};
\node at (0.555555555556,0.32075014955) {$7\rho$};
\node at (1.11111111111,0.6415002991) {$8\rho$};
\node at (-0.555555555556,0.32075014955) {$7\rho$};
\node at (0.0,0.6415002991) {$8\rho$};
\node at (-1.66666666667,0.32075014955) {$7\rho$};
\node at (-1.11111111111,0.6415002991) {$8\rho$};
\node at (-2.77777777778,0.32075014955) {$7\rho$};
\node at (-2.22222222222,0.6415002991) {$8\rho$};
\node at (2.22222222222,1.2830005982) {$9\rho$};
\node at (1.11111111111,1.2830005982) {$9\rho$};
\node at (1.66666666667,1.60375074775) {$10\rho$};
\node at (0.0,1.2830005982) {$9\rho$};
\node at (0.555555555556,1.60375074775) {$10\rho$};
\node at (-1.11111111111,1.2830005982) {$9\rho$};
\node at (-0.555555555556,1.60375074775) {$10\rho$};
\node at (-2.22222222222,1.2830005982) {$9\rho$};
\node at (-1.66666666667,1.60375074775) {$10\rho$};
\node at (1.66666666667,2.24525104685) {$11\rho$};
\node at (0.555555555556,2.24525104685) {$11\rho$};
\node at (1.11111111111,2.5660011964) {$12\rho$};
\node at (-0.555555555556,2.24525104685) {$11\rho$};
\node at (0.0,2.5660011964) {$12\rho$};
\node at (-1.66666666667,2.24525104685) {$11\rho$};
\node at (-1.11111111111,2.5660011964) {$12\rho$};
\node at (1.11111111111,3.2075014955) {$12\rho$};
\node at (0.0,3.2075014955) {$12\rho$};
\node at (0.555555555556,3.52825164505) {$12\rho$};
\node at (-1.11111111111,3.2075014955) {$12\rho$};
\node at (-0.555555555556,3.52825164505) {$12\rho$};
\node at (0.555555555556,4.16975194415) {$12\rho$};
\node at (-0.555555555556,4.16975194415) {$12\rho$};
\node at (0.0,4.4905020937) {$12\rho$};
\node at (0.0,5.1320023928) {$12\rho$};
\node at (-0.0,-3.46410161514) {0};
\end{tikzpicture}

%% file: fk-gap.tikz
\begin{tikzpicture}
\draw [color=black] (5.0000000000, -2.8867513459) -- (0.0000000000, -2.8867513459);
\draw [color=black] (5.0000000000, -2.8867513459) -- (2.5000000000, 1.4433756730);
\draw [color=black] (0.0000000000, -2.8867513459) -- (-5.0000000000, -2.8867513459);
\draw [color=black] (0.0000000000, -2.8867513459) -- (2.5000000000, 1.4433756730);
\draw [color=black] (0.0000000000, -2.8867513459) -- (-2.5000000000, 1.4433756730);
\draw [color=black] (-5.0000000000, -2.8867513459) -- (-2.5000000000, 1.4433756730);
\draw [color=black] (2.5000000000, 1.4433756730) -- (-2.5000000000, 1.4433756730);
\draw [color=black] (2.5000000000, 1.4433756730) -- (0.0000000000, 5.7735026919);
\draw [color=black] (-2.5000000000, 1.4433756730) -- (0.0000000000, 5.7735026919);
\node at (0.0,6.81273317644) {\Huge $e^1$};
\node at (-5.9,-3.40636658822) {\Huge $e^2$};
\node at (5.9,-3.40636658822) {\Huge $e^3$};
\node at (2.5,-3.92598183049) {\Huge 1/6};
\node at (4.65,-0.202072594216) {\Huge 1/6};
\node at (-2.5,-3.92598183049) {\Huge 1/6};
\node at (2.15,-1.24130307876) {\Huge 1/4};
\node at (-2.15,-1.24130307876) {\Huge 1/4};
\node at (-4.65,-0.202072594216) {\Huge 1/6};
\node at (0.0,2.48260615752) {\Huge 1/4};
\node at (2.15,4.12805442471) {\Huge 1/6};
\node at (-2.15,4.12805442471) {\Huge 1/6};
\end{tikzpicture}

%% file: fk-corner.tikz
\begin{tikzpicture}
\draw [color=black] (5.0000000000, -2.8867513459) -- (0.0000000000, -2.8867513459);
\draw [color=black] (5.0000000000, -2.8867513459) -- (2.5000000000, 1.4433756730);
\draw [color=black] (0.0000000000, -2.8867513459) -- (-5.0000000000, -2.8867513459);
\draw [color=black] (0.0000000000, -2.8867513459) -- (2.5000000000, 1.4433756730);
\draw [color=black] (0.0000000000, -2.8867513459) -- (-2.5000000000, 1.4433756730);
\draw [color=black] (-5.0000000000, -2.8867513459) -- (-2.5000000000, 1.4433756730);
\draw [color=black] (2.5000000000, 1.4433756730) -- (-2.5000000000, 1.4433756730);
\draw [color=black] (2.5000000000, 1.4433756730) -- (0.0000000000, 5.7735026919);
\draw [color=black] (-2.5000000000, 1.4433756730) -- (0.0000000000, 5.7735026919);
\draw [color=black] (0, 5.7735026919) -- (0.0, 7.50555349947);
\draw [color=black] (-5.0, -2.88675134595) -- (-6.5, -3.75277674973);
\draw [color=black] (5.0, -2.88675134595) -- (6.5, -3.75277674973);
\draw[fill] (0.0000000000, -5.7735026919) circle [radius=0.2000000000];
\node at (0.0,-6.81273317644) {\Huge $O_1$};
\draw[fill] (5.0000000000, 2.8867513459) circle [radius=0.2000000000];
\node at (5.9,3.40636658822) {\Huge $O_2$};
\draw[fill] (-5.0000000000, 2.8867513459) circle [radius=0.2000000000];
\node at (-5.9,3.40636658822) {\Huge $O_3$};
\draw[fill] (2.5000000000, -1.4433756730) circle [radius=0.2000000000];
\draw[fill] (-2.5000000000, -1.4433756730) circle [radius=0.2000000000];
\draw[fill] (0.0000000000, -0.0000000000) circle [radius=0.2000000000];
\draw[fill] (0.0000000000, 2.8867513459) circle [radius=0.2000000000];
\draw [color=black, line width=4] (0.0000000000, 2.8867513459) -- (5.0000000000, 2.8867513459);
\draw [color=black, line width=4] (0.0000000000, 2.8867513459) -- (-5.0000000000, 2.8867513459);
\draw [color=black, line width=4] (-2.5000000000, -1.4433756730) -- (0.0000000000, -5.7735026919);
\draw [color=black, line width=4] (-2.5000000000, -1.4433756730) -- (-5.0000000000, 2.8867513459);
\draw [color=black, line width=4] (2.5000000000, -1.4433756730) -- (0.0000000000, -5.7735026919);
\draw [color=black, line width=4] (2.5000000000, -1.4433756730) -- (5.0000000000, 2.8867513459);
\end{tikzpicture}

%% file: fk-ball-1.tikz
\begin{tikzpicture}
\draw [color=black] (5.0000000000, -2.8867513459) -- (0.0000000000, -2.8867513459);
\draw [color=black] (5.0000000000, -2.8867513459) -- (2.5000000000, 1.4433756730);
\draw [color=black] (0.0000000000, -2.8867513459) -- (-5.0000000000, -2.8867513459);
\draw [color=black] (0.0000000000, -2.8867513459) -- (2.5000000000, 1.4433756730);
\draw [color=black] (0.0000000000, -2.8867513459) -- (-2.5000000000, 1.4433756730);
\draw [color=black] (-5.0000000000, -2.8867513459) -- (-2.5000000000, 1.4433756730);
\draw [color=black] (2.5000000000, 1.4433756730) -- (-2.5000000000, 1.4433756730);
\draw [color=black] (2.5000000000, 1.4433756730) -- (0.0000000000, 5.7735026919);
\draw [color=black] (-2.5000000000, 1.4433756730) -- (0.0000000000, 5.7735026919);
\draw [color=black] (0, 5.7735026919) -- (0.0, 7.50555349947);
\draw [color=black] (-5.0, -2.88675134595) -- (-6.5, -3.75277674973);
\draw [color=black] (5.0, -2.88675134595) -- (6.5, -3.75277674973);
\draw[fill] (0.0000000000, -5.7735026919) circle [radius=0.2000000000];
\node at (0.0,-6.81273317644) {\Huge $O_1$};
\draw[fill] (5.0000000000, 2.8867513459) circle [radius=0.2000000000];
\node at (5.9,3.40636658822) {\Huge $O_2$};
\draw[fill] (-5.0000000000, 2.8867513459) circle [radius=0.2000000000];
\node at (-5.9,3.40636658822) {\Huge $O_3$};
\draw[fill] (2.5000000000, -1.4433756730) circle [radius=0.2000000000];
\draw[fill] (-2.5000000000, -1.4433756730) circle [radius=0.2000000000];
\draw[fill] (0.0000000000, -0.0000000000) circle [radius=0.2000000000];
\draw[fill] (0.0000000000, 2.8867513459) circle [radius=0.2000000000];
\draw [color=black, line width=4] (0.0000000000, 2.8867513459) -- (0.0000000000, -0.0000000000);
\draw [color=black, line width=4] (0.0000000000, 2.8867513459) -- (5.0000000000, 2.8867513459);
\draw [color=black, line width=4] (-2.5000000000, -1.4433756730) -- (0.0000000000, -0.0000000000);
\draw [color=black, line width=4] (-2.5000000000, -1.4433756730) -- (-5.0000000000, 2.8867513459);
\draw [color=black, line width=4] (2.5000000000, -1.4433756730) -- (0.0000000000, -0.0000000000);
\draw [color=black, line width=4] (2.5000000000, -1.4433756730) -- (0.0000000000, -5.7735026919);
\end{tikzpicture}

%% file: fk-ball-2.tikz
\begin{tikzpicture}
\draw [color=black] (5.0000000000, -2.8867513459) -- (0.0000000000, -2.8867513459);
\draw [color=black] (5.0000000000, -2.8867513459) -- (2.5000000000, 1.4433756730);
\draw [color=black] (0.0000000000, -2.8867513459) -- (-5.0000000000, -2.8867513459);
\draw [color=black] (0.0000000000, -2.8867513459) -- (2.5000000000, 1.4433756730);
\draw [color=black] (0.0000000000, -2.8867513459) -- (-2.5000000000, 1.4433756730);
\draw [color=black] (-5.0000000000, -2.8867513459) -- (-2.5000000000, 1.4433756730);
\draw [color=black] (2.5000000000, 1.4433756730) -- (-2.5000000000, 1.4433756730);
\draw [color=black] (2.5000000000, 1.4433756730) -- (0.0000000000, 5.7735026919);
\draw [color=black] (-2.5000000000, 1.4433756730) -- (0.0000000000, 5.7735026919);
\draw [color=black] (0, 5.7735026919) -- (0.0, 7.50555349947);
\draw [color=black] (-5.0, -2.88675134595) -- (-6.5, -3.75277674973);
\draw [color=black] (5.0, -2.88675134595) -- (6.5, -3.75277674973);
\draw[fill] (0.0000000000, -5.7735026919) circle [radius=0.2000000000];
\node at (0.0,-6.81273317644) {\Huge $O_1$};
\draw[fill] (5.0000000000, 2.8867513459) circle [radius=0.2000000000];
\node at (5.9,3.40636658822) {\Huge $O_2$};
\draw[fill] (-5.0000000000, 2.8867513459) circle [radius=0.2000000000];
\node at (-5.9,3.40636658822) {\Huge $O_3$};
\draw[fill] (2.5000000000, -1.4433756730) circle [radius=0.2000000000];
\draw[fill] (-2.5000000000, -1.4433756730) circle [radius=0.2000000000];
\draw[fill] (0.0000000000, -0.0000000000) circle [radius=0.2000000000];
\draw[fill] (0.0000000000, 2.8867513459) circle [radius=0.2000000000];
\draw [color=black, line width=4] (0.0000000000, 2.8867513459) -- (5.0000000000, 2.8867513459);
\draw [color=black, line width=4] (0.0000000000, 2.8867513459) -- (-5.0000000000, 2.8867513459);
\draw [color=black, line width=4] (0.0000000000, 2.8867513459) -- (0.0000000000, -0.0000000000);
\draw [color=black, line width=4] (-2.5000000000, -1.4433756730) -- (0.0000000000, -0.0000000000);
\draw [color=black, line width=4] (-2.5000000000, -1.4433756730) -- (0.0000000000, -5.7735026919);
\end{tikzpicture}